\documentclass[pra, aps, twocolumn, superscriptaddress, nofootinbib]{revtex4-2}
\usepackage{hyperref}
\usepackage{amsmath}
\usepackage{amsfonts}
\usepackage{amsthm}
\usepackage{amssymb}
\usepackage{graphicx}
\usepackage{caption}
\usepackage{subcaption}
\usepackage{float}
\restylefloat{table}
\usepackage{siunitx}
\usepackage{bbm}
\usepackage{bm}
\usepackage{comment}
\usepackage{xcolor}
\usepackage[shortlabels]{enumitem}
\usepackage{algorithm}
\usepackage{algorithmic}
\usepackage{makecell}
\usepackage{titletoc}
\usepackage{outlines}
\usepackage[shortlabels]{enumitem}
\usepackage{footnote}

\newtheorem{theorem}{Theorem}
\newtheorem{lemma}[theorem]{Lemma}
\newtheorem{proposition}[theorem]{Proposition}
\newtheorem{definition}{Definition}

\newtheorem{remark}{Remark}
\newtheorem{task}{Task}

\DeclareMathOperator{\Tr}{Tr}

\DeclareMathOperator{\re}{Re}
\DeclareMathOperator{\poly}{poly}
\DeclareMathOperator{\sign}{sign}

\DeclareMathOperator{\op}{op}

\DeclareMathOperator{\Sp}{Sp}

\usepackage{mathtools}
\usepackage{soul}

\definecolor{mypink}{RGB}{255, 0, 213}
\definecolor{mypurple2}{RGB}{170,0,255}
\definecolor{myred}{RGB}{255, 0, 85}

\theoremstyle{definition}

\newtheorem{fact}{Fact}

\makeatletter
\DeclareFontFamily{OMX}{MnSymbolE}{}
\DeclareSymbolFont{MnLargeSymbols}{OMX}{MnSymbolE}{m}{n}
\SetSymbolFont{MnLargeSymbols}{bold}{OMX}{MnSymbolE}{b}{n}
\DeclareFontShape{OMX}{MnSymbolE}{m}{n}{
    <-6>  MnSymbolE5
   <6-7>  MnSymbolE6
   <7-8>  MnSymbolE7
   <8-9>  MnSymbolE8
   <9-10> MnSymbolE9
  <10-12> MnSymbolE10
  <12->   MnSymbolE12
}{}
\DeclareFontShape{OMX}{MnSymbolE}{b}{n}{
    <-6>  MnSymbolE-Bold5
   <6-7>  MnSymbolE-Bold6
   <7-8>  MnSymbolE-Bold7
   <8-9>  MnSymbolE-Bold8
   <9-10> MnSymbolE-Bold9
  <10-12> MnSymbolE-Bold10
  <12->   MnSymbolE-Bold12
}{}

\let\llangle\@undefined
\let\rrangle\@undefined
\DeclareMathDelimiter{\llangle}{\mathopen}%
                     {MnLargeSymbols}{'164}{MnLargeSymbols}{'164}
\DeclareMathDelimiter{\rrangle}{\mathclose}%
                     {MnLargeSymbols}{'171}{MnLargeSymbols}{'171}
\makeatother

\newcommand{\ket}[1]{| #1 \rangle}
\newcommand{\bra}[1]{\langle #1 |}

\newcommand{\kett}[1]{| #1 \rrangle}
\newcommand{\braa}[1]{\llangle #1 |}
\newcommand{\braakett}[2]{\llangle #1|#2 \rrangle}

\newcommand{\tr}{\mathrm{tr}}
\newcommand{\eps}{\varepsilon}

\newcommand{\E}{\mathop{{}\mathbb{E}}}
\newcommand{\Var}{\mathrm{Var}}
\newcommand{\Zd}{\mathbb{Z}_d}

\newcommand{\Cld}{\mathrm{Cl}_d}

\allowdisplaybreaks

\begin{document}
\title{Exponential learning advantages with conjugate states and minimal quantum memory}
\author{Robbie King}
\affiliation{Google Quantum AI, Venice, CA 90291, USA}
\affiliation{Department of Computing and Mathematical Sciences, Caltech, Pasadena, CA 91125, USA}
\author{Kianna Wan}
\affiliation{Google Quantum AI, Venice, CA 90291, USA}
\affiliation{Stanford Institute for Theoretical Physics, Stanford University, Stanford, CA 94305, USA}
\author{Jarrod R. McClean}
\affiliation{Google Quantum AI, Venice, CA 90291, USA}
\begin{abstract}
The ability of quantum computers to directly manipulate and analyze quantum states stored in quantum memory allows them to learn about aspects of our physical world that would otherwise be invisible given a modest number of measurements. Here we investigate a new learning resource which could be available to quantum computers in the future -- measurements on the unknown state accompanied by its complex conjugate $\rho \otimes \rho^\ast$. For a certain shadow tomography task, we surprisingly find that measurements on only copies of $\rho \otimes \rho^\ast$ can be exponentially more powerful than measurements on $\rho^{\otimes K}$, even for large $K$. This expands the class of provable exponential advantages using only a constant overhead quantum memory, or minimal quantum memory, and we provide a number of examples where the state $\rho^\ast$ is naturally available in both computational and physical applications. In addition, we precisely quantify the power of classical shadows on single copies under a generalized Clifford ensemble and give a class of quantities that can be efficiently learned. The learning task we study in both the single copy and quantum memory settings is physically natural and corresponds to real-space observables with a limit of bosonic modes, where it achieves an exponential improvement in detecting certain signals under a noisy background. We quantify a new and powerful resource in quantum learning, and we believe the advantage may find applications in improving quantum simulation, learning from quantum sensors, and uncovering new physical phenomena.
\end{abstract}
\maketitle

\section{Introduction}

\begin{figure*}
\centering
\includegraphics[width=11cm]{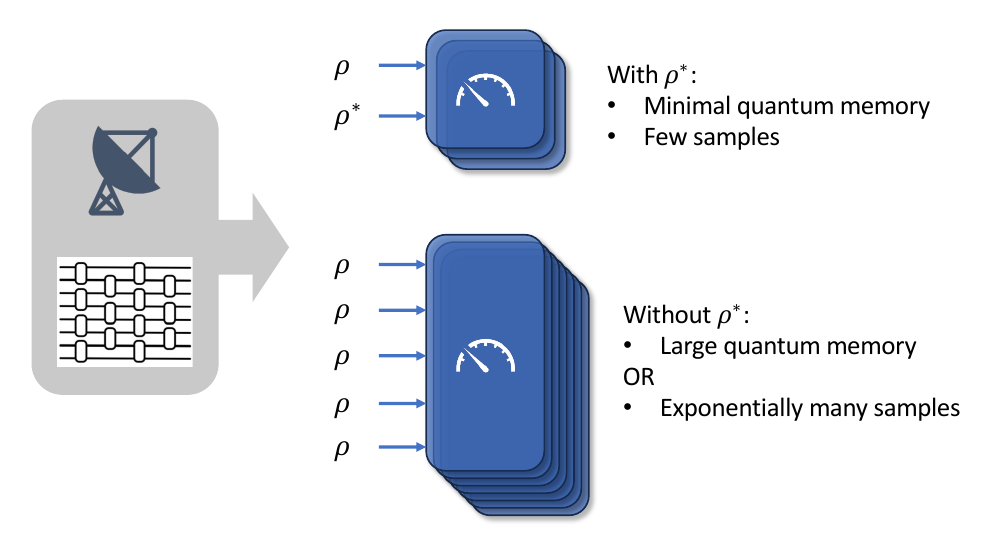}
\caption{A cartoon of the techniques and novel resources in this work.  Here we quantify the advantage endowed by minimal quantum memories containing $\rho \otimes \rho^\ast$ in learning about natural properties of quantum states coming from quantum sensors or digital quantum simulations.  (Left) The resources $\rho$ and $\rho^*$ are available from both computational sources such as a known quantum circuit or natural sources like certain quantum sensor setups.  (Right) This resource provides an exponential advantage in queries and computation for some learning tasks when only a minimal quantum memory (or constant number of copies) is available, as in most devices in the foreseeable future.  Without $\rho^\ast$, one needs either a large quantum memory or exponentially many samples.  Several applications of this technique are introduced and we believe this will motivate the development of further applications of minimal quantum memories.}
\label{fig:Overview}
\end{figure*}

Efficiently extracting information from quantum states is a central task in quantum information science. It is crucial in physical experiments and will be critical in simulations run on future quantum computers. Often learning the details of an entire quantum state is not required, but rather we would like to extract the expectation values of a set of interesting observables.  

Naively, measurements of many properties of quantum states are constrained by the uncertainty principle for non-commuting operators.   Additionally, tomographic techniques for exactly learning a quantum state up to a stringent standard like worst case observable error (or trace distance) are known to scale exponentially in the number of qubits or polynomially in the size of the Hilbert space~\cite{o2016efficient,anshu2023survey}.  Surprisingly however, the development of shadow tomography techniques demonstrated that one can learn a set of expectation values of even non-commuting observables with high probability with a shockingly small number of samples, scaling only polylogarithmically in the number of observables~\cite{aaronson2018shadow, aaronson2018online, aaronson2019gentle, buadescu2021improved}. While powerful, the general schemes suffer from two large caveats: they are computationally inefficient, and they require immense quantum memories, sometimes millions of times the size of the original state, to enable huge entangled measurements.
Classical shadows \cite{huang2020predicting} were developed to circumvent both of these limitations – it is computationally efficient, and requires only single-copy measurements for a wide class of useful observables. However, classical shadows place limitations on the sets of observables that are available, for example some schemes are only able to learn observables which are either local or low rank.

It is now known that many of the limitations of classical shadows performed only on single copies at a time are fundamental. For general quantum states, certain collections of observables can only be learned with a logarithmic number of samples by exploiting entangled measurements across multiple copies of a state~\cite{chen2022exponential,aharonov2022quantum,huang2022quantum}. This was shown to be true even for some of the simplest large sets of observables, namely Pauli operators on $n$ qubits.
Phrased a different way, the ability to make entangled measurements on copies of a quantum state can grant exponentially more power in learning tasks. Since such schemes require a quantum memory to store simultaneous copies of an unknown state and this type of advantage cannot be overcome even by an arbitrary amount of classical computation when samples are limited, this constitutes a promising future application of quantum computers. It was demonstrated experimentally in Ref.~\cite{huang2022quantum} that the advantage persists even for small numbers of qubits and in the presence of noise.  Importantly, Ref.~\cite{huang2022quantum} showed that exponential advantages are available using only 2 copies at a time of the state in quantum memory, or an example of a minimal quantum memory for which the number of copies required is independent of the learning task.  This demonstrated the existence of learning tasks for which extremely limited quantum resources could provide huge advantages, in contrast to the need of general shadow tomography to have memories that could be millions of times larger than the system of interest even when tasked with estimating observables to a precision of only $10^{-3}$.

In this work, we explore a novel resource for learning which is able to grant exponential advantages using only a minimal quantum memory (space for only 2 copies) – the ability to make joint measurements on an unknown quantum state with its complex conjugate, denoted $\rho \otimes \rho^*$. We give a learning task that can be achieved with low sample complexity using measurements on $\rho \otimes \rho^*$.  In contrast, without access to $\rho^*$, the same learning task requires exponentially more measurements, unless the size of the quantum memory is allowed to expand to practically unrealistic sizes as a function of the learning task. While it is known that quantum memory access to $\rho$ and $\rho^\ast$ provide an efficient means to sample in the Heisenberg-Weyl basis, accomplishing a similar learning task to one discussed here~\cite{montanaro2017learning,gross2021schur,grewal2023improved}, to the knowledge of the authors a quantum memory size lower bound without access to $\rho^*$ has not been shown previously.  Although the operation of complex conjugation to an unknown state $\rho^*$ is not physical and cannot be implemented efficiently, in Appendix ~\ref{app:natural_states} we highlight a wide class of cases where the complex conjugate state is available and discuss potential limitations.  We prove a lower bound showing that that copies of $\rho^*$ and $\rho$ without quantum memory are also insufficient for the learning task.  

Previous schemes using minimal quantum memory have only been able to resolve the magnitude of the observable~\cite{huang2021information,huang2022quantum} requiring a quantum memory scaling polynomially in precision to determine the sign.  The question of whether it is possible to determine both the magnitude and sign using a minimal quantum memory is resolved in Ref.~\cite{inprep} and we introduce a specialization here for the specific class of operators we are interested in.

Our exponential separation holds for a natural and physically motivated set of observables. The setting is a $d$-dimensional Hilbert space that discretizes position and momentum space with a natural limit of a continuous bosonic mode, and the operators we learn if we take the infinite dimensional limit are the bosonic displacement operators \cite{brady2023advances}. These operators more naturally correspond to real-space arrays of quantum sensors.  Recent developments in reconfigurable atom arrays~\cite{bluvstein2023logical} may provide a fruitful test bed for applications in a sensing context for example, especially given their wide bandwidth and sensitivity in other applications~\cite{osterwalder1999using,sedlacek2012microwave,holloway2017atom,wade2017real,simons2021rydberg}.
In addition to the learning algorithm using $\rho \otimes \rho^\ast$, we develop a version of classical shadows tailored to the $d$-dimensional bosonic setting.
It uses a uniform distribution over the generalized $d$-dimensional Clifford group to make single-copy measurements with good predictive power. Despite the more limited power of single copies, we identify a wide class of quantities that are efficiently learnable.

The core technique of our learning algorithm using $\rho \otimes \rho^\ast$ relies on using extensions via tensor products to create commutativity among displacement operators with Pauli operators as a special case.  One may wonder if this technique can be applied to a broader class of operators than displacement operators, and we partly resolve this question in the negative by showing commutativity via tensor extension naturally corresponds to a definition of the Heisenberg-Weyl group with a relationship to uniqueness via the Stone-von Neumann theorem.

Finally, we highlight two potential applications of our work.
The first is to learn efficiently from the output of quantum simulations run on quantum computers. Given a quantum algorithm and the explicit quantum circuit that prepares a state of interest $\rho$, one can easily complex-conjugate the gates of the quantum circuit to prepare the complex conjugate state $\rho^\ast$. Then incurring a factor of two overhead in the size of the quantum computer, one can perform entangled measurements on $\rho \otimes \rho^\ast$ to realize this exponential advantage.
The second realm of applications lies in learning from quantum data collected by quantum sensors, such arrays that have been proposed for long baseline interferometry \cite{gottesman2012longer,degen2017quantum,khabiboulline2019optical,bland2021quantum}.  The operators we consider are naturally related to regular arrays of quantum sensors arranged in real-space, and we provide examples where $\rho^*$ can be obtained in Appendix ~\ref{app:sensor_arrays}.  In this setting, our technique represents a specialized form of mixedness testing~\cite{chen2021hierarchy}, where it allows exponentially improved signal to noise ratios in determining if a signal is present in high background noise settings for certain classes of states. 
Potential applications and natural sources of $\rho^*$ are discussed in Appendix ~\ref{app:natural_states}.  We hope these results motivate the discovery of additional applications of minimal quantum memories equipped with complex conjugate resources.

\section{Necessity and power of conjugate quantum states}

Here we provide some background on the operators and states that we aim to learn, and then state our main theorems showing the exponential power of access to the complex conjugate resource in a minimal quantum memory.  Shifts in discrete position and momentum space are given by the $d$-dimensional clock and shift operators, $Z$ and $X$, which are generalizations of the qubit Pauli operators to $d$ dimensions.
\begin{align}
X &=
\begin{pmatrix}
0 & 0 & \dots & 0 & 1 \\
1 & 0 & \dots & 0 & 0 \\
\vdots & \vdots & \ddots & \vdots & \vdots \\
0 & 0 & \dots & 1 & 0 \\
\end{pmatrix}, \\
Z &=
\begin{pmatrix}
1 & 0 & \dots & 0 \\
0 & \omega & \dots & 0 \\
\vdots & \vdots & \ddots & \vdots \\
0 & 0 & \dots & \omega^{d-1} \\
\end{pmatrix},
\end{align}
where $\omega \coloneqq e^{i 2\pi / d}$. The operators map naturally to a 1D discrete line in real space, and can be interchanged via the $d$-dimensional quantum Fourier transform. 
Combined position and momentum shifts may be lumped together into displacement operators $D_{q,p}$, defined briefly as
\begin{equation}
D_{q,p} = e^{i \pi q p / d} X^q Z^p
\end{equation}
and in more detail in Appendix ~\ref{sec:background}.

\begin{definition}
The \emph{displacement amplitudes} of a $d$-dimensional state $\rho$ are
\begin{equation}
y_{q,p} = \Tr\left(D_{q,p} \rho\right)
\end{equation}
\end{definition}

The displacement operators may be used to form a basis for quantum states, and there are $d^2$ displacement amplitudes. The central task we consider in this paper is to estimate all $y_{q,p}$ to precision $\varepsilon$ given copies of an unknown quantum state $\rho$.

\begin{task} \label{task:main}
\emph{(Informal)}
Given access to a quantum state $\rho$, estimate all the displacement amplitudes $\{y_{q,p}\}$ to precision $\varepsilon$ with high probability.
\end{task}

Our first result is a sample complexity lower bound showing the minimal size of a conventional quantum memory required to efficiently perform this task without access to the resource $\rho^*$.  We assume we can measure $K$ copies $\rho^{\otimes K}$ at a time, possibly in entangled bases, but allow no access to $\rho^*$.

\begin{theorem} \label{thm:displacement_lower_bd_2_informal}
Let $d$ be the dimension of the Hilbert space. Assume $d$ is prime. Any protocol which learns the magnitudes of all displacement amplitudes to precision $\varepsilon$ with probability $2/3$ by measuring copies of $\rho^{\otimes K}$ for $K \leq 1 / (12 \varepsilon)$ requires $\Omega(\sqrt{d} / (K^2 \varepsilon^2))$ measurements.
\end{theorem}
The full proof of this is given in Appendix~\ref{sec:boson_lower_bd_2}.
Using a quantum memory with $\rho^{\otimes K}$, performing the learning task whilst consuming a number of copies scaling only as $\text{polylog}(d)$ is impossible, even if $K$ grows as large as $1 / (12 \varepsilon)$.  Hence it is impossible to efficiently perform the task with a minimal quantum memory.

This negative result may lead one to conclude that quantum memories are not as powerful as one might hope for physical learning tasks, but there is a resolution to this challenge which reveals an interesting subtlety in the power of quantum computing in analyzing quantum data.
While a quantum memory containing $K$ states $\rho^{\otimes K}$ is insufficient, measurements on the state $\rho \otimes \rho^*$ are able to learn the displacement amplitudes of $\rho$ up to a sign. Not only does the learning algorithm using $\rho \otimes \rho^*$ have logarithmic sample complexity, but the algorithm is very simple and computationally efficient.
Note that we use the term ``up to a sign'' to convey that because these are unitary but not always Hermitian operators, they are complex valued and we are able to learn some, but not all, phase information about that value with this procedure.
With the detailed proof and algorithm given in Appendix \ref{sec:disp_alg}, we show
\begin{theorem} \label{thm:displacement_informal}
There is an algorithm which can learn all displacement amplitudes up to a possible minus sign, with precision $\varepsilon$, using $\mathcal{O}(\log{d} / \varepsilon^4)$ samples. The algorithm makes measurements only on copies of $\rho \otimes \rho^\ast$ contained in a minimal quantum memory. Moreover, the algorithm is computationally efficient.
\end{theorem}
When viewed together, these theorems highlight the exponential advantage of using $\rho^*$ as a resource in learning tasks.  Indeed as we show in Appendix~\ref{app:rho_star_hardness}, $\rho^*$ is quite a powerful resource in other contexts, and may not be available for totally general unknown states.  This naturally leads one to wonder whether the state $\rho^*$ has this power in the single copy setting, but indeed lower bounds rule this out and show that entangled measurements are also a necessary component of the learning algorithm in Theorem \ref{thm:displacement_informal}. Our following theorem articulates this more precisely.
\begin{theorem} \label{thm:displacement_lower_bd_1_informal}
Let $d$ be the dimension of the Hilbert space. Any single-copy protocol which learns the magnitudes of all displacement amplitudes to precision $\varepsilon$ with probability $2/3$ requires a number of copies scaling as $\Omega(d / \varepsilon^2)$. This holds even if the protocol has access to single-copy measurements of both $\rho$ and $\rho^\ast$.
\end{theorem}
With the detailed proof given in Appendix \ref{sec:boson_lower_bd_1}, this result reiterates the conclusion that entangled measurements using quantum memories have dramatically more power than those that can process only a single copy at a time, and indeed that $\rho^*$ inside a minimal quantum memory is a powerful and novel resource.

The algorithm to learn the displacement amplitudes up to a sign uses a technique where one attaches a second system and constructs a mutually commuting set of operators on the joint system which contain information about the original non-commuting operators on the single system. This trick was used previously to develop a shadow tomography algorithm for Pauli operators in \cite{huang2021information}.
Given the success of this technique, it is natural to wonder what other types of measurement and systems this could be applied to more generally.
In Appendix \ref{sec:commutation}, we give evidence that the ability to use this trick is unique to displacement operators which arise from representations of the Heisenberg groups. The qubit Pauli operators are a special case, and our work provides arguably the maximal generalization. We can show the following theorem using an idea reminiscent of the Stone-von Neumann theorem.
\begin{theorem} \label{thm:stone_von_neumann}
Let $U,V$ be unitaries of finite order $d$ on some Hilbert space $\mathcal{H}$. Suppose $U$ and $V$ do \emph{not} commute, but we can attach a second Hilbert space $\mathcal{H}'$ with unitaries $\tilde{U},\tilde{V}$ such that $U \otimes \tilde{U}$ and $V \otimes \tilde{V}$ commute. Then there is some unitary transformation of $\mathcal{H}$ mapping $U,V$ to a direct sum of displacement operators.
\end{theorem}
Theorem \ref{thm:stone_von_neumann} suggests that to go beyond these classes of observables with efficient shadow tomography, we must search for new quantum learning primitives.

\section{Resolving the signs}

So far the tasks discussed have only pertained to learning expectation values up to a sign.  Previous work has shown that learning the magnitudes of Pauli operators was possible with a minimal quantum memory, but determining the signs required a large quantum memory with task dependent size~\cite{huang2021information} leaving the question of an efficient, minimal memory protocol open.
The signs of these operators clearly contain useful information, and hence one may ask if it is possible to efficiently measure the signs in an information theoretic and computationally efficient sense using only measurements on $\rho \otimes \rho^\ast$ as well.
This question is resolved by techniques developed in Ref.~\cite{inprep}; we introduce a specialization here for the specific class of operators we are interested in, with detailed algorithms in Appendix ~\ref{sec:determining_signs}.

\begin{theorem} \label{thm:determine_signs}
There is an algorithm which can learn all $d^2$ displacement amplitudes (including their sign) using $\mathcal{O}(\log{d} / \varepsilon^4)$ samples. The algorithm makes measurements only on copies of $\rho \otimes \rho^\ast$. The runtime of the algorithm is $\poly(d, \varepsilon^{-1})$.
\end{theorem}

The algorithm in Theorem \ref{thm:determine_signs} relies on two key ideas \cite{inprep}: one is to use a hypothesis state to `shift the origin' of a subsequent magnitude measurement, and the other is to use matrix multiplicative weights \cite{arora2007combinatorial, aaronson2018online} as a subroutine to efficiently determine the hypothesis state.

\section{Application 1: Quantum data from quantum computation}

As a first application, we consider cases where an explicit quantum circuit is known for a state we wish to study. 
When conducting a physical experiment, we are often preparing some natural quantum state and subsequently performing measurements.
In quantum simulation, we aim to design quantum algorithms that simulate Nature, so that the quantum algorithm prepares the physical quantum state of interest. One may wonder, is there any advantage to having a quantum algorithm which prepares state $\rho$, rather than accessing copies of $\rho$ through an experimental setup? In particular, are there natural learning tasks where performing quantum simulation gives a big benefit?  

Task \ref{task:main} answers this question in the affirmative. While other polynomial advantages to having access to the source code are known~\cite{kothari2023mean}, the learning task here demonstrates an exponential advantage over black box access.  In general experimental setups, it can sometimes be unclear how to access the complex conjugate $\rho^\ast$ of the state of interest $\rho$.  We detail this difficulty in Appendix~\ref{app:rho_star_hardness} and discuss some natural cases where it is accessible in Appendix~\ref{app:natural_states} more broadly.
However, if we have a quantum algorithm which prepares $\rho$ we can easily access $\rho^\ast$ on our quantum computer -- we simply complex conjugate the quantum algorithm itself.

More concretely, suppose unitary $U$ prepares state $\rho$ via $\rho = \Tr_{\bar{S}}\left(U |0...0\rangle\langle0...0| U^\dag\right)$, and we have an efficient quantum circuit for $U$. Then by complex conjugating every gate in the circuit we can implement unitary $U^\ast$, which will prepare state $\rho^\ast$ via $\rho^\ast = \Tr_{\bar{S}}\left(U^\ast |0...0\rangle\langle0...0| U^T\right)$.

Theorems \ref{thm:displacement_lower_bd_2_informal} and \ref{thm:displacement_informal} then exhibit an exponential cost saving for Task \ref{task:main} from having white-box access to a quantum algorithm which prepares a quantum state of interest $\rho$.

\section{Application 2: Quantum data from nature}

As a second category of applications, we consider unknown quantum states collected from nature.  For example, these states could be gathered via quantum sensors or transduced from other quantum systems.  The ability to learn about unknown states $\rho$ with exponentially fewer samples using a minimal quantum memory with only $K=2$ prompted experimental demonstrations of this idea showing they were robust even with noisy operations today~\cite{huang2022quantum}.  Despite these promising results, connecting these advantages to existing quantum sensor states today has been challenging, as many quantum sensors today are single qubit, ensembles of single qubits, stretched single qubits like GHZ states, or cavity modes~\cite{degen2017quantum}.

In contrast to the collection of qubit case, the operators considered in $d$-dimensions here connect naturally with quantum sensor arrays in regular spatial arrangements with connections to applications like very long baseline interferometry enabled by quantum communication~\cite{gottesman2012longer}.  We detail some of these connections in Appendix~\ref{app:sensor_arrays}, and given examples of setups where access to the complex conjugate state $\rho^*$ may be available either exactly or approximately.  The displacement operators are a natural description of discrete position and momentum for real-space arrays, especially in a quantum regime where few excitations are expected and background thermal noise is high.

For these scenarios, we argue that one way to view the results here is as a specialized form of mixedness testing, where for a natural class of signals, exponentially fewer samples are require to detect the presence of the signal when combined with a sea of background noise.  For applications like detection of radio signals as in NMR, it is common for an infinite temperature background to be quite strong, and this may find applications in that area.  These results provide an additional setting for which exponential advantage in signal detection is possible with quantum memory, outside of the existing hierarchies of conventional quantum memory~\cite{chen2021hierarchy}.

\section{Generalized Clifford shadows}

Given the current difficulty of going beyond minimal quantum memories, one may wonder about the information that can be gleaned from single copies in this $d$-dimensional setting. One of the most effective methods for consuming single copies in the case of qubit operators is classical shadows. Here we develop a version suitable for the $d$-level quantum system, by considering the classical shadows associated with random generalized Clifford circuits. The group $\Cld$ of generalized Cliffords in $d$ dimensions is the normaliser of the Heisenberg-Weyl group. In \cite{becker2024classical} they analyze a related construction in the continuous-variable setting.

\begin{theorem} \label{thm:clifford_shadow}
Let $d$ be prime, and let $\{\ket{j}: j \in \Zd\}$ be the computational basis for $\mathbb{C}^d$. For any unknown state $\rho$ in $\mathbb{C}^d$, any subset of the quantities $\{\bra{i}U^\dagger \rho U\ket{j}: i,j \in \Zd, i \neq j, U \in \Cld\}$ can be estimated with high probability to within additive error $\varepsilon$ using at most $\mathcal{O}(\log (d)/\varepsilon^2)$ single-copy measurements of $\rho$. Moreover, the measurements can be done up front, independent of $U,i,j$, and the classical postprocessing to compute estimates is efficient.
\end{theorem}

In other words, we can efficiently measure the transition elements of $\rho$ with respect to any and all stabiliser bases, i.e., $\{U\ket{j}: j \in \Zd\}$ for all $U \in \Cld$, using a number of copies that scales only logarithmically with $d$.  
$\Cld$ includes, for example, the quantum Fourier transform over $d$ dimensions.
In Appendix~\ref{sec:shadows}, we prove Theorem~\ref{thm:clifford_shadow} and analyse at a general level the classical shadows corresponding to the uniform distribution over $\Cld$, by evaluating $k$-fold twirl channels for $\Cld$ for $k$ up to $3$. In particular, we derive an explicit expression for the variance of estimates for $\tr(O\rho)$, for an arbitrary operator $O$ (Theorem~\ref{thm: shadows variance}). Interestingly, like the classical shadows associated with $n$-qubit Clifford circuits, this variance depends on the Hilbert-Schmidt norm of $O$, 
but it also has a further dependence on the overlaps of $O$ with displacement operators in a particular way (that is somewhat reminiscent of the variance for matchgate shadows in Ref.~\cite{wan2022matchgate}). Because of this additional dependence, the variance is not small for all low-rank observables; for instance, it scales linearly with $d$ when $O = \ket{j}\bra{j}$, which is why we only consider off-diagonal elements $\bra{i}U^\dagger \rho U\ket{j}$ with $i \neq j$ in Theorem~\ref{thm:clifford_shadow}.

It can also be shown using Theorem~\ref{thm: shadows variance} that the variance for measuring displacement operators is $\Omega(d)$. This implies that learning all displacement amplitudes using this particular classical shadows procedure may require
$\Omega(d/\varepsilon^2)$ copies of $\rho$, which is consistent with our lower bound result in Theorem~\ref{thm:displacement_lower_bd_1_informal}.

\section{Conclusions and outlook}
Exponential advantages in learning about physical systems made possible by quantum memories and quantum control represent a new opportunity for early quantum computers.
These techniques will find use in the readout stage of quantum simulations run on quantum computers and provide the exciting prospect of learning from quantum sensors with unprecedented efficiency.
While there remain many practical challenges in the implementation of this full quantum data pipeline and the full extent of applications of this technology remain unknown, there has been much recent progress.

In this work we investigated a new resource for learning – the ability to perform measurements on $\rho \otimes \rho^\ast$. We showed that for a natural set of operators, these measurements can provide exponential savings in sample complexity while using only a minimal quantum memory. Moreover, the setting of our learning task is closer to quantum sensors people imagine constructing, providing a more natural design path towards advantage in real applications.  We also resolved the question of whether a minimal quantum memory is sufficient to learn the signs of operators as well as just their magnitudes, and investigated a new form of classical shadows native to the $d$-dimensional setting.

We argued that the particular technique based on commutation achieved by simple tensor products that has been prominent in achieving learning advantages with quantum control is uniquely satisfied by the displacement operators we consider here. This raises the question of if there are more techniques yet to be discovered which can exploit entangled measurements on quantum memory to accomplish learning tasks.

\bigskip
{\bf \large Acknowledgements}
\medskip

We thank Dave Bacon for key early discussions on the properties of the Heisenberg group, Bill Huggins for detailed feedback, Ryan Babbush for important discussions, and Robin Kothari and David Gosset for developing the key ideas for the content of Appendix \ref{sec:determining_signs}.

\bibliographystyle{IEEEtran}
\bibliography{refs}

\begin{thebibliography}{10}
\providecommand{\url}[1]{#1}
\csname url@samestyle\endcsname
\providecommand{\newblock}{\relax}
\providecommand{\bibinfo}[2]{#2}
\providecommand{\BIBentrySTDinterwordspacing}{\spaceskip=0pt\relax}
\providecommand{\BIBentryALTinterwordstretchfactor}{4}
\providecommand{\BIBentryALTinterwordspacing}{\spaceskip=\fontdimen2\font plus
\BIBentryALTinterwordstretchfactor\fontdimen3\font minus
  \fontdimen4\font\relax}
\providecommand{\BIBforeignlanguage}[2]{{%
\expandafter\ifx\csname l@#1\endcsname\relax
\typeout{** WARNING: IEEEtran.bst: No hyphenation pattern has been}%
\typeout{** loaded for the language `#1'. Using the pattern for}%
\typeout{** the default language instead.}%
\else
\language=\csname l@#1\endcsname
\fi
#2}}
\providecommand{\BIBdecl}{\relax}
\BIBdecl

\bibitem{o2016efficient}
R.~O'Donnell and J.~Wright, ``Efficient quantum tomography,'' in
  \emph{Proceedings of the forty-eighth annual ACM symposium on Theory of
  Computing}, 2016, pp. 899--912.

\bibitem{anshu2023survey}
A.~Anshu and S.~Arunachalam, ``A survey on the complexity of learning quantum
  states,'' \emph{Nature Reviews Physics}, pp. 1--11, 2023.

\bibitem{aaronson2018shadow}
S.~Aaronson, ``Shadow tomography of quantum states,'' in \emph{Proceedings of
  the 50th annual ACM SIGACT symposium on theory of computing}, 2018, pp.
  325--338.

\bibitem{aaronson2018online}
S.~Aaronson, X.~Chen, E.~Hazan, S.~Kale, and A.~Nayak, ``Online learning of
  quantum states,'' \emph{Advances in neural information processing systems},
  vol.~31, 2018.

\bibitem{aaronson2019gentle}
S.~Aaronson and G.~N. Rothblum, ``Gentle measurement of quantum states and
  differential privacy,'' in \emph{Proceedings of the 51st Annual ACM SIGACT
  Symposium on Theory of Computing}, 2019, pp. 322--333.

\bibitem{buadescu2021improved}
C.~B{\u{a}}descu and R.~O'Donnell, ``Improved quantum data analysis,'' in
  \emph{Proceedings of the 53rd Annual ACM SIGACT Symposium on Theory of
  Computing}, 2021, pp. 1398--1411.

\bibitem{huang2020predicting}
H.-Y. Huang, R.~Kueng, and J.~Preskill, ``Predicting many properties of a
  quantum system from very few measurements,'' \emph{Nature Physics}, vol.~16,
  no.~10, pp. 1050--1057, 2020.

\bibitem{chen2022exponential}
S.~Chen, J.~Cotler, H.-Y. Huang, and J.~Li, ``Exponential separations between
  learning with and without quantum memory,'' in \emph{2021 IEEE 62nd Annual
  Symposium on Foundations of Computer Science (FOCS)}.\hskip 1em plus 0.5em
  minus 0.4em\relax IEEE, 2022, pp. 574--585.

\bibitem{aharonov2022quantum}
D.~Aharonov, J.~Cotler, and X.-L. Qi, ``Quantum algorithmic measurement,''
  \emph{Nature communications}, vol.~13, no.~1, p. 887, 2022.

\bibitem{huang2022quantum}
H.-Y. Huang, M.~Broughton, J.~Cotler, S.~Chen, J.~Li, M.~Mohseni, H.~Neven,
  R.~Babbush, R.~Kueng, J.~Preskill \emph{et~al.}, ``Quantum advantage in
  learning from experiments,'' \emph{Science}, vol. 376, no. 6598, pp.
  1182--1186, 2022.

\bibitem{montanaro2017learning}
\BIBentryALTinterwordspacing
A.~Montanaro, ``Learning stabilizer states by bell sampling,'' \emph{arXiv
  preprint arXiv:1707.04012}, 2017. [Online]. Available:
  \url{https://arxiv.org/abs/1707.04012}
\BIBentrySTDinterwordspacing

\bibitem{gross2021schur}
D.~Gross, S.~Nezami, and M.~Walter, ``Schur--weyl duality for the clifford
  group with applications: Property testing, a robust hudson theorem, and de
  finetti representations,'' \emph{Communications in Mathematical Physics},
  vol. 385, no.~3, pp. 1325--1393, 2021.

\bibitem{grewal2023improved}
S.~Grewal, V.~Iyer, W.~Kretschmer, and D.~Liang, ``Improved stabilizer
  estimation via bell difference sampling,'' \emph{arXiv preprint
  arXiv:2304.13915}, 2023.

\bibitem{huang2021information}
H.-Y. Huang, R.~Kueng, and J.~Preskill, ``Information-theoretic bounds on
  quantum advantage in machine learning,'' \emph{Physical Review Letters}, vol.
  126, no.~19, p. 190505, 2021.

\bibitem{inprep}
R.~King~et al., 2024, in preparation.

\bibitem{brady2023advances}
A.~J. Brady, A.~Eickbusch, S.~Singh, J.~Wu, and Q.~Zhuang, ``Advances in
  bosonic quantum error correction with gottesman-kitaev-preskill codes:
  Theory, engineering and applications,'' \emph{arXiv preprint
  arXiv:2308.02913}, 2023.

\bibitem{bluvstein2023logical}
\BIBentryALTinterwordspacing
D.~Bluvstein, S.~J. Evered, A.~A. Geim, S.~H. Li, H.~Zhou, T.~Manovitz,
  S.~Ebadi, M.~Cain, M.~Kalinowski, D.~Hangleiter \emph{et~al.}, ``Logical
  quantum processor based on reconfigurable atom arrays,'' \emph{Nature}, pp.
  1--3, 2023. [Online]. Available:
  \url{https://www.nature.com/articles/s41586-023-06927-3}
\BIBentrySTDinterwordspacing

\bibitem{osterwalder1999using}
\BIBentryALTinterwordspacing
A.~Osterwalder and F.~Merkt, ``Using high rydberg states as electric field
  sensors,'' \emph{Physical review letters}, vol.~82, no.~9, p. 1831, 1999.
  [Online]. Available:
  \url{https://journals.aps.org/prl/abstract/10.1103/PhysRevLett.82.1831}
\BIBentrySTDinterwordspacing

\bibitem{sedlacek2012microwave}
\BIBentryALTinterwordspacing
J.~A. Sedlacek, A.~Schwettmann, H.~K{\"u}bler, R.~L{\"o}w, T.~Pfau, and J.~P.
  Shaffer, ``Microwave electrometry with rydberg atoms in a vapour cell using
  bright atomic resonances,'' \emph{Nature physics}, vol.~8, no.~11, pp.
  819--824, 2012. [Online]. Available:
  \url{https://www.nature.com/articles/nphys2423}
\BIBentrySTDinterwordspacing

\bibitem{holloway2017atom}
\BIBentryALTinterwordspacing
C.~L. Holloway, M.~T. Simons, J.~A. Gordon, P.~F. Wilson, C.~M. Cooke, D.~A.
  Anderson, and G.~Raithel, ``Atom-based rf electric field metrology: from
  self-calibrated measurements to subwavelength and near-field imaging,''
  \emph{IEEE Transactions on Electromagnetic Compatibility}, vol.~59, no.~2,
  pp. 717--728, 2017. [Online]. Available:
  \url{https://ieeexplore.ieee.org/abstract/document/7812705}
\BIBentrySTDinterwordspacing

\bibitem{wade2017real}
\BIBentryALTinterwordspacing
C.~G. Wade, N.~{\v{S}}ibali{\'c}, N.~R. de~Melo, J.~M. Kondo, C.~S. Adams, and
  K.~J. Weatherill, ``Real-time near-field terahertz imaging with atomic
  optical fluorescence,'' \emph{Nature Photonics}, vol.~11, no.~1, pp. 40--43,
  2017. [Online]. Available:
  \url{https://www.nature.com/articles/nphoton.2016.214}
\BIBentrySTDinterwordspacing

\bibitem{simons2021rydberg}
\BIBentryALTinterwordspacing
M.~T. Simons, A.~B. Artusio-Glimpse, A.~K. Robinson, N.~Prajapati, and C.~L.
  Holloway, ``Rydberg atom-based sensors for radio-frequency electric field
  metrology, sensing, and communications,'' \emph{Measurement: Sensors},
  vol.~18, p. 100273, 2021. [Online]. Available:
  \url{https://www.sciencedirect.com/science/article/pii/S2665917421002361}
\BIBentrySTDinterwordspacing

\bibitem{gottesman2012longer}
\BIBentryALTinterwordspacing
D.~Gottesman, T.~Jennewein, and S.~Croke, ``Longer-baseline telescopes using
  quantum repeaters,'' \emph{Physical review letters}, vol. 109, no.~7, p.
  070503, 2012. [Online]. Available: \url{https://arxiv.org/pdf/1107.2939.pdf}
\BIBentrySTDinterwordspacing

\bibitem{degen2017quantum}
C.~L. Degen, F.~Reinhard, and P.~Cappellaro, ``Quantum sensing,'' \emph{Reviews
  of modern physics}, vol.~89, no.~3, p. 035002, 2017.

\bibitem{khabiboulline2019optical}
\BIBentryALTinterwordspacing
E.~T. Khabiboulline, J.~Borregaard, K.~De~Greve, and M.~D. Lukin, ``Optical
  interferometry with quantum networks,'' \emph{Physical review letters}, vol.
  123, no.~7, p. 070504, 2019. [Online]. Available:
  \url{https://journals.aps.org/prl/abstract/10.1103/PhysRevLett.123.070504}
\BIBentrySTDinterwordspacing

\bibitem{bland2021quantum}
\BIBentryALTinterwordspacing
J.~Bland-Hawthorn, M.~J. Sellars, and J.~G. Bartholomew, ``Quantum memories and
  the double-slit experiment: implications for astronomical interferometry,''
  \emph{JOSA B}, vol.~38, no.~7, pp. A86--A98, 2021. [Online]. Available:
  \url{https://arxiv.org/pdf/2103.07590.pdf}
\BIBentrySTDinterwordspacing

\bibitem{chen2021hierarchy}
S.~Chen, J.~Cotler, H.-Y. Huang, and J.~Li, ``A hierarchy for replica quantum
  advantage,'' \emph{arXiv preprint arXiv:2111.05874}, 2021.

\bibitem{arora2007combinatorial}
S.~Arora and S.~Kale, ``A combinatorial, primal-dual approach to semidefinite
  programs,'' in \emph{Proceedings of the thirty-ninth annual ACM symposium on
  Theory of computing}, 2007, pp. 227--236.

\bibitem{kothari2023mean}
R.~Kothari and R.~O'Donnell, ``Mean estimation when you have the source code;
  or, quantum monte carlo methods,'' in \emph{Proceedings of the 2023 Annual
  ACM-SIAM Symposium on Discrete Algorithms (SODA)}.\hskip 1em plus 0.5em minus
  0.4em\relax SIAM, 2023, pp. 1186--1215.

\bibitem{becker2024classical}
S.~Becker, N.~Datta, L.~Lami, and C.~Rouz{\'e}, ``Classical shadow tomography
  for continuous variables quantum systems,'' \emph{IEEE Transactions on
  Information Theory}, 2024.

\bibitem{wan2022matchgate}
\BIBentryALTinterwordspacing
K.~Wan, W.~J. Huggins, J.~Lee, and R.~Babbush, ``Matchgate shadows for
  fermionic quantum simulation,'' \emph{Communications in Mathematical
  Physics}, vol. 404, no.~2, p. 629–700, Oct. 2023. [Online]. Available:
  \url{http://dx.doi.org/10.1007/s00220-023-04844-0}
\BIBentrySTDinterwordspacing

\bibitem{cotler2021revisiting}
\BIBentryALTinterwordspacing
J.~Cotler, H.-Y. Huang, and J.~R. McClean, ``Revisiting dequantization and
  quantum advantage in learning tasks,'' \emph{arXiv preprint
  arXiv:2112.00811}, 2021. [Online]. Available:
  \url{https://arxiv.org/abs/2112.00811}
\BIBentrySTDinterwordspacing

\bibitem{ebler2023optimal}
\BIBentryALTinterwordspacing
D.~Ebler, M.~Horodecki, M.~Marciniak, T.~M{\l}ynik, M.~T. Quintino, and
  M.~Studzi{\'n}ski, ``Optimal universal quantum circuits for unitary complex
  conjugation,'' \emph{IEEE Transactions on Information Theory}, 2023.
  [Online]. Available: \url{https://arxiv.org/pdf/2206.00107.pdf}
\BIBentrySTDinterwordspacing

\bibitem{helgaker2013molecular}
T.~Helgaker, P.~Jorgensen, and J.~Olsen, \emph{Molecular electronic-structure
  theory}.\hskip 1em plus 0.5em minus 0.4em\relax John Wiley \& Sons, 2013.

\bibitem{poulin2009sampling}
D.~Poulin and P.~Wocjan, ``Sampling from the thermal quantum gibbs state and
  evaluating partition functions with a quantum computer,'' \emph{Physical
  review letters}, vol. 103, no.~22, p. 220502, 2009.

\bibitem{riera2012thermalization}
A.~Riera, C.~Gogolin, and J.~Eisert, ``Thermalization in nature and on a
  quantum computer,'' \emph{Physical review letters}, vol. 108, no.~8, p.
  080402, 2012.

\bibitem{chowdhury2016quantum}
A.~N. Chowdhury and R.~D. Somma, ``Quantum algorithms for gibbs sampling and
  hitting-time estimation,'' \emph{arXiv preprint arXiv:1603.02940}, 2016.

\bibitem{su2021fault}
Y.~Su, D.~W. Berry, N.~Wiebe, N.~Rubin, and R.~Babbush, ``Fault-tolerant
  quantum simulations of chemistry in first quantization,'' \emph{PRX Quantum},
  vol.~2, no.~4, p. 040332, 2021.

\bibitem{babbush2023_chem}
R.~Babbush, W.~J. Huggins, D.~W. Berry, S.~F. Ung, A.~Zhao, D.~R. Reichman,
  H.~Neven, A.~D. Baczewski, and J.~Lee, ``Quantum simulation of exact electron
  dynamics can be more efficient than classical mean-field methods,''
  \emph{Nature Communications}, vol.~14, no. 4058, 2023.

\bibitem{schuster2023learning}
T.~Schuster, M.~Niu, J.~Cotler, T.~O'Brien, J.~R. McClean, and M.~Mohseni,
  ``Learning quantum systems via out-of-time-order correlators,''
  \emph{Physical Review Research}, vol.~5, no.~4, p. 043284, 2023.

\bibitem{artusio2022modern}
\BIBentryALTinterwordspacing
A.~Artusio-Glimpse, M.~T. Simons, N.~Prajapati, and C.~L. Holloway, ``Modern rf
  measurements with hot atoms: A technology review of rydberg atom-based radio
  frequency field sensors,'' \emph{IEEE Microwave Magazine}, vol.~23, no.~5,
  pp. 44--56, 2022. [Online]. Available:
  \url{https://ieeexplore.ieee.org/abstract/document/9748947}
\BIBentrySTDinterwordspacing

\bibitem{tresp2016single}
\BIBentryALTinterwordspacing
C.~Tresp, C.~Zimmer, I.~Mirgorodskiy, H.~Gorniaczyk, A.~Paris-Mandoki, and
  S.~Hofferberth, ``Single-photon absorber based on strongly interacting
  rydberg atoms,'' \emph{Physical Review Letters}, vol. 117, no.~22, p. 223001,
  2016. [Online]. Available:
  \url{https://journals.aps.org/prl/abstract/10.1103/PhysRevLett.117.223001}
\BIBentrySTDinterwordspacing

\bibitem{shannon1949communication}
\BIBentryALTinterwordspacing
C.~E. Shannon, ``Communication in the presence of noise,'' \emph{Proceedings of
  the IRE}, vol.~37, no.~1, pp. 10--21, 1949. [Online]. Available:
  \url{https://ieeexplore.ieee.org/document/1697831}
\BIBentrySTDinterwordspacing

\bibitem{knill1998power}
E.~Knill and R.~Laflamme, ``Power of one bit of quantum information,''
  \emph{Physical Review Letters}, vol.~81, no.~25, p. 5672, 1998.

\bibitem{o2022quantum}
T.~E. O’Brien, L.~B. Ioffe, Y.~Su, D.~Fushman, H.~Neven, R.~Babbush, and
  V.~Smelyanskiy, ``Quantum computation of molecular structure using data from
  challenging-to-classically-simulate nuclear magnetic resonance experiments,''
  \emph{PRX Quantum}, vol.~3, no.~3, p. 030345, 2022.

\bibitem{grassberger2001note}
J.~Grassberger and G.~H{\"o}rmann, ``A note on representations of the finite
  heisenberg group and sums of greatest common divisors,'' \emph{Discrete
  Mathematics \& Theoretical Computer Science}, vol.~4, 2001.

\bibitem{rosenberg2004selective}
J.~Rosenberg, ``A selective history of the stone-von neumann theorem,''
  \emph{Contemporary Mathematics}, vol. 365, pp. 331--354, 2004.

\bibitem{asadian2016heisenberg}
A.~Asadian, P.~Erker, M.~Huber, and C.~Kl{\"o}ckl, ``Heisenberg-weyl
  observables: Bloch vectors in phase space,'' \emph{Physical Review A},
  vol.~94, no.~1, p. 010301, 2016.

\bibitem{mele2023introduction}
A.~A. Mele, ``Introduction to haar measure tools in quantum information: A
  beginner's tutorial,'' \emph{arXiv preprint arXiv:2307.08956}, 2023.

\bibitem{hostens2005stabilizer}
E.~Hostens, J.~Dehaene, and B.~De~Moor, ``Stabilizer states and clifford
  operations for systems of arbitrary dimensions and modular arithmetic,''
  \emph{Physical Review A}, vol.~71, no.~4, p. 042315, 2005.

\bibitem{fulton2013representation}
W.~Fulton and J.~Harris, \emph{Representation theory: a first course}.\hskip
  1em plus 0.5em minus 0.4em\relax Springer Science \& Business Media, 2013,
  vol. 129.

\bibitem{gross2006hudson}
D.~Gross, ``Hudson’s theorem for finite-dimensional quantum systems,''
  \emph{Journal of mathematical physics}, vol.~47, no.~12, 2006.

\bibitem{kocia2017discrete}
L.~Kocia, Y.~Huang, and P.~Love, ``Discrete wigner function derivation of the
  aaronson--gottesman tableau algorithm,'' \emph{Entropy}, vol.~19, no.~7, p.
  353, 2017.

\end{thebibliography}

\onecolumngrid
\appendix
\pagebreak

\section{Background} \label{sec:background}

\subsection{Displacement operators} \label{sec:disp_ops}

Given a $d$-dimensional quantum system $\mathbb{C}^d$ with basis $\{|0\rangle,\dots,|d-1\rangle\}$, we define the operators $X$ and $Z$ by
\begin{align}
X &: \ket{j} \mapsto \ket{j+1} \\
Z &: \ket{j} \mapsto \omega^j \ket{j}.
\end{align}
Here and throughout, the addition inside the ket is modulo $d$, and 
\begin{equation}
\omega \coloneqq e^{i 2\pi / d}.
\end{equation}
The matrix representations of $X$ and $Z$ are
\begin{align}
X &=
\begin{pmatrix}
0 & 0 & 0 & \dots & 0 & 1 \\
1 & 0 & 0 & \dots & 0 & 0 \\
0 & 1 & 0 & \dots & 0 & 0 \\
0 & 0 & 1 & \dots & 0 & 0 \\
\vdots & \vdots & \vdots & \ddots & \vdots & \vdots \\
0 & 0 & 0 & \dots & 1 & 0 \\
\end{pmatrix}, \\
Z &=
\begin{pmatrix}
1 & 0 & 0 & \dots & 0 \\
0 & \omega & 0 & \dots & 0 \\
0 & 0 & \omega^2 & \dots & 0 \\
\vdots & \vdots & \vdots & \ddots & \vdots \\
0 & 0 & 0 & \dots & \omega^{d-1}.  \label{Z} \\
\end{pmatrix}
\end{align}
$X,Z$ are traceless and unitary but not Hermitian in general. $X$ is a \emph{shift} rotation, generated by the discrete momentum operator, and $Z$ is a \emph{phase} rotation, generated by the discrete position operator. For $d=2$, $X$ and $Z$ coincide with the usual Pauli matrices, and in this case they are in fact Hermitian. $X,Z$ obey the commutation relation
\begin{align} \label{eq:XZ_comm}
Z X &= \omega X Z \nonumber\\
\implies \ Z^p X^q &= \omega^{qp} X^q Z^p
\end{align}

Define the \emph{displacement operator} $D_{q,p}$ by
\begin{equation}
D_{q,p} = e^{i \pi q p / d} X^q Z^p
\end{equation}
$D_{q,p}$ acts on basis vectors by
\begin{equation} \label{Dqp act}
D_{q,p} : |j\rangle \rightarrow e^{i \pi (q+2j) p / d} |j+q\rangle
\end{equation}
$D_{q,p}$ is a unitary which can be interpreted as shifting by the vector $(q,p)$ in discrete position-momentum phase space.

\begin{proposition} \label{prop:disp_properties}
The displacement operators have properties
\begin{align}
&D^{-1}_{q,p} = D^\dag_{q,p} = D_{-q,-p} \label{D dagger} \\
&D^\ast_{q,p} = D_{q,-p} \\
&D^T_{q,p} = D_{-q,p} \label{eq:D_transpose}\\
&D_{q',p'} D_{q,p} = e^{i 2\pi (qp' - q'p) / d} D_{q,p} D_{q',p'} \label{eq:D_comm}\\
&D_{q,p}^k = D_{kq,kp}
\end{align}
\end{proposition}

\begin{proposition} \label{prop:disp_basis}
$\{D_{q,p}\}$ form a basis of the $(d \times d)$-dimensional space of operators on $\mathbb{C}^d$. Moreover, they are orthogonal in the Hilbert-Schmidt inner product.
\begin{equation} \label{Dqp orthogonal}
\Tr\left(D_{q,p}^\dagger D_{q',p'}\right) = d \cdot \delta_{q,q'} \delta_{p,p'}
\end{equation}
\end{proposition}

Proposition \ref{prop:disp_basis} lets us decompose any $d$-dimensional density matrix $\rho$ as
\begin{equation}
\rho = \frac{1}{d} \sum_{q,p} \Tr\left(D_{q,p} \rho\right) D^\dag_{q,p}
\end{equation}
This can be viewed as the $d$-dimensional analog of the Bloch vector representation
\begin{equation}
\rho = \frac{1}{d} \Big(\mathbbm{1} + \sum_{(q,p) \neq (0,0)} y_{q,p} D^\dag_{q,p}\Big) \quad , \quad y_{q,p} = \Tr\left(D_{q,p} \rho\right)
\end{equation}

\subsection{Bosonic limit $d \rightarrow \infty$} \label{sec:cts_limit}

For an overview of continuous bosonic modes, see \cite{brady2023advances}.

As $d$ increases, the $d$-dimensional system described above forms an increasingly good discrete approximation to an infinite-dimensional continuous bosonic mode. To formalize this, let's set
\begin{equation}
x = \sqrt{\frac{\pi}{d}} (q,p)
\end{equation}
$x$ lives on a lattice with spacing $\sim 1/\sqrt{d}$ and size $\sim \sqrt{d}$ with periodic boundary conditions. As $d$ goes to infinity, $x$ becomes a continuous phase space variable in the plane.

For the quantum harmonic oscillator with frequency $\omega$ and mass $m$, the characteristic length and momentum scales are $\sqrt{\hbar / m \omega}$ and $\sqrt{\hbar m \omega}$ respectively. Thus if we can measure the position to precision $\delta$, this is effectively measuring the particle on a discrete lattice with
\begin{equation}
d \sim \frac{\hbar}{m \omega \delta^2}
\end{equation}

The analog of Proposition \ref{prop:disp_basis} in the continuous limit is
\begin{equation}
\Tr{D_{x} D_{x'}} = \pi \delta^2(x - x')
\end{equation}
which leads to the infinite-dimensional Bloch representation
\begin{equation}
\rho = \frac{1}{\pi} \int [d^2 x] y_{x} D^\dag_{x} \quad , \quad y_{x} = \Tr\left(D_{x} \rho\right)
\end{equation}

$D_{x}$ represents a shift in phase plane by $x$. The commutation relation Equation \ref{eq:D_comm} becomes
\begin{equation}
D_{x'} D_{x} = e^{i 2 x^T J x'} D_{x} D_{x'}
\end{equation}
where $J$ is the \emph{symplectic form}, defined by
\begin{equation}
J =
\begin{pmatrix}
    0 & 1 \\
    -1 & 0
\end{pmatrix}.
\end{equation}

If we have $n$ bosonic modes, the phase space variable becomes
\begin{equation}
\vec{x} = (\vec{q},\vec{p}) = (q_1,\dots,q_n,p_1,\dots,p_n)
\end{equation}

\subsection{Gaussian operations}

Gaussian unitaries can be defined in three equivalent ways:
\begin{itemize}
    \item They are the \emph{normalizer} of the Heisenberg-Weyl group. That is, they conjugate displacement operators to themselves.
    \item They are symplectic transformations of phase space. (See Equation \ref{symplectic_phase_space}.)
    \item They are generated by Hamiltonians quadratic in position and momentum operators.
\end{itemize}

A matrix $S \in \mathbb{R}^{2n \times 2n}$ is \emph{symplectic} if
\begin{equation}
S J S^T = J \quad , \quad J =
\begin{pmatrix}
    0 & I_n \\
    -I_n & 0
\end{pmatrix}
\end{equation}
$J$ is the symplectic form. Denote the set of symplectic matrices $\Sp(2n)$.

For every symplectic $S \in \Sp(2n)$, there is a corresponding Gaussian unitary $U_S$ on $n$ bosonic modes. This acts on the position and momentum operators via
\begin{equation} \label{symplectic_phase_space}
U_S^\dag \vec{x} U_S = S \vec{x}
\end{equation}
The adjoint action of a Gaussian unitary on the displacement operators is
\begin{equation}
U_S^\dag D(\vec{x}) U_S  = D(S^{-1} \vec{x})
\end{equation}
Note that for a single mode, the symplectic condition $S J S^T = J$ simply becomes $\det{S} = 1$. Gaussian unitaries are the infinite-dimensional version of Clifford operations.

\section{Learning displacement amplitudes} \label{sec:disp_alg}

\begin{definition}
The \emph{displacement amplitudes} of a $d$-dimensional state $\rho$ are
\begin{equation}
y_{q,p} = \Tr\left(D_{q,p} \rho\right)
\end{equation}
\end{definition}

Note however that the coefficients $y_{q,p}$ are now \emph{complex} in general, since $D_{q,p}$ are not always Hermitian. The Hermiticity of $\rho$ is reflected in the relation 
\begin{equation}
y_{q,p}^\ast = y_{-q,-p}
\end{equation}

Given copies of a $d$-dimensional state and its conjugate $\rho \otimes \rho^\ast$, the first phase of our learning algorithm will aim to learn the displacement amplitudes $\{ \pm y_{q,p}\}$ up to a sign $\pm$. The second phase will use entangled measurements across a few more copies to resolve the signs.

\subsection{Estimating displacement amplitudes up to a sign}

Using commutation relation Equation \ref{eq:D_comm}, it can be checked that the following operators mutually commute when acting on two systems $\mathbb{C}^d \otimes \mathbb{C}^d$.
\begin{equation}
\{ D_{q,p} \otimes D_{-q,p} \}
\end{equation}
This suggests that we can measure two systems in the joint eigenbasis of these operators. If we placed $\rho \otimes \rho^\ast$ in the two systems for some $d$-dimensional state $\rho$, we would have
\begin{align} \label{eq:DD_rho}
\Tr\left((D_{q,p} \otimes D_{-q,p}) (\rho \otimes \rho^\ast)\right) &= \Tr\left(D_{q,p} \rho\right) \Tr\left(D^T_{q,p} \rho^\ast\right) \nonumber\\
&= \Tr\left(D_{q,p} \rho\right)^2 = y_{q,p}^2
\end{align}
Here we used Proposition \ref{prop:disp_properties}.

We now construct the desired eigenbasis. It will form a $d$-dimensional generalization of the Bell basis on 2 qubits. Define
\begin{align}
|\Phi_{0,0}\rangle &= \frac{1}{\sqrt{d}} \sum_j |j\rangle |-j\rangle \label{Phi00} \\
|\Phi_{a,b}\rangle &= (X^a Z^b \otimes \mathbbm{1}) |\Phi_{0,0}\rangle \nonumber\\
&= \frac{1}{\sqrt{d}} \sum_j e^{i 2\pi b j / d} |j+a\rangle |-j\rangle \label{Phiab}
\end{align}
It can be checked that $\{|\Phi_{a,b}\rangle\}$ forms an orthonormal basis of $\mathbb{C}^d \otimes \mathbb{C}^d$. Furthermore, we can calculate
\begin{equation} \label{eq:DD_phi}
(D_{q,p} \otimes D_{-q,p}) |\Phi_{a,b}\rangle = e^{i 2\pi (ap - bq) / d} |\Phi_{a,b}\rangle.
\end{equation}

We provide more intuition for the simultaneous diagonalisability of the $D_{q,p} \otimes D_{-q,p}$ and their common eigenbasis $\ket{\Phi_{a,b}}$ in Appendix~\ref{sec:adjoint}.

\begin{algorithm}[H] \caption{Learning displacement amplitudes up to a sign.} \label{alg:disp}

\textbf{Input:}
\begin{itemize}
    \item A list of $M$ displacement indices $\mathcal{S} \in \mathbb{Z}_d^2$, $|\mathcal{S}| = M$.
    \item Precision $\varepsilon$.
    \item $N$ copies of $\rho \otimes \rho^\ast$, where $\rho$ is an unknown quantum state in $d$ dimensions.
\end{itemize}

\medskip
\textbf{Output:} For each $(q,p) \in \mathcal{S}$, either output $\hat{u}_{q,p} = 0$ or some $\hat{u}_{q,p}$ with $|\hat{u}_{q,p}| \geq \frac{\sqrt{2}}{\sqrt{3}} \varepsilon$. If $\hat{u}_{q,p} = 0$, we are guaranteed $|y_{q,p}| \leq \varepsilon$. If $|\hat{u}_{q,p}| \geq \frac{\sqrt{2}}{\sqrt{3}} \varepsilon$, we are guaranteed $|\pm \hat{u}_{q,p} - y_{q,p}| \leq \frac{\sqrt{2}}{2\sqrt{3}} \varepsilon$ for one of the choices of sign $\pm$.

\medskip
\textbf{Algorithm:}
\begin{enumerate}
    \item Measure each copy of $\rho \otimes \rho^\ast$ in the basis $\{|\Phi_{a,b}\rangle\}$, receiving outcomes $\{(a^{(k)},b^{(k)})\}_{k=1}^N$.
    \item Now given $(q,p) \in \mathcal{S}$, compute
    \begin{equation}
    \hat{v}_{q,p} = \frac{1}{N} \sum_{k=1}^N \exp\left(i 2\pi (a^{(k)}p - b^{(k)}q) / d\right).
    \end{equation}
    \item If $|\hat{v}_{q,p}| \leq \frac{2}{3} \varepsilon^2$, output $\hat{u}_{q,p} = 0$.
    \item Else output $\hat{u}_{q,p} = \sqrt{\hat{v}_{q,p}}$, choosing the primary square root without loss of generality.
\end{enumerate}
\end{algorithm}

\begin{theorem} \label{thm:disp_alg}
Algorithm \ref{alg:disp} succeeds with high probability using $N = \mathcal{O}(\log{M}/\varepsilon^4)$ copies of $\rho \otimes \rho^\ast$.
\end{theorem}
\begin{proof}
Suppose we measure $\rho \otimes \rho^\ast$ in the basis $\{|\Phi_{a,b}\rangle\}$, and get the distribution $(a,b) \sim \mathcal{P}$. From Equations \ref{eq:DD_rho} and \ref{eq:DD_phi}, we can construct an estimator for $y_{q,p}^2$.
\begin{align} \label{eq:estimator}
y_{q,p}^2 &= \Tr\left((D_{q,p} \otimes D_{-q,p}) (\rho \otimes \rho^\ast)\right) \nonumber\\
&= \mathbb{E}_{(a,b) \sim \mathcal{P}} \ \Tr\left((D_{q,p} \otimes D_{-q,p}) |\Phi_{a,b}\rangle\langle\Phi_{a,b}|\right) \nonumber\\
&= \mathbb{E}_{(a,b) \sim \mathcal{P}} \ \exp\left(i 2\pi (ap - bq) / d\right)
\end{align}
Applying Hoeffding's inequality in the complex plane tells us that with $N = \mathcal{O}(\log{M} / \varepsilon^4)$ copies,
\begin{equation}
\big|\hat{v}_{q,p} - y_{q,p}^2\big| \leq \frac{1}{3} \varepsilon^2
\end{equation}
for any $M$ displacement operators $D_{q,p}$ with high probability.

\bigskip
{\bf Case 1:} $|\hat{v}_{q,p}| \leq \frac{2}{3} \varepsilon^2$. Then with high probability
\begin{equation}
|y_{q,p}^2| \leq |\hat{v}_{q,p}| + |\hat{v}_{q,p} - y_{q,p}^2| \leq \frac{2}{3} \varepsilon^2 + \frac{1}{3} \varepsilon^2 = \varepsilon^2
\ \implies \ |y_{q,p}| \leq \varepsilon
\end{equation}
and the estimate $\hat{u}_{q,p} = 0$ satisfies $|\hat{u}_{q,p} - y_{q,p}| \leq \varepsilon$.

\bigskip
{\bf Case 2:} $|\hat{v}_{q,p}| > \frac{2}{3} \varepsilon^2$. Then $|\hat{u}_{q,p}| > \frac{\sqrt{2}}{\sqrt{3}} \varepsilon$. Negating $\hat{u}_{q,p}$ if necessary, suppose we choose the correct hemisphere for $\hat{u}_{q,p}$. This guarantees that
\begin{equation}
|\hat{u}_{q,p} + y_{q,p}| \geq |\hat{u}_{q,p}| > \frac{\sqrt{2}}{\sqrt{3}} \varepsilon
\end{equation}
Now
\begin{align}
\frac{1}{3} \varepsilon^2 &\geq |\hat{v}_{q,p} - y_{q,p}^2| \\
&= |\hat{u}_{q,p}^2 - y_{q,p}^2| = |\hat{u}_{q,p} + y_{q,p}| \cdot |\hat{u}_{q,p} - y_{q,p}| \\
&\geq \frac{\sqrt{2}}{\sqrt{3}} \ \varepsilon \cdot |\hat{u}_{q,p} - y_{q,p}|
\end{align}
\begin{equation}
\implies \ |\hat{u}_{q,p} - y_{q,p}| \leq \frac{\sqrt{2}}{2\sqrt{3}} \varepsilon
\end{equation}
finishing the proof.
\end{proof}

\subsection{Implementing the measurements}

\begin{proposition}
If we encode the $d$-dimensional system in $\mathcal{O}(\log{d})$ qubits on a quantum computer, then we can measure in the basis $\{|\Phi_{a,b}\rangle\}$ in time $\poly{\log{d}}$.
\end{proposition}
\begin{proof}
The \emph{quantum Fourier transform} $W$ in $d$ dimensions is the map
\begin{equation}
W : |b\rangle \rightarrow \frac{1}{\sqrt{d}} \sum_j e^{-i 2\pi b j / d} |j\rangle
\end{equation}
with matrix representation
\begin{equation}
W = \frac{1}{\sqrt{d}}
\begin{pmatrix}
1 & 1 & 1 & \dots & 1 \\
1 & \omega^{-1} & \omega^{-2} & \dots & \omega^{-(d-1)} \\
1 & \omega^{-2} & \omega^{-4} & \dots & \omega^{-2(d-1)} \\
\vdots & \vdots & \vdots & \ddots & \vdots \\
1 & \omega^{-(d-1)} & \omega^{-2(d-1)} & \dots & \omega^{-(d-1)^2} \\
\end{pmatrix}
\end{equation}
$W$ is a unitary. In relation to $X,Z$, it has the properties
\begin{align}
W X W^\dag &= Z \\
W Z W^\dag &= X^\dag
\end{align}
If we encode the $d$-dimensional system in $\mathcal{O}(\log{d})$ qubits, then $W$ can be implemented in time $\poly{\log{d}}$.

Define the \emph{controlled shift} operator $CX$ on two $d$-dimensional systems to act by
\begin{equation}
CX : |j\rangle |l\rangle \rightarrow |j+l\rangle |l\rangle
\end{equation}

$W$ and $CX$ together can be used to transform the standard product basis on two $d$-dimensional systems to the entangled $\{|\Phi_{a,b}\rangle\}$ basis.
\begin{equation}
(CX)^{-1} \cdot (\mathbbm{1} \otimes W) \cdot |a\rangle |b\rangle = |\Phi_{a,b}\rangle
\end{equation}
This lets us implement a measurement in the basis $\{|\Phi_{a,b}\rangle\}$ by first inverting this transformation and then measuring in the standard product basis.
\end{proof}

\begin{figure}[H]
\centering
\begin{subfigure}{.5\textwidth}
  \centering
  \includegraphics[width=3in]{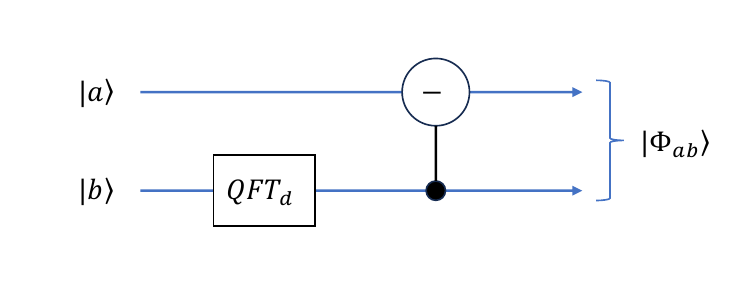}
  \caption{Transforming between bases}
\end{subfigure}%
\begin{subfigure}{.5\textwidth}
  \centering
  \includegraphics[width=3in]{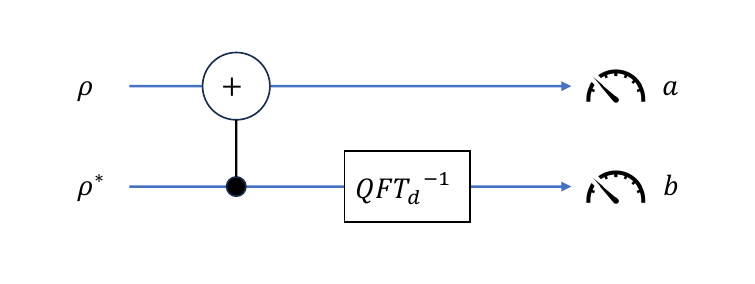}
  \caption{Algorithm \ref{alg:disp}}
\end{subfigure}
\caption{}
\label{plt_prod}
\end{figure}

\subsection{$n$ subsystems} \label{sec:n_bosons}

Suppose now we have $n$ $d$-dimensional subsystems, with Hilbert space $(\mathbb{C}^d)^{\otimes n}$. We can promote $q,p,a,b$ to vectors $\vec{q},\vec{p},\vec{a},\vec{b} \in \mathbb{Z}_d^n$ and define
\begin{align}
D_{\vec{q},\vec{p}} &= D_{q_1,p_1} \otimes \dots \otimes D_{q_n,p_n} \\
|\Phi_{\vec{a},\vec{b}}\big\rangle &= |\Phi_{a_1,b_1}\rangle \otimes \dots \otimes |\Phi_{a_n,b_n}\rangle
\end{align}

All statements in Proposition \ref{prop:disp_properties} hold with $q,p$ replaced with the vectors $\vec{q},\vec{p}$. In the commutation relation Equation \ref{eq:D_comm}, the products $qp',q'p$ are to be replaced with dot products $\vec{q}\cdot\vec{p'}, \ \vec{q'}\cdot\vec{p}$.

Following Proposition \ref{prop:disp_basis}, $\{D_{\vec{q},\vec{p}}\}$ form a basis of the $(d^n \times d^n)$-dimensional space of operators on $(\mathbb{C}^d)^{\otimes n}$, orthogonal in the Hilbert-Schmidt inner product.
\begin{equation}
\Tr\left(D_{\vec{q},\vec{p}} D_{\vec{q'},\vec{p'}}\right) = d \cdot \delta_{q_1,q'_1} \delta_{p_1,p'_1} \dots \delta_{q_n,q'_n} \delta_{p_n,p'_n}
\end{equation}
We have the Bloch vector decomposition
\begin{equation}
\rho = \frac{1}{d^n} \Big(\mathbbm{1} + \sum_{\vec{q},\vec{p}} y_{\vec{q},\vec{p}} D^\dag_{\vec{q},\vec{p}}\Big) \quad , \quad y_{\vec{q},\vec{p}} = \Tr\left(D(\vec{q},\vec{p}) \rho\right)
\end{equation}

The analog of Equation \ref{eq:estimator} is
\begin{equation} \label{eq:estimator_n}
y_{\vec{q},\vec{p}}^2 = \mathbb{E}_{(\vec{a},\vec{b})} \exp\left(i 2\pi (\vec{a} \cdot \vec{p} - \vec{b} \cdot \vec{q}) / d\right)
\end{equation}

\begin{theorem}
The natural generalization of Algorithm \ref{alg:disp} to $n$ qudits has the same guarantee as in Theorem $\ref{thm:disp_alg}$, and likewise the measurements can be implemented in time $\poly{\log{d}}$.
\end{theorem}

\subsection{Learning algorithm in infinite dimensions}

Using Gaussian unitaries, we can phrase Algorithm \ref{alg:disp} in the infinite-dimensional setting. Step 1 of Algorithm \ref{alg:disp} is to measure $\rho \otimes \rho^\ast$ in the generalized Bell basis $\{|\Phi_{a,b}\rangle\}$. This is achieved by applying a controlled shift Gaussian operation, followed by \emph{homodyne measurements} along certain quadratures. Let the position and momentum variables of the first and second register be $q_1,p_1,q_2,p_2$.

\begin{figure}[H]
\centering
  \includegraphics[width=4in]{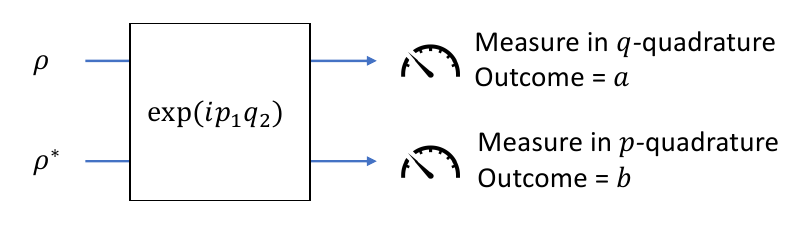}
  \caption{Algorithm \ref{alg:disp} in infinite dimensions}
  \label{fig:cts_alg}
\
\end{figure}

On two modes, the controlled shift corresponds to the Gaussian unitary
\begin{equation}
CX = \exp{i p_1 q_2}
\end{equation}
It has symplectic matrix
\begin{equation}
CX = U_S \quad , \quad S =
\begin{pmatrix}
    1 & 0 & 0 & 0 \\
    1 & 1 & 0 & 0 \\
    0 & 0 & 1 & -1 \\
    0 & 0 & 0 & 1
\end{pmatrix}.
\end{equation}

In the Heisenberg picture, the procedure shown in Figure \ref{fig:cts_alg} transforms the variables $q_1, p_2$ to the observables
\begin{equation}
a = q_1 + q_2 \quad , \quad b = - p_1 + p_2
\end{equation}

We finish the procedure by measuring the first register in the $q$ quadrature, and the second register in the $p$ quadrature. Note measuring in the $p$ quadrature is equivalent to doing a quantum Fourier transform, or phase shift, before measuring in the $q$ quadrature.

Define
\begin{equation}
\alpha = (a,b)
\end{equation}
The estimator for the learning algorithm has a nice expression in terms of the symplectic product. The analogue of Equation \ref{eq:estimator} is
\begin{equation}
y_{x}^2 = \mathbb{E}_{\alpha} \exp\left(i 2 \alpha^T J x\right)
\end{equation}

An example of a class of states where the displacement amplitudes are non-trivial are the GKP codestates. These are the states stabilized by two displacement operators $D_{x}$ and $D_{x'}$ for some $x$ and $x'$ satisfying $x^T J x' = \pi$.

\section{Determining signs} \label{sec:determining_signs}

\subsection{Determining signs of displacement operators}

Here we describe a method to resolve the signs $\pm$ with an additional $\mathcal{O}(\log{M} / \varepsilon^4)$ copies of $\rho$. The method only consumes a single copy of $\rho$ at a time. The ideas in this section are taken from a different work, currently in preparation \cite{inprep}.

The intuition behind the algorithm is that the ability to measure $|x|$ and $|x + \varepsilon|$ should determine $x$; that is, we have estimated the magnitude with a shifted origin. The algorithm's strategy is to find a hypothesis quantum state $\tilde{\rho}$ which it can use to shift the origin of the magnitude measurements. In order to find $\tilde{\rho}$, we will use Algorithm \ref{alg:hypo_state} in Section \ref{sec:hypo_state} below, which in turn relies on results from \cite{arora2007combinatorial, aaronson2018online}.

\begin{algorithm}[H] \caption{Determining the signs of displacement amplitudes.} \label{alg:disp_signs}

\textbf{Input:}
\begin{itemize}
    \item A list of $M$ displacement indices $\mathcal{S} \in \mathbb{Z}_d^2$, $|\mathcal{S}| = M$.
    \item Precision $\varepsilon$.
    \item $\mathcal{O}((\log{d} + \log{M}) / \varepsilon^4)$ copies of unknown state $\rho$.
    \item The output of Algorithm \ref{alg:disp} for $\mathcal{S}$ and $\rho$ and precision $\varepsilon / 2$.
\end{itemize}

\medskip
\textbf{Output:} For each $(q,p) \in \mathcal{S}$ such that $\hat{u}_{q,p} \neq 0$, output a sign $\hat{r}_{q,p}$ so that $|\hat{r}_{q,p} \hat{u}_{q,p} - y_{q,p}| \leq \varepsilon$. ($\hat{u}_{q,p}$ is the output of Algorithm \ref{alg:disp}.)

\medskip
\textbf{Algorithm:}
\begin{enumerate}
    \item Denote $\mathcal{T} = \{(q,p) : \hat{u}_{q,p} \neq 0\}$. We must output a sign $\hat{r}_{q,p}$ for each $(q,p) \in \mathcal{T}$.
    \item Using Algorithm \ref{alg:hypo_state}, find a classical description of a density matrix $\tilde{\rho}$ satisfying
    \begin{equation}
    |\Tr\left(D_{q,p} \tilde{\rho}\right) - s_{q,p} \hat{u}_{q,p}| \leq \frac{\sqrt{2}}{2\sqrt{3}} \varepsilon \quad \forall \ (q,p) \in \mathcal{T}
    \end{equation}
    for some choices of signs $s_{q,p} \in \{+1,-1\}$. This requires $\mathcal{O}(\log{d} / \varepsilon^4)$ single-copy measurements of $\rho$. Using the classical description of $\tilde{\rho}$, we are able to manufacture copies of its complex conjugate state $\tilde{\rho}^\ast$ on our quantum computer.
    \item Measure $N = \mathcal{O}(\log{M} / \varepsilon^4)$ copies of $\rho \otimes \tilde{\rho}^\ast$ in the basis $\{|\Phi_{a,b}\rangle\}$, receiving outcomes $\{(a^{(k)},b^{(k)})\}_{k=1}^N$.
    \item Now given $(q,p) \in \mathcal{T}$, compute
    \begin{equation}
    \hat{v}_{q,p} = \frac{1}{N} \sum_{k=1}^N \exp\left(i 2\pi (a^{(k)}p - b^{(k)}q) / d\right).
    \end{equation}
    \item Output
    \begin{equation}
    \hat{r}_{q,p} =
    \begin{cases}
        +s_{q,p} & |\arg\left(\hat{v}_{q,p}\right) - 2 \arg\left(\hat{u}_{q,p}\right)| \leq \pi/2 \\
        -s_{q,p} & |\arg\left(\hat{v}_{q,p}\right) - 2 \arg\left(\hat{u}_{q,p}\right)| > \pi/2
    \end{cases}
    \end{equation}
\end{enumerate}
\end{algorithm}

\begin{remark}
The runtime of Algorithm \ref{alg:disp_signs} is polynomial in the dimension $d$ of the Hilbert space and in $\varepsilon^{-1}$.
\end{remark}

\begin{theorem}
Algorithm \ref{alg:disp_signs} succeeds with high probability.
\end{theorem}
\begin{proof}
Denote $\tilde{y}_{q,p} = \Tr\left(D_{q,p} \tilde{\rho}\right)$. By design, $|\tilde{y}_{q,p} - s_{q,p} \hat{u}_{q,p}| \leq \frac{\sqrt{2}}{2\sqrt{3}} \varepsilon \ \forall (q,p) \in \mathcal{T}$. Recall that $|\hat{u}_{q,p}| \geq \frac{\sqrt{2}}{\sqrt{3}} \varepsilon \ \forall (q,p) \in \mathcal{T}$. This means that
\begin{equation}
|\arg\left(\tilde{y}_{q,p}\right) - \arg\left(s_{q,p} \hat{u}_{q,p}\right)| \leq \pi/6 \ \forall (q,p) \in \mathcal{T}
\end{equation}
Recall that $|y_{q,p} - r_{q,p} \hat{u}_{q,p}| \leq \frac{\sqrt{2}}{2\sqrt{3}} \varepsilon \ \forall (q,p) \in \mathcal{T}$ for one of the choices of sign $r_{q,p} \in \{+1,-1\}$. It is also the case that
\begin{equation}
|\arg\left(y_{q,p}\right) - \arg\left(r_{q,p} \hat{u}_{q,p}\right)| \leq \pi/6 \ \forall (q,p) \in \mathcal{T}
\end{equation}
If we guessed correctly with $\tilde{\rho}$ and $s_{q,p} = r_{q,p}$, then
\begin{equation}
s_{q,p} = r_{q,p} \ \implies \ |\arg\left(y_{q,p}\right) + \arg\left(\tilde{y}_{q,p}\right) - 2\arg\left(u_{q,p}\right)| \leq \pi/3
\end{equation}
On other other hand, if we guessed incorrectly and $s_{q,p} = - r_{q,p}$, then
\begin{equation}
s_{q,p} = -r_{q,p} \ \implies \ |\arg\left(y_{q,p}\right) + \arg\left(\tilde{y}_{q,p}\right) - 2\arg\left(u_{q,p}\right)| \geq 2\pi/3
\end{equation}
Thus to determine $r_{q,p}$, it is sufficient to estimate $\arg\left(y_{q,p}\right) + \arg\left(\tilde{y}_{q,p}\right) = \arg\left(y_{q,p} \tilde{y}_{q,p}\right)$.

Now consider the product $y_{q,p} \tilde{y}_{q,p}$ for $(q,p) \in \mathcal{T}$. By triangle inequalities,
\begin{align}
|y_{q,p} \tilde{y}_{q,p}| &= |y_{q,p}| \cdot |\tilde{y}_{q,p}| \\
&\geq (|\hat{u}_{q,p}| - |y_{q,p} - r_{q,p} \hat{u}_{q,p}|) \cdot (|\hat{u}_{q,p}| - |\tilde{y}_{q,p} - s_{q,p} u_{q,p}|) \\
&\geq \left( \frac{\sqrt{2}}{\sqrt{3}} \varepsilon - \frac{\sqrt{2}}{2\sqrt{3}} \varepsilon \right)^2 = \frac{1}{6} \varepsilon^2
\end{align}

Suppose that in Step 3 we measure $\rho \otimes \tilde{\rho}^\ast$ in the basis $\{|\Phi_{a,b}\rangle\}$, and get the distribution $(a,b) \sim \mathcal{P}$. This gives an unbiased estimator of $y_{q,p} \tilde{y}_{q,p}$.
\begin{align}
y_{q,p} \tilde{y}_{q,p} &= \Tr\left(D_{q,p} \rho\right) \Tr\left(D_{q,p} \tilde{\rho}\right) \nonumber\\
&= \Tr\left(D_{q,p} \rho\right) \Tr\left(D_{-q,p} \tilde{\rho}^\ast\right) \nonumber\\
&= \Tr\left((D_{q,p} \otimes D_{-q,p}) (\rho \otimes \tilde{\rho}^\ast)\right) \nonumber\\
&= \mathbb{E}_{(a,b) \sim \mathcal{P}} \ \Tr\left((D_{q,p} \otimes D_{-q,p}) |\Phi_{a,b}\rangle\langle\Phi_{a,b}|\right) \nonumber\\
&= \mathbb{E}_{(a,b) \sim \mathcal{P}} \ \exp\left(i 2\pi (ap - bq) / d\right)
\end{align}
On the second line used Equation \ref{eq:D_transpose}. On the fourth line, we decohered in the simultaneous eigenbasis of the operators $\{D_{q,p} \otimes D_{-q,p}\}$, as in Algorithm \ref{alg:disp}. In the final line, we used Equation \ref{eq:DD_phi}.

Applying Hoeffding's inequality in the complex plane tells us that with $N = \mathcal{O}(\log{M} / \varepsilon^4)$ copies,
\begin{equation}
\big|\hat{v}_{q,p} - y_{q,p} \tilde{y}_{q,p}\big| < \frac{1}{12} \varepsilon^2
\end{equation}
for any $M$ displacement operators $D_{q,p}$ with high probability. Since $|y_{q,p} \tilde{y}_{q,p}| \geq \varepsilon^2 / 6$ for all $(q,p) \in \mathcal{T}$, this guarantees
\begin{equation}
|\arg\left(\hat{v}_{q,p}\right) - \arg\left(y_{q,p} \tilde{y}_{q,p}\right)| < \pi/6 \ \forall (q,p) \in \mathcal{T}
\end{equation}
and thus
\begin{align}
s_{q,p} = r_{q,p} \ \implies \ &|\arg\left(\hat{v}_{q,p}\right) - 2\arg\left(\hat{u}_{q,p}\right)| \\
&\leq |\arg\left(y_{q,p} \tilde{y}_{q,p}\right) - 2\arg\left(\hat{u}_{q,p}\right)| + |\arg\left(\hat{v}_{q,p}\right) - \arg\left(y_{q,p} \tilde{y}_{q,p}\right)| \\
&< \pi/3 + \pi/6 = \pi/2 \\
s_{q,p} = -r_{q,p} \ \implies \ &|\arg\left(\hat{v}_{q,p}\right) - 2\arg\left(\hat{u}_{q,p}\right)| \\
&\geq |\arg\left(y_{q,p} \tilde{y}_{q,p}\right) - 2\arg\left(\hat{u}_{q,p}\right)| - |\arg\left(\hat{v}_{q,p}\right) - \arg\left(y_{q,p} \tilde{y}_{q,p}\right)| \\
&> 2\pi/3 - \pi/6 = \pi/2
\end{align}
\end{proof}

\subsection{Finding hypothesis state} \label{sec:hypo_state}

\begin{algorithm}[H] \caption{Finding hypothesis state.} \label{alg:hypo_state}

\textbf{Input:}
\begin{itemize}
    \item A stream of observables $(E_1,E_2,\dots)$ on Hilbert space of dimension $d$.
    \item $\mathcal{O}(\log{d} / \varepsilon^4)$ copies of unknown state $\rho$.
    \item Estimates $\{\hat{u}_j\}$ satisfying $|r_j \hat{u}_j - \Tr\left(E_j \rho\right)| \leq \varepsilon / 2$ for some choices of signs $r_j \in \{+1,-1\}$.
\end{itemize}

\medskip
\textbf{Output:} A classical description of a density matrix $\tilde{\rho}$ satisfying
\begin{equation}
|\Tr\left(E_j \tilde{\rho}\right) - s_j \hat{u}_j| \leq \varepsilon \ \forall j
\end{equation}
for some choices of signs $s_j \in \{+1,-1\}$.

\medskip
\textbf{Algorithm:}
\begin{enumerate}
    \item Set $T := \lceil 16 \log{d} / \varepsilon^2 \rceil$ and $\beta := \sqrt{\log{d} / T}$.
    \item Initialize $\omega^{(0)} = \mathbbm{1}/d$ and $t = 0$.
    \item For each $j$,
        \subitem Compute $\tilde{y}^{(t)}_j = \Tr\left(E_j \omega^{(t)}\right)$.
        \subitem If $|\tilde{y}^{(t)}_j - \hat{u}_j| \geq \varepsilon$ \emph{and} $|\tilde{y}^{(t)}_j - \hat{u}_j| \geq \varepsilon$, declare ERROR and do:
            \subsubitem Measure $\mathcal{O}(1/\varepsilon^2)$ copies of $\rho$ to measure the sign $r_j \in \{+1,-1\}$ with high probability.
            \subsubitem Set
            \begin{equation}
            M^{(t)} = \frac{1}{2|\tilde{y}_j - r_j \hat{u}_j|} \big((\tilde{y}_j^\ast - r_j \hat{u}_j^\ast) E_j + (\tilde{y}_j - r_j \hat{u}_j) E_j^\dag\big).
            \end{equation}
            \subsubitem Set
            \begin{equation}
            \omega^{(t+1)} := \frac{\exp\left(- \beta \sum_{\tau=1}^t M^{(\tau)}\right)}{\Tr\left(\exp\left(- \beta \sum_{\tau=1}^t M^{(\tau)}\right)\right)}
            \end{equation}
            \subsubitem $t \leftarrow t+1$.
    \item Output $\tilde{\rho} = \omega^{(t)}$.
\end{enumerate}
\end{algorithm}

\begin{remark}
The runtime of Algorithm \ref{alg:hypo_state} is polynomial in the dimension $d$ of the Hilbert space and in $\varepsilon^{-1}$.
\end{remark}

\begin{proposition} \label{prop:MMW}
\emph{(\cite{arora2007combinatorial} Theorem 3.1)}
The choice
\begin{equation}
\omega^{(t+1)} := \frac{\exp\left(- \beta \sum_{\tau=1}^t M^{(\tau)}\right)}{\Tr\left(\exp\left(- \beta \sum_{\tau=1}^t M^{(\tau)}\right)\right)}
\end{equation}
satisfies
\begin{equation}
\sum_{t=1}^T \Tr\left(M^{(t)} \omega^{(t)}\right) - \lambda_{\min} \left(\sum_{t=1}^T M^{(t)}\right) \ \leq \ \beta \sum_{t=1}^T \Tr\left((M^{(t)})^2 \omega^{(t)}\right) + \frac{n}{\beta}
\end{equation}
\end{proposition}
\begin{proof}
Omitted.
\end{proof}

The choice of $\omega^{(t)}$ in Proposition \ref{prop:MMW} is known as the \emph{matrix multiplicative weights} algorithm, and the left hand side is often referred to as the \emph{regret}. If the matrices are normalized as $||M^{(t)}||_{\op} \leq 1$, we can write the regret bound as
\begin{equation}
\sum_{t=1}^T \Tr\left(M^{(t)} \rho^{(t)}\right) \leq \lambda_{\min}\left(\sum_{t=1}^T M^{(t)}\right) + \beta T + \frac{n}{\beta}
\end{equation}
A common choice is $\beta = \sqrt{n / T}$. In this case we get the bound
\begin{equation}
\sum_{t=1}^T \Tr\left(M^{(t)} \rho^{(t)}\right) \leq \lambda_{\min}\left(\sum_{t=1}^T M^{(t)}\right) + 2 \sqrt{nT}
\end{equation}

\begin{theorem} \label{thm:MMW}
\emph{(\cite{aaronson2018online} Theorem 1, adapted)}
Algorithm \ref{alg:hypo_state} makes at most $T = \lceil 16 \log{d} / \varepsilon^2 \rceil$ errors, regardless of the number of observables $(E_1,E_2,\dots)$ presented.
\end{theorem}
\begin{proof}
Let error number $t$ occur on $E_{j^{(t)}}$. View $T$ as unknown, but nevertheless insist on the relation $\beta = \sqrt{n / T}$. By the regret bound in Proposition \ref{prop:MMW}, we have
\begin{align}
&\sum_{t=1}^T \re\left(\frac{\tilde{y}_{j^{(t)}}^\ast - r_{j^{(t)}} \hat{u}_{j^{(t)}}^\ast}{|\tilde{y}_{j^{(t)}} - r_{j^{(t)}} \hat{u}_{j^{(t)}}|} \tilde{y}_{j^{(t)}}\right) \\
&\leq \lambda_{\min}\left(\sum_{t=1}^T \frac{\tilde{y}_{j^{(t)}}^\ast - r_{j^{(t)}} \hat{u}_{j^{(t)}}^\ast}{2|\tilde{y}_{j^{(t)}} - r_{j^{(t)}} \hat{u}_{j^{(t)}}|} E_{j^{(t)}} + \text{h.c.} \right) + 2 \sqrt{nT} \\
&\leq \sum_{t=1}^T \re\left(\frac{\tilde{y}_{j^{(t)}}^\ast - r_{j^{(t)}} \hat{u}_{j^{(t)}}^\ast}{|\tilde{y}_{j^{(t)}} - r_{j^{(t)}} \hat{u}_{j^{(t)}}|} \Tr\left(E_{j^{(t)}} \rho\right)\right) + 2 \sqrt{nT}
\end{align}
To get the last equation, we substituted in the particular quantum state $\rho$.

Rewriting, we get
\begin{equation} \label{eq:regret_implication}
\sum_{t=1}^T \re\left( \frac{\tilde{y}_{j^{(t)}}^\ast - r_{j^{(t)}} \hat{u}_{j^{(t)}}^\ast}{|\tilde{y}_{j^{(t)}} - r_{j^{(t)}} \hat{u}_{j^{(t)}}|} \Big(\tilde{y}_{j^{(t)}} - \Tr\left(E_{j^{(t)}} \rho\right)\Big) \right) \leq 2 \sqrt{nT}
\end{equation}

Now since $E_{j^{(t)}}$ was an error, we have $|\tilde{y}_{j^{(t)}} - r_{j^{(t)}} \hat{u}_{j^{(t)}}| \geq \varepsilon$. But also $|r_j \hat{u}_j - \Tr\left(E_j \rho\right)| \leq \varepsilon / 2$. From these, we can deduce the following
\begin{align}
&|\arg\left(\tilde{y}_{j^{(t)}} - r_{j^{(t)}} \hat{u}_{j^{(t)}}\right) - \arg\left(\tilde{y}_{j^{(t)}} - \Tr\left(E_{j^{(t)}} \rho\right)\right)| \leq \pi/6 \\
\implies \ &(\tilde{y}_{j^{(t)}}^\ast - r_{j^{(t)}} \hat{u}_{j^{(t)}}^\ast) (\tilde{y}_{j^{(t)}} - \Tr\left(E_{j^{(t)}} \rho\right)) \geq (1 - 1/\sqrt{2}) \cdot |\tilde{y}_{j^{(t)}} - r_{j^{(t)}} \hat{u}_{j^{(t)}}| \cdot |\tilde{y}_{j^{(t)}} - \Tr\left(E_{j^{(t)}} \rho\right)| \label{eq:online1}\\
& \left|\tilde{y}_{j^{(t)}} - \Tr\left(E_{j^{(t)}} \rho\right)\right| \geq \frac{\varepsilon}{2} \label{eq:online2}
\end{align}

Equations \ref{eq:regret_implication} and \ref{eq:online1} lets us deduce
\begin{equation}
\sum_{t=1}^T \left|\Tr\left(E_{j^{(t)}} \omega^{(t)}\right) - \Tr\left(E_{j^{(t)}} \rho\right)\right| \leq 2 \sqrt{nT}
\end{equation}
and finally Equation \ref{eq:online2} gives
\begin{align}
T \cdot \frac{\varepsilon}{2} &\leq 2 \sqrt{nT} \\
\implies \ T &\leq \frac{16 n}{\varepsilon^2}
\end{align}
\end{proof}

\section{Natural States}
\label{app:natural_states}
A key question is where one may encounter pairs of states $\rho$ and $\rho^*$.  Here we discuss potential sources of $\rho$ and $\rho^*$ and some of the challenges and opportunities in different application spaces.  In addition to general systems, we highlight how they relate to types of quantum sensors that have been proposed for different applications.  A key element of this connection is the conceptual connection to the generalized operators we focus on in this text and the exponential overhead in connecting standard Pauli operators to these more natural choices.

\subsection{The power of $\rho^*$ and hardness of producing $\rho^*$ from copies of $\rho$}
\label{app:rho_star_hardness}
The operation of complex conjugation of the entries of a density matrix $\rho$, denoted here by $\rho^*$ is not strictly a physical operation, in part, due to its basis dependence.  As physics should not depend on the basis one represents a problem in, this makes it clear the operation is unphysical.  While this basis dependence does not prevent it from being a useful operation, this non-physicality suggests that it may be difficult in general to produce $\rho^*$ from only copies of an unknown $\rho$, and that there is no efficient quantum operation for the general task.  Here we briefly review existing arguments why this is the case from a computational perspective.

Two input models of interest from the computational perspective are when one has access to copies of the state $\rho$ with no additional information about the state, and when one has blackbox access to a unitary $U$ that prepares a general $\rho$ without access to the source code, or description of the quantum circuit that produces $U$.  Considering first the input model where one may receive copies of $\rho$ for unit cost, it was shown in the proof of Theorem 1 of Ref. ~\cite{cotler2021revisiting}, that any quantum algorithm that can locate the random all-real-entry state among a collection of generally random states requires $\Omega(2^{n/2}) = \Omega(\sqrt{d})$ copies of $\rho$ where $n$ is the number of qubits in the constituent states.  If one could produce copies of $\rho^*$ from a polynomial number of copies of $\rho$, this would allow one to check whether a vector was all real to precision $\epsilon$ in time $O(c(n)/\eps^2)$, where $c(n)$ is the theoretical polynomial cost of complex conjugation on the unknown state $\rho$, and allow solution of the real-vector search problem in a time that is only linear in the number of vectors on top of this.  This is in clear contradiction to the proven lower bound of Ref.~\cite{cotler2021revisiting} and hence polynomial-time conjugation is not possible in the totally general case given only copies of an unknown $\rho$.

In the case of having black box access to $U$ that prepares $\rho$ from a reference state, but not its source code, it was shown in Ref.~\cite{ebler2023optimal} that the optimal procedure to act the complex conjugated unitary map $U^*$ is given by $d-1$ calls to the unitary black box $U$, where $d$ is the dimension of the unitary.  A procedure for producing $\rho^*$ for any given reference state that maps to $\rho$ would be equivalent to implementing $U^*$, and hence a procedure that can do this in less than $d-1=2^n-1$ calls to $U$ for a qubit system with $n$ qubits would be in violation of the lower bound.  As a result, for the most general unknown $\rho$, it is inefficient to produce $\rho^*$ with either blackbox access to copies of $\rho$ or a unitary that prepares $\rho$ from a particular reference state.  In spite of this general limitation, there are many useful cases for which production of $\rho^*$ is efficient, a subset of which we review in subsequent sections.

\subsection{Naturally real density matrices}
As a starting point, a wide class of useful states for which $\rho$ and $\rho^*$ are available, is the somewhat trivial case of a totally real density matrix, where naturally $\rho=\rho^*$.  While this makes the creation and use of $\rho^*$ something of a formality, it helps to highlight the cases where these techniques are applicable in practice.

Consider, for example, any real Hamiltonian $H$.  Real Hamiltonians encompass essentially all of chemistry without magnetic fields~\cite{helgaker2013molecular}, all $XZ$-type spin Hamiltonians, and many real-space Hamiltonians built from spatially local pieces without magnetic fields.  One central object of study for these systems is the density matrix at thermal equilibrium for some temperature $T=1/\beta$, or the Gibbs state, which is given by
\begin{align}
    \rho_\beta = \exp(\beta H) / \mathcal{Z}
\end{align}
where $\mathcal{Z}$ is a real-valued normalization constant.  As the Hamiltonian is real-valued, the Gibbs state is clearly also manifestly real by properties of the exponential function. 
Much effort has been devoted to the creation of these states on quantum computers~\cite{poulin2009sampling,riera2012thermalization,chowdhury2016quantum}, which is often quite expensive.  Hence efficient use of these resources can be quite important in both theory and practice.  This raises the question of when the displacement operators may be the subject of interest in physical systems, and a natural choice is real space systems, i.e.
\begin{align}
    H = -\sum_j \nabla_j^2 / m_j + V(x_1, ..., x_n).
\end{align}
where $x_i$ is a real-valued coordinate and the kinetic energy operator $-\sum_j \nabla_j^2 / m_j$ is taken with respect to a Euclidean geometry along each of the coordinates $x_i$.  This choice of a Euclidean geometry is what marks the displacement operators as the most relevant, as now $X$ will correspond naturally to evolution under the kinetic energy and $Z$ will correspond naturally to evolution under the potential term $V(x_1,...x_n)$.  While implementing this on a quantum computer requires an appropriate discretization step that respects the geometry of the original model, these types of discretization are not uncommon in first quantized simulations of chemistry, which are among some of the most efficient for very large systems on a quantum computer~\cite{su2021fault}.  Indeed, first quantized dynamics are so efficient, that performing them exactly on a quantum computer can be more efficient than comparable mean-field dynamics on a classical computer, and efficient classical shadows methods have been developed in this setting for understanding chemistry in the context of real-space bases~\cite{babbush2023_chem}. 

Another class of real-valued density matrices are those composed of mixtures of eigenstates of real-valued Hamiltonians, superpositions of degenerate eigenstates with real-valued relative phases, their time evolutions, and mixtures therein.  This is because a real-valued Hamiltonian can always be expressed in a real-valued eigenbasis. Even under time evolution, no relative phase accumulates as these states are in equilibrium such that $\partial_t \rho = 0$.  Given some real eigenbasis $\{E_i, \ket{E_i}\}$, this is easy to see from
\begin{align}
    \rho &= \sum_j p_j \ket{E_j} \bra{E_j} \\
    \partial_t \rho &= i[H, \rho] = i E_j \rho - i \rho E_j = 0.
\end{align}
This is important to note, as particular eigenstates, often low energy, are the specific subject of study in many quantum simulation experiments.  Gibbs states are a special case of this more general property, and both play a central role in many areas of chemistry and physics.  Hence for these studies we always expect to have access to $\rho$ and $\rho^*$ through the trivial relation $\rho=\rho^*$, and any benefits from the techniques in this work apply naturally to all of these systems.

\subsection{Time evolution from a real-valued initial state}
A key application of interest in quantum computing is quantum simulation of quantum systems.  The state generated by the time evolution of a physical system can be quite complex and contain interesting physics one wishes to study.  In these cases, even time-independent evolution can be quite interesting and challenging to simulate on a classical computer.  While at a glance this model may seem to be equivalent to having the source code for a quantum evolution, casting it as a physical time evolution allows one to consider schemes for approximate time reversal as an alternative to exact access to source code as has been proposed in applications to NMR and quantum learning~\cite{schuster2023learning}.

In the absence of magnetic fields or other time reversal symmetry breaking terms, it is always possible to choose a basis such that the eigenstates are real-valued.  In such a case, any initial state that is real-valued can be evolved under the real Hamiltonian $H$ as
\begin{align}
    \ket{\psi(t)} = \exp(-i H t) \ket{\psi(0)}.
\end{align}
In such a case, we often have a prescription for either evolving backwards in time $t \rightarrow -t$ or equivalently flipping the signs in the Hamiltonian $H \rightarrow -H$ such that
\begin{align}
    \ket{\psi(-t)} = \exp(i H t) \ket{\psi(0)} = \ket{\psi(t)}^*.
\end{align}
Hence for any time-reversal symmetric Hamiltonian, the ability to reverse time is equivalent to creating the complex conjugate state $\rho^*$, and copies may be produced as needed when the ability to create the initial state and perform the reverse time evolution is available upon demand.  

\subsection{Sensor arrays}
\label{app:sensor_arrays}
Another broad class of quantum sensors naturally related to the displacement operators are quantum sensor arrays, with a specific emphasis on low intensities.  Arrays of sensors in this regime have been discussed for enabling very long baseline optical interferometry for dramatically improved resolution in astronomical imaging~\cite{gottesman2012longer,khabiboulline2019optical,bland2021quantum}.  These sensor arrays are sometimes proposed to be composed of neutral atoms that are tunable across a wide range of frequencies from MHz to THz~\cite{osterwalder1999using,sedlacek2012microwave,holloway2017atom,wade2017real,simons2021rydberg,artusio2022modern} as well as optical single photons~\cite{tresp2016single}, and recent progress in tunable neutral atom arrays may open up other applications with such setups~\cite{bluvstein2023logical}.  Here we will introduce and explore this category of source in the context of conjugate quantum memory.

In the low intensity regime, incoming light is essentially a weak coherent state, but with an average photon number much less than 1.  This means the probability of two-photon and higher events is quite small, and we may assume there is a single incident photon to consider at all locations of the array.  To build intuition, we consider a simplified but reasonable situation physical situation where we obtain $\rho$ and $\rho^*$.

Consider an array of sensors with two accessible levels and an incoming photon with wavevector $k$ coming from a fixed source that repeatedly generates identical photons at some fixed time interval.  This array may be distributed in essentially any configuration, but the simplest and most applicable to the displacement operator setup is a Euclidean 1D array of sensors with a fixed length $L$ and number of equally spaced sensors $n^x$.  Current atomic traps allow neutral atoms to be prepared in high fidelity identical states with long-memory times~\cite{bluvstein2023logical}, and as such as an initial simplifying assumption, we take the detuning with respect to the incident photon to be uniform and the decoherence to be negligible.  Moreover, we consider a spatial extent of light much larger than any array under consideration such that any position dependence of the absorption amplitude is approximately uniform as in Ref.~\cite{gottesman2012longer}.  

Fourier analysis on the uniform grid provides us some insight on the resolution and limitations of measurements we would take.  In particular, the resolution of the wavevector $\Delta k$ is determined by the reciprocal of the physical length of the array itself, that is
\begin{align}
    \Delta k = \frac{1}{L}
\end{align}
which is directly related to imaging resolution in interferometry, and the origin of the desire to create very-long baseline interferometers enabled by quantum technology~\cite{gottesman2012longer}.  In addition, the maximum observable (unaliased) wavevector is determined by the Nyquist-Shannon sampling theorem~\cite{shannon1949communication} derived in this case from the number of points in the interval, or
\begin{align}
    |k| < \frac{1}{2 \Delta x} = \frac{L}{2 n_x}.
\end{align}

For the case of a weak interaction, one may increase the number of absorbers to match complete absorption, or alternatively post-select results on the case that at least 1 photon is present in any complete setup and have the state
\begin{align}
    \ket{\psi} \propto \sum_{x=1}^{n^x - 1} e^{-2 \pi i (k x) / L} \ket{x}
\end{align}
where $\ket{x}$ is the state with the $x$'th absorber excited and the rest 0 and the associated pure state density matrix is $\rho=\ket{\psi}\bra{\psi}$.  

This simple example provides a platform to understand the conceptual merit of working directly with displacement operators as opposed to qubit Pauli operators.  A general application area we are interested in is the use of quantum computers to do minimal processing directly on quantum states (with or without quantum memory).  As such one might expect an (imperfect) transduction and encoding step into $O(\lceil \log n^x \rceil )$ logical qubits for compactness.  The step of encoding the above $W$-like state and manipulations on the bit representation can be done via tools like unary iteration without significant overhead.  Once in this encoding, the effective memory time and number of processing operations without concern about memory or imperfect gate fidelity extends greatly.  It is also possible at that juncture to choose between something like measuring Pauli operators on tensor products $\log n_x $ qubits, or directly measuring the displacement operators living in dimension $n_x$.  

If we consider the cases where our natural operations are measurements of $k-$local operators (in the qubit or qudit sense), the local qubit operators on the above state have a somewhat unnatural representation.  A $Z$ measurement on the first qubit would correspond to occupation in the left or right half of the line, whereas a $Z$ measurement on the last qubit would indicate occupation on evenly or oddly labeled sites within the lattice.  For the $Z$-like displacement operators, one may measure all the qubits transversally and infer the positional value for the operator $D_{0, 1} = Z_{n_x}$.  This draws a stark contrast with the $X$ local qubit operators and the $X$-like displacement operators, such as $D_{1,0}=X_{n_x}$.

In contrast, any local $X$ measurement (obtained via local Hadamard transformations) on select qubits would yield a result in a local momentum basis that reflects the geometry of a hypercube rather than a Euclidean line.  These two coincide only for the unique common ground state of the two operators that is of uniform phase.  Indeed perhaps the biggest conceptual difference between the qubit and qudit pictures of the same line is the mismatch in the geometry implied by the qubit $X$ operators reflecting an implied hypercube geometry and the displacement operators $D_{1,0}=X_{n_x}$ reflecting a Euclidean geometry.  If it were, in general possible to infer the result of measurements in the qudit picture after a QFT reading only transversal measurements after a Hadamard transversal operation efficiently, this would replace the need for many QFT's found in phase estimation algorithms, for example.

The above state is naturally an eigenstate of the displacement operator $D_{1, 0} = X_{n^x}$ and its measurement yields up to $\log n^x$ bits of precision the very physical interpretation as the photon's wavevector along the $x$ spatial coordinate.  In contrast, measuring the local qubit $X$ operators yields a mixture of corresponding eigenstates that cannot be efficiently reassembled into the wavevector precisely without a number of measurements scaling like $n_x$ in the worst case.

So far we have not invoked the advantage of using quantum memory and $\rho^*$ for these states, but merely motivated the importance and naturalness of displacement operators in this context.  To see how $\rho^*$ is available in this simplified system where we assume identical photons are received at regular intervals, consider building a second copy of the sensor array inside a mirrored cavity with a known change in the pathlength of the incoming photon along the array axis $b$, such that the positions of the new detectors are shifted by the value $b$.  Using two mirrors, we may reflect the wavevector $k$ about the array axis from $k \rightarrow -k$ yielding
\begin{align}
    \ket{\psi'} \propto \sum_{x=1}^{n^x - 1} e^{2 \pi i (k (x-b)) / L} \ket{x} = e^{-2 \pi i k b / L} \sum_{x=1}^{n^x - 1} e^{2 \pi i (k x)/ L} \ket{x} = e^{-2 \pi i k b} \ket{\psi}^*.
\end{align}
In this case, we are able to obtain the complex conjugate state somewhat easily while picking up only a global phase.  While the assumptions that led to this were idealized in several ways, for more realistic cases it may be possible to use a similar procedure to build an approximation of the state that still allows one to retain advantage.

Assuming we build the conjugate state or an approximation of it, we now turn to the utility of the state in the sensor array context.  A common problem in sensor arrays, especially those tuned to low energy excitations like radio waves in thermal environments is a heavy amount of background thermal noise.  For example in the case of NMR, the background noise is often modeled as a nearly infinite temperature, or a totally mixed state.  For the most general type of signal when $\varepsilon \ll 1$, determining if any signal is present is equivalent to the mixedness testing problem, which can require a number of samples scaling like the dimension of the space in the worst case even with access to many copies at once~\cite{chen2021hierarchy}.  One way of viewing the advantage from the measurements in this scenario is a specialized form of mixedness testing, which permits distinguishing a sensor state from the maximally mixed state with an efficient number of samples, but only over a specialized class of states.

The natural questions are then what types of sensor state can we distinguish from noise using the conjugate memory resource, and how natural are they in physical systems?  The states we prove explicit lower bounds for while being sample efficient to distinguish from the maximally mixed state with conjugate memory are of the form
\begin{align}
    \rho = \rho_{qp} = \frac{1}{d} \big(\mathbbm{1} + r \varepsilon E_{q,p} \big).
\end{align}
By extension, we can also sample efficiently distinguish randomly states of the form 
\begin{align}
    \rho = \frac{1}{d} \big(\mathbbm{1} + \sum_{q, p} \alpha_{qp} E_{q,p} \big)
\end{align}
from the maximally mixed state so long as we sample according to the maximum magnitude of $\alpha_{pq}$, and we conjecture similar lower bounds apply for single copies considering the first case as a special case of the second so long as $\varepsilon \ll 1$ and the number of such operators is relatively small.

To gain intuition, consider first the operator $D_{0, 1} = Z$.  The corresponding state $\rho_{qp}$ on a 1D sensor array is a mixed state in the position basis, with probabilities of detection at each site given by $(1 + r \varepsilon \cos (2 \pi x/ L)) / d$.  Similarly, if we consider $D_{0, 2}=Z^2$, we have site populations $(1 + r \varepsilon \cos (2 \pi x/ L))^2 / d$ and more generally $D_{0, q}$ gives site populations $(1 + r \varepsilon \cos (2 \pi x/ L))^q / d$.  Generalizing this to the flexible sum above but still only in the position basis, we see we can have site populations given by trigonometric polynomials.  Hence it can distinguish very general population distribution patterns, like a Gaussian distribution with heavy background noise, from the totally random mixed state, but efficiency gain will be determined by the shape of the function.  By symmetry, the same is true of momentum space distributions.  However, there cannot be an advantage for states guaranteed to be probabilistic mixtures of only position or momentum eigenstates as efficient sampling within the two bases separately is possible, there must be non-trivial coherence in the states that manifests as a randomly oriented basis with $p, q \in [0, n_x - 1]$.  A combination of these can be a quite general expression of a physical state, like an incoherent distribution of Gaussian wavepackets on a real-space grid with heavy background noise.   

As one makes more general population distributions, one has to consider when the technique remains sample efficient and a gap to single copy techniques persist.  We leave precise derivations of these parameter regimes as future work, but in the $\varepsilon \ll 1$ regime the techniques can be viewed as an exponential improvement in signal to noise ratios per sample and an exponential improvement in the detection of any signal at all in analogy to specialized mixedness testing.  Multi-copy mixedness testing has also been considered to form a hierarchy for different numbers of copies as well~\cite{chen2021hierarchy}. In the $\varepsilon \approx 1$ regime, the question of advantage may be harder, as potentially one can construct a nearly pure state where other sample efficient techniques like classical shadows may apply.  We leave this interesting more general construction as an open question in applications of the techniques here.

\subsection{Access to $U^*$ via access to $U^\dagger$ on high temperature states}
Related to the above case of Gibbs states, another interesting model where one can access $U^*$ is when one has access to $U^\dagger$.  We note that access only to $U^\dagger$ does not generally permit efficient access to $U^*$, but here we examine a special case. The maximally entangled state produced by a simple Bell basis change
\begin{align}
    \ket{\Psi} = \frac{1}{\sqrt{2^n}} \sum_{i=0}^{2^n-1} \ket{i}\ket{i}
\end{align}
is the so-called thermofield double of an infinite temperature system on $n$ qubits.  For this state, the following identity is known to be true for any operators $V$, $W$
\begin{align}
    (V \otimes W) \ket{\Psi} = (V W^{T} \otimes I) \ket{\Psi} 
\end{align}
where $W^T$ is the transpose of $W$.  Hence by choosing $V=I$ and $W=U^\dagger$, which physically means accessing $U^\dagger$ once on only one of the registers and doing nothing on the other after preparing the entangled state, we have implemented
\begin{align}
    U^* \otimes I \ket{\Psi}
\end{align}
which is equivalent to the action of $U^*$ on the infinite temperature state when the second register is traced out or ignored.  On its own, this does nothing of note due to the infinite temperature state being invariant under unitary transformations, but if the application of $U^\dagger$ is controlled on ancilla qubits, it can be used to project into the basis with its conjugate basis, as well as other more general tasks.  As a special case, if a single ancilla qubit is prepared in the $\ket{+}$ state, this is related to the one-clean qubit model of computation, also known as DQC1~\cite{knill1998power}.  This at least contains DQC1 and we believe there may be interesting applications within this resource context.  

While restricted compared to general polynomial time quantum computation, the DQC1 model is known to have various applications including NMR structure prediction in chemistry~\cite{o2022quantum}, but we are not aware of the explicit use of $U^*$ as a resource within this model to date.  Finally we note that the application of $U$ to two-parts of the above Bell state, followed by Bell measurements, has been used to efficiently distinguish between a random real and general random unitary matrix, which is closely related to the identities discussed here~\cite{aharonov2022quantum,huang2022quantum,chen2022exponential}.

\section{Commutation properties of Heisenberg group} \label{sec:commutation}

Almost all of the techniques to make practical use of minimal quantum memories so far take advantage of a trick to make non-commuting operators into commuting ones that offer partial information about the original operators.  Here we will show that the displacement operators we consider in this paper are the most general class of operators for which this trick works, and hence one must go beyond this method to learn about other classes of operators with a minimal quantum memory.

Recall that Algorithm \ref{alg:disp} to learn the displacement amplitudes up to a sign relied on the crucial property that the following set of operators commute on two registers:
\begin{equation}
\{ D_{q,p} \otimes D_{-q,p} \}
\end{equation}
This was because braiding two displacement operators picks up a complex phase:
\begin{equation} \label{braid1}
D_{q',p'} D_{q,p} = e^{i 2\pi (qp' - q'p) / d} D_{q,p} D_{q',p'}
\end{equation}
and the second register picks up the conjugate phase:
\begin{equation} \label{braid2}
D_{-q',p'} D_{-q,p} = e^{-i 2\pi (qp' - q'p) / d} D_{-q,p} D_{-q',p'}
\end{equation}
so that they cancel and
\begin{equation}
(D_{q,p} \otimes D_{-q,p}) (D_{q',p'} \otimes D_{-q',p'}) = (D_{q',p'} \otimes D_{-q',p'}) (D_{q,p} \otimes D_{-q,p})
\end{equation}
Information about the state could be extracted from simultaneously measuring these commuting operators, since
\begin{equation}
\Tr\left((D_{q,p} \otimes D_{-q,p}) (\rho \otimes \rho^\ast)\right) = \Tr\left(D_{q,p} \rho\right)^2
\end{equation}

We now ask the question: how generally can this method be used? The following theorem gives evidence that the technique \emph{cannot} be applied more broadly than the setting of displacement operators.

\begin{theorem}
Let $U,V$ be unitaries of finite order $d$ on some Hilbert space $\mathcal{H}$. Suppose $U$ and $V$ do \emph{not} commute, but we can attach a second Hilbert space $\mathcal{H}'$ with unitaries $\tilde{U},\tilde{V}$ such that $U \otimes \tilde{U}$ and $V \otimes \tilde{V}$ commute. Then there is some unitary transformation of $\mathcal{H}$ mapping $U,V$ to a direct sum of displacement operators.
\end{theorem}
\begin{proof}
By assumption, $U$ and $V$ do \emph{not} commute, so
\begin{equation}
U V U^{-1} V^{-1} \neq \mathbbm{1}
\end{equation}
but $U \otimes \tilde{U}$ and $V \otimes \tilde{V}$ commute, so
\begin{equation}
(U V U^{-1} V^{-1}) \otimes (\tilde{U} \tilde{V} \tilde{U}^{-1} \tilde{V}^{-1}) = \mathbbm{1}.
\end{equation}
Only scalars can pass through tensor factors, so we must have
\begin{equation}
U V U^{-1} V^{-1} = \omega \mathbbm{1}
\end{equation}
for some phase $\omega \neq 1$. Now using that $V^d = \mathbbm{1}$,
\begin{align}
&\omega^d \mathbbm{1} = \left(\omega V\right)^d = \left(U V U^{-1}\right)^d = U V^d U^{-1} = \mathbbm{1} \\
\implies \ &\omega^d = 1.
\end{align}

Overall, we can write
\begin{equation}
U V U^{-1} V^{-1} = \omega \mathbbm{1} \quad , \quad U^d = V^d = (\omega \mathbbm{1})^d = \mathbbm{1}
\end{equation}
The discrete Heisenberg group $H(\mathbb{Z}_d)$ has presentation
\begin{equation}
H(\mathbb{Z}_d) = \{x, z, w \ : xzx^{-1}z^{-1} = w, \ x^d = z^d = w^d = 1 , \ wx = xw, \ wz = zw\}
\end{equation}
Thus $\{U,V,\omega \mathbbm{1}\}$ generate a representation of $H(\mathbb{Z}_d)$.

To complete the proof, we invoke the representation theory of $H(\mathbb{Z}_d)$. The irreducible representations (irreps) of $H(\mathbb{Z}_d)$ are classified in \cite{grassberger2001note} using character theory. They are indexed by three integers $(\gamma, \alpha, \beta) \in \mathbb{Z}_d \times \mathbb{Z}_h \times \mathbb{Z}_h$ where $h := \text{gcd}(\gamma,d)$. The dimension of $\text{irrep}(\gamma, \alpha, \beta)$ is $d' = d/h$. Define the clock and shift operators in $d'$ dimensions as usual
\begin{equation}
X :=
\begin{pmatrix}
0 & 0 & \dots & 0 & 1 \\
1 & 0 & \dots & 0 & 0 \\
\vdots & \vdots & \ddots & \vdots & \vdots \\
0 & 0 & \dots & 1 & 0 \\
\end{pmatrix}, \qquad
Z :=
\begin{pmatrix}
1 & 0 & \dots & 0 \\
0 & \omega & \dots & 0 \\
\vdots & \vdots & \ddots & \vdots \\
0 & 0 & \dots & \omega^{d-1} \\
\end{pmatrix},
\end{equation}
and $\omega := e^{2\pi i / d'}$. Then we can write $\text{irrep}(\gamma, \alpha, \beta)$ as
\begin{align}
\text{irrep}(\gamma, \alpha, \beta) \ : \ H(\mathbb{Z}_d) &\rightarrow GL(d') \\
w &\rightarrow \omega^\gamma \\
x &\rightarrow X^\alpha \\
z &\rightarrow Z^\beta
\end{align}

As we can see, all of the irreducible representations consist of displacement operators in some basis, thus decomposing our $H(\mathbb{Z}_d)$-representation generated by $\{U,V,\omega \mathbbm{1}\}$ into a direct sum over its irreducible components completes the proof.
\end{proof}

The representation theory of the discrete Heisenberg group $H(\mathbb{Z}_d)$ is reminiscent of that of the continuous Heisenberg group $H(\mathbb{R})$, where the Stone-von Neumann theorem tells us that there is a unique representation up to unitary equivalence \cite{rosenberg2004selective}.

\section{Simultaneous diagonalisability from Pauli braiding} \label{sec:adjoint}

So far, we have justified that $D_{q,p} \otimes D_{-q,p} = D_{q,p}\otimes D_{q,p}^T$ simply through explicit verification (see e.g., Eqs.~\eqref{braid1} and~\eqref{braid2}, and we exhibited a simultaneous eigenbasis, whose eigenvalues can also be found through explicit calculation (Eq.~\eqref{eq:DD_phi}). Similarly, the operators $D_{q,p} \otimes D_{q,-p} = D_{q,p} \otimes D_{q,p}^*$ also mutually commute. In this section, we aim to provide more foundation for these commutation relations and their associated eigenvalue-eigenvector equations, by showing that the simultaneously diagonalisability of $D_{q,p} \otimes D_{q,p}^*$ is in fact equivalent to the well-known property that Paulis leave other Paulis fixed under conjugation, up to a global phase, i.e., for any two Paulis $P$ and $Q$, 
\begin{equation} \label{PQbraid} P Q P^\dagger = \alpha Q\end{equation}
for some phase $\alpha$. This equivalence is given by a simple linear algebra maneuver known as operator-vector correspondence, which will moreover allow us to directly obtain the simultaneous eigenbasis and the corresponding eigenvalues.

Fixing a computational basis $\{\ket{i}\}$, operator-vector correspondence is given by the linear ``vectorisation'' map $\mathrm{vec}$ defined by
\[ \mathrm{vec}(\ket{i}\bra{j}) = \ket{i} \otimes \ket{j}. \] It follows from linearity that for arbitrary states $\ket{\psi}$ and $\ket{\varphi}$, we have
\[ \mathrm{vec}(\ket{\psi}\bra{\varphi}) = \ket{\psi}\otimes \ket{\varphi}^* = \ket{\psi}(\bra{\varphi})^T,\]
where the complex conjugate and transpose are taken with respect to the chosen computational basis. In particular,
\[ \mathrm{vec}(P\ket{i}\bra{j}P^\dagger) = P\ket{i}(\bra{j}P^\dagger)^T = P\ket{i} \otimes P^*\ket{j}. \]
Thus, when we conjugate some operator by $P$, under operator-vector correspondence the $P^\dagger$ on one side gets transposed and becomes a $P^*$.
So writing $Q \equiv \sum_{i,j \in \Zd} Q_{ij}\ket{i}\bra{j}$ in terms of basis elements and applying $\mathrm{vec}$ to both sides of the Pauli braiding equation Eq.~\eqref{PQbraid}, we obtain
\[ \sum_{i,j \in \Zd} Q_{ij} P\ket{i} \otimes P^*\ket{j} = \alpha \sum_{i,j \in \Zd} Q_{ij} \ket{i} \otimes \ket{j}, \]
i.e.,
\begin{equation} \label{ev vecQ} (P \otimes P^*)\mathrm{vec}(Q) = \alpha \, \mathrm{vec}(Q). \end{equation}
Therefore, Eq.~\eqref{PQbraid} is equivalent to the statement that $\mathrm{vec}(Q)$ is an eigenvector of $P \otimes P^*$, with eigenvalue $\alpha$. This holds for all Paulis $P$ and $Q$, and since the $\mathrm{vec}(Q)$ form an orthogonal basis of states (from the fact that Paulis comprise a Hilbert-Schmidt orthogonal basis for operators), the $P \otimes P^*$ are simultaneously diagonalisable.

The explicit version of Eq.~\eqref{PQbraid} is
\begin{equation} \label{Dqp braid} D_{q,p} D_{a,b} D_{q,p}^\dagger = \omega^{ap-qb} D_{a,b}  \end{equation}
(from Eq.~\eqref{eq:D_comm}), and we have
\[ D_{a,b} = \omega^{ab/2} \sum_{j\in\Zd}\omega^{jb}\ket{j + a} \bra{j} \]
(from Eq.~\eqref{Dqp act}). Hence, using Eq.~\eqref{ev vecQ}, which becomes $(D_{q,p}\otimes D_{q,p}^*)\mathrm{vec}(D_{a,b}) = \omega^{ap-qb}\mathrm{vec}(D_{a,b})$ with these substitutions, we immediately see that the simultaneous eigenbasis for $\{D_{q,p} \otimes D_{q,p}^*\}$ consists of the vectorisations of the Paulis $D_{a,b}$,
\begin{equation} \mathrm{vec}(D_{a,b}) \propto \sum_{j \in \Zd} \omega^{jb} \ket{j+a}\ket{j} \end{equation}
for $a, b \in \Zd$, and the corresponding eigenvalues of $D_{q,p}$ are $\omega^{ap - qb}$.

In representation-theoretic terms, the simultaneous diagonalisability of $D_{q,p} \otimes D_{q,p}^*$ is also equivalent to the statement that the representation $U \mapsto U \otimes U^*$ of the Heisenberg group consists solely of 1-dimensional irreps. However, it is nontrivial to directly build the irreps of this representation from $U$ and $U^*$.
Here, we are instead using operator-vector correspondence to relate the $U \mapsto U \otimes U^*$ representation to the adjoint representation $U \mapsto U(\, \cdot \,)U^\dagger$, because the latter is well-known to consist of 1-dimensional irreps (Eq.~\eqref{PQbraid}). 

To get the analogous result for $D_{q,p}\otimes D_{q,p}^T$, we use the fact that $SD_{q,p}S = D_{-q,-p} = D_{q,p}^\dagger$, where $S \coloneqq \sum_{j \in \Zd}\ket{-j}\bra{j}$. Then, Eq.~\eqref{Dqp braid} can be written $D_{q,p}D_{a,b}SD_{q,p}S = \omega^{ap-qb}D_{a,b}$, or $D_{q,p}D_{a,b}SD_{q,p} = \omega^{ap-qb}D_{a,b}S$, which vectorises to
\[ (D_{q,p} \otimes D_{q,p}^T) \mathrm{vec}(D_{a,b}S) = \omega^{ap-qb}D_{a,b}\mathrm{vec}(D_{a,b}S). \]
Thus, the simultaneous eigenbasis of $D_{q,p} \otimes D_{q,p}^T$ is given by 
\[ \mathrm{vec}(D_{a,b}S) \propto \sum_{j \in \Zd}\omega^{jb}\ket{j + a}\ket{-j}, \]
which are precisely the $\ket{\Phi_{a,b}}$ states defined in Eqs.~\eqref{Phi00} and~\eqref{Phiab}, and the corresponding eigenvalues of $D_{q,p}$ are $\omega^{ap-qb}$ as claimed in Eq.~\eqref{eq:DD_phi}.


\section{Sample complexity lower bounds} \label{sec:single_copy}

The aim of this section is to show that the use of entangled measurements across $\rho \otimes \rho^\ast$ are essential to our learning task. Any strategy which uses only single-copy measurements on $\rho$ and $\rho^\ast$, even an adaptive one, necessarily requires exponentially many copies. On the other hand, any strategy which uses only $\rho$, and does not have access to $\rho^\ast$, must necessarily make entangled measurements across $\Omega(1/\varepsilon)$ copies at a time.

The proofs will follow the techniques introduced in \cite{chen2022exponential, chen2021hierarchy}. Here we briefly state two lemmas adapted from \cite{chen2022exponential} which will be useful to us.

\begin{lemma} \label{lem:4.8}
\emph{(\cite{chen2022exponential} Lemma 4.8)} 
When we only consider the classical outcome of the POVM measurement and neglect the post-measurement quantum state, then any POVM can be simulated by a rank-1 POVM with some postprocessing. A rank-1 POVM is defined by $\{w_s, |\psi_s\rangle\}_s$ where $w_s > 0$ and
\begin{equation}
\sum_s w_s |\psi_s\rangle\langle\psi_s| = \mathbbm{1}
\end{equation}
The outcomes probabilities are
\begin{equation}
\mathbb{P}_\rho(s) = w_s \langle\psi_s|\rho|\psi_s\rangle
\end{equation}
\end{lemma}

This lemma lets us consider only rank-1 POVMs without loss of generality.

\begin{lemma} \label{lem:5.4}
\emph{(\cite{chen2022exponential} Lemma 5.4)} Suppose we have a many vs one distinguishing task
\begin{itemize}
    \item \emph{(YES)} $\rho = \rho_x$ for some random $x$ and some family of states $\{\rho_x\}$.
    \item \emph{(NO)} $\rho = \frac{1}{d} \mathbbm{1}$ maximally mixed.
\end{itemize}
If all outcomes $l$ of a protocol satisfy
\begin{equation} \label{eq:one_sided_condition}
\frac{\mathbb{E}_x{\mathbb{P}_{\rho_x}(l)}}{\mathbb{P}_{\mathbbm{1}/d}(l)} \geq 1 - \delta
\end{equation}
then the probability of success is at most $(1 + \delta) / 2$. Note that the outcome $l$ refers to the entire history of many measurements, which are possibly adaptive.
\end{lemma}

We will also use a set of operators defined in \cite{asadian2016heisenberg}. These are given by
\begin{equation} \label{eq:E_def}
E_{q,p} = \chi D_{q,p} + \chi^\ast D_{-q,-p} \quad , \quad \chi = \frac{1+i}{2}
\end{equation}
and are known as \emph{displacement observables}.
\begin{proposition} \label{prop:properties_obs} \emph{\cite{asadian2016heisenberg}}
The displacement observables are Hermitian and have properties
\begin{align}
&E^\ast_{q,p} = E^T_{q,p} = E_{-q,p} \label{eq:comp_conj_E} \\
&\|E_{q,p}\|_{\op} \leq \sqrt{2} \label{eq:norm_bd} \\
&\Tr\left(E_{q,p} E_{q',p'}\right) = d \cdot \delta_{q,q'} \delta_{p,p'}
\end{align}
\end{proposition}

\subsection{Single copies of $\rho$ and $\rho^\ast$} \label{sec:boson_lower_bd_1}

\begin{theorem} \label{thm:single_copy}
Let $d$ be the dimension of the Hilbert space. Any single-copy protocol which learns $|\Tr\left(D_{q,p} \rho\right)|$ to precision $\varepsilon$ for all $(q,p) \in \mathbb{Z}_d^2$ with probability $2/3$ requires $\Omega(d / \varepsilon^2)$ copies. This holds even if the protocol has access to single-copy measurements of both $\rho$ and $\rho^\ast$.
\end{theorem}

This will follow from Proposition \ref{prop:distinguishing_single_copy}, which concerns a certain distinguishing task.

\begin{proposition} \label{prop:distinguishing_single_copy}
Let $d$ be the dimension of the Hilbert space. Consider the task of distinguishing between the following two scenarios:
\begin{itemize}
    \item \emph{(YES)} $\rho = \rho_{q,p,r} = \frac{1}{d} \big(\mathbbm{1} + r \varepsilon E_{q,p} \big)$ for some uniformly random $(q,p) \in \mathbb{Z}_d^2 \setminus \{(0,0)\}$ and some uniformly random sign $r \in \{\pm 1\}$. (We assume $0 < \varepsilon < 1$.)
    \item \emph{(NO)} $\rho = \frac{1}{d} \mathbbm{1}$ maximally mixed.
\end{itemize}
Any single-copy protocol requires $\Omega(d / \varepsilon^2)$ copies in order to succeed with probability $2/3$. This holds even if the protocol has access to single-copy measurements of both $\rho$ and $\rho^\ast$.
\end{proposition}

\bigskip
\begin{proof}\emph{of Theorem \ref{thm:single_copy} using Proposition \ref{prop:distinguishing_single_copy}.}
Suppose protocol $\mathcal{A}$ is able to learn $|\Tr\left(D_{q,p} \rho\right)|$ to precision $\varepsilon$ for all $(q,p) \in \mathbb{Z}_d^2 \setminus \{(0,0)\}$ with probability $2/3$. Consider the task in Proposition \ref{prop:distinguishing_single_copy} with $\varepsilon$ replaced by $3 \sqrt{2} \varepsilon$. We will argue that $\mathcal{A}$ is able to succeed at this task with probability $2/3$.
\begin{align}
\Tr\left(D_{q,p} \rho_{q,p,r}\right) &= \frac{3 \sqrt{2} r \varepsilon}{d} \Tr\left(D_{q,p} E_{q,p}\right) \\
&= \begin{cases}
3 \sqrt{2} r \varepsilon & d \ \text{even} , \ q = p = d/2 \\
3 \sqrt{2} \chi^\ast r \varepsilon & \text{otherwise}
\end{cases} \\
\implies \ |\Tr\left(D_{q,p} \rho_{q,p,r}\right)| &\geq 3 \varepsilon
\end{align}
using Proposition \ref{prop:disp_basis}. Thus $\mathcal{A}$ can distinguish any $\rho_{q,p,r}$ from the maximally mixed state, which has
\begin{equation}
\Tr\left(D_{q,p} \frac{\mathbbm{1}}{d}\right) = 0 \ \forall q, p
\end{equation}
In this case, Proposition \ref{prop:distinguishing_single_copy} gives us a sample complexity lower bound of $\Omega(d / \varepsilon^2)$.
\end{proof}

\bigskip
\begin{proof}\emph{of Proposition \ref{prop:distinguishing_single_copy}.}
Suppose single-copy protocol $\mathcal{A}$ uses $T$ copies of $\rho$ or $\rho^\ast$, and at step $t$ applies the rank-1 POVM $\{w^t_s, |\psi^t_s\rangle\}_s$. By Lemma \ref{lem:4.8}, this is without loss of generality. Note the slight abuse of notation, since the POVM of later steps are allowed to depend on the outcomes of earlier measurements. Suppose the outcome of the measurements are $l = (s_1,\dots,s_T)$. We have
\begin{equation}
\mathbb{P}_\rho(l) = \prod_{t=1}^T w^t_{s_t} \langle\psi^t_{s_t}|\rho^t|\psi^t_{s_t}\rangle
\end{equation}
where $\rho^t$ is either $\rho$ or $\rho^\ast$ for each $t$. Note that
\begin{equation}
\rho_{q,p,r}^\ast = \frac{1}{d} \big(\mathbbm{1} + r \varepsilon E_{-q,p} \big)
\end{equation}
using Equation \ref{eq:comp_conj_E}. Let
\begin{equation}
\gamma_t = \begin{cases}
+1 & \rho^t = \rho \\
-1 & \rho^t = \rho^\ast
\end{cases}
\end{equation}

The aim is to establish an inequality like Equation \ref{eq:one_sided_condition}, so that we can invoke Lemma \ref{lem:5.4}
\begin{align}
\frac{\mathbb{E}_{(q,p) \neq (0,0)} \mathbb{E}_r {\mathbb{P}_{\rho_{q,p,r}}(l)}}{\mathbb{P}_{\mathbbm{1}/d}(l)} &= \mathbb{E}_{(q,p) \neq (0,0)} \mathbb{E}_r \prod_{t=1}^T \frac{w^t_{s_t} + r \varepsilon w^t_{s_t} \langle\psi^t_{s_t}|E_{\gamma_t q,p}|\psi^t_{s_t}\rangle}{w^t_{s_t}} \\
&= \mathbb{E}_{(q,p) \neq (0,0)} \mathbb{E}_r \exp\Big( \sum_{t=1}^T \log\big(1 + r \varepsilon \langle\psi^t_{s_t}|E_{\gamma_t q,p}|\psi^t_{s_t}\rangle\big) \Big) \\
&\geq \exp\Big( \sum_{t=1}^T \mathbb{E}_{(q,p) \neq (0,0)} \mathbb{E}_r \log\big(1 + r \varepsilon \langle\psi^t_{s_t}|E_{\gamma_t q,p}|\psi^t_{s_t}\rangle\big) \Big) \label{eq:Jensen1} \\
&= \exp\Big( \sum_{t=1}^T \frac{1}{2} \mathbb{E}_{(q,p) \neq (0,0)} \log\big(1 - \varepsilon^2 \langle\psi^t_{s_t}|E_{\gamma_t q,p}|\psi^t_{s_t}\rangle^2\big) \Big) \\
&\geq \exp\Big(- \sum_{t=1}^T \varepsilon^2 \mathbb{E}_{(q,p) \neq (0,0)} \langle\psi^t_{s_t}|E_{\gamma_t q,p}|\psi^t_{s_t}\rangle^2 \Big) \label{eq:log1} \\
&\geq \exp(- T \varepsilon^2 \Gamma) \\
&\geq 1 - T \varepsilon^2 \Gamma \label{eq:gamma_calculation1}
\end{align}
where
\begin{equation}
\Gamma = \sup_{|\psi\rangle} \ \mathbb{E}_{(q,p) \neq (0,0)} \langle\psi|E_{\pm q,p}|\psi\rangle^2
\end{equation}
In Equation \ref{eq:Jensen1} we used Jensen's inequality. In Equation \ref{eq:log1} we used $\log{1-x} \geq -2x \ \forall x \in [0,0.79]$, which is valid as long as $\varepsilon \leq 0.62$ by Equation \ref{eq:norm_bd}.

It remains to upper bound $\Gamma$. At this point, it is clear that we can drop the $\pm$ coming from the use of the conjugate state, since it is averaged over all $q,p$. We can express $\Gamma$ as
\begin{equation}
\Gamma = \sup_{|\psi\rangle} \ \langle\psi|\langle\psi| \mathbb{E}_{(q,p) \neq (0,0)} (E_{q,p} \otimes E_{q,p}) |\psi\rangle|\psi\rangle
\end{equation}
It can be checked using Equation \ref{eq:E_def} that
\begin{equation}
E_{q,p} \otimes E_{q,p} + E_{-q,-p} \otimes E_{-q,-p} = D_{q,p} \otimes D_{-q,-p} + D_{-q,-p} \otimes D_{q,p}
\end{equation}
and thus we can write
\begin{equation}
\mathbb{E}_{(q,p) \neq (0,0)} (E_{q,p} \otimes E_{q,p}) = \mathbb{E}_{(q,p) \neq (0,0)} (D_{q,p} \otimes D_{-q,-p})
\end{equation}
It can be checked that
\begin{equation}
\sum_{q,p} (D_{q,p} \otimes D_{-q,-p}) = d \cdot \text{SWAP}
\end{equation}
and so
\begin{equation}
\mathbb{E}_{(q,p) \neq (0,0)} (D_{q,p} \otimes D_{-q,-p}) = \frac{d}{d^2 - 1} \cdot \text{SWAP} - \frac{1}{d^2 - 1} \cdot \mathbbm{1}
\end{equation}
Putting this all together, we get
\begin{equation}
\Gamma = \frac{1}{d^2 - 1} \sup_{|\psi\rangle} \ \langle\psi|\langle\psi| (d \cdot \text{SWAP} - \mathbbm{1}) |\psi\rangle|\psi\rangle = \frac{d - 1}{d^2 - 1} = \frac{1}{d+1}
\end{equation}

Returning to Equation \ref{eq:gamma_calculation1}, we have
\begin{equation}
\frac{\mathbb{E}_{(q,p) \neq (0,0)} \mathbb{E}_r {\mathbb{P}_{\rho_{q,p,r}}(l)}}{\mathbb{P}_{\mathbbm{1}/d}(l)} \geq 1 - \frac{T \varepsilon^2}{d+1}
\end{equation}
By Lemma \ref{lem:5.4}, our single-copy protocol $\mathcal{A}$ succeeds with probability at most $(1 + T \varepsilon^2 / (d+1)) / 2$. This completes the proof of Proposition \ref{prop:distinguishing_single_copy}: to succeed with probability $2/3$, $\mathcal{A}$ requires $T = \Omega(d / \varepsilon^2)$.
\end{proof}

\subsection{Entangled measurements without $\rho^\ast$} \label{sec:boson_lower_bd_2}

\begin{theorem} \label{thm:no_conjugate}
Let $d$ be the dimension of the Hilbert space. Assume $d$ is prime. Any protocol which learns $|\Tr\left(D_{q,p} \rho\right)|$ to precision $\varepsilon$ for all $(q,p) \in \mathbb{Z}_d^2$ with probability $2/3$ by measuring copies of $\rho^{\otimes K}$ for $K \leq 1 / (12 \varepsilon)$ requires $\Omega(\sqrt{d} / (K^2 \varepsilon^2))$ measurements.
\end{theorem}

This will follow from Proposition \ref{prop:distinguishing_no_conjugate}. The proof of Proposition \ref{prop:distinguishing_no_conjugate} is similar to Appendix D of \cite{chen2021hierarchy}.

\begin{proposition} \label{prop:distinguishing_no_conjugate}
Let $d$ be the dimension of the Hilbert space. Assume $d$ is prime. Consider the task of distinguishing between the following two scenarios:
\begin{itemize}
    \item \emph{(YES)} $\rho = \rho_{q,p,r} = \frac{1}{d} \big(\mathbbm{1} + r \varepsilon E_{q,p} \big)$ for some uniformly random $(q,p) \in \mathbb{Z}_d^2 \setminus \{(0,0)\}$ and some uniformly random sign $r \in \{\pm 1\}$. (We assume $0 < \varepsilon < 1$.)
    \item \emph{(NO)} $\rho = \frac{1}{d} \mathbbm{1}$ maximally mixed.
\end{itemize}
Any protocol succeeding with probability $2/3$ which measures copies of $\rho^{\otimes K}$ for $K \leq 1 / (2 \sqrt{2} \varepsilon)$ requires $\Omega(\sqrt{d} / (K^2 \varepsilon^2))$ measurements.
\end{proposition}

\bigskip
\begin{proof}\emph{of Theorem \ref{thm:no_conjugate} using Proposition \ref{prop:distinguishing_no_conjugate}.}
Suppose protocol $\mathcal{A}$ is able to learn $|\Tr\left(D_{q,p} \rho\right)|$ to precision $\varepsilon$ for all $(q,p) \in \mathbb{Z}_d^2 \setminus \{(0,0)\}$ with probability $2/3$. Consider the task in Proposition \ref{prop:distinguishing_no_conjugate} with $\varepsilon$ replaced by $3 \sqrt{2} \varepsilon$. We will argue that $\mathcal{A}$ is able to succeed at this task with probability $2/3$.
\begin{align}
\Tr\left(D_{q,p} \rho_{q,p,r}\right) &= \frac{3 \sqrt{2} r \varepsilon}{d} \Tr\left(D_{q,p} E_{q,p}\right) \\
&= \begin{cases}
3 \sqrt{2} r \varepsilon & d \ \text{even} , \ q = p = d/2 \\
3 \sqrt{2} \chi^\ast r \varepsilon & \text{otherwise}
\end{cases} \\
\implies \ |\Tr\left(D_{q,p} \rho_{q,p,r}\right)| &\geq 3 \varepsilon
\end{align}
using Proposition \ref{prop:disp_basis}. Thus $\mathcal{A}$ can distinguish any $\rho_{q,p,r}$ from the maximally mixed state, which has
\begin{equation}
\Tr\left(D_{q,p} \frac{\mathbbm{1}}{d}\right) = 0 \ \forall q, p
\end{equation}
In this case, Proposition \ref{prop:distinguishing_no_conjugate} says that $\mathcal{A}$ requires either entangled measurements across at least $1 / 2 \sqrt{2} \cdot (3 \sqrt{2} \varepsilon) = 1 / (12 \varepsilon)$ copies of $\rho$ at a time, or $\Omega(\sqrt{d} / (K^2 \varepsilon^2))$ total copies.
\end{proof}

\bigskip
\begin{proof}\emph{of Proposition \ref{prop:distinguishing_no_conjugate}.}
Suppose protocol $\mathcal{A}$ measures $K$ copies of $\rho$ at a time, where $K \leq 1 / (2 \sqrt{2} \varepsilon)$. At step $t$, $\mathcal{A}$ applies the rank-1 POVM $\{w^t_s, |\psi^t_s\rangle\}_s$, where $t$ goes from $1$ up to $T$. By Lemma \ref{lem:4.8}, this is without loss of generality. Note the slight abuse of notation, since the POVM of later steps are allowed to depend on the outcomes of earlier measurements. Suppose the outcome of the measurements are $l = (s_1,\dots,s_T)$. We have
\begin{equation}
\mathbb{P}_\rho(l) = \prod_{t=1}^T w^t_{s_t} \langle\psi^t_{s_t}|\rho^{\otimes K}|\psi^t_{s_t}\rangle
\end{equation}

The aim is to establish an inequality like Equation \ref{eq:one_sided_condition}, so that we can invoke Lemma \ref{lem:5.4}
\begin{align}
\frac{\mathbb{E}_{(q,p) \neq (0,0)} \mathbb{E}_r {\mathbb{P}_{\rho_{q,p,r}}(l)}}{\mathbb{P}_{\mathbbm{1}/d}(l)} &= \mathbb{E}_{(q,p) \neq (0,0)} \mathbb{E}_r \prod_{t=1}^T \langle\psi^t_{s_t}| \big(\mathbbm{1} + r \varepsilon E_{q,p}\big)^{\otimes K} |\psi^t_{s_t}\rangle \\
&= \mathbb{E}_{(q,p) \neq (0,0)} \mathbb{E}_r \prod_{t=1}^T F^t_{q,p,r} \\
&= \mathbb{E}_{(q,p) \neq (0,0)} \mathbb{E}_r \exp\Big( \sum_{t=1}^T \log{F^t_{q,p,r}} \Big) \\
&\geq \exp\Big( \sum_{t=1}^T \mathbb{E}_{(q,p) \neq (0,0)} \mathbb{E}_r \log{F^t_{q,p,r}} \Big) \label{eq:Jensen2} \\
&= \exp\Big( \sum_{t=1}^T \frac{1}{2} \mathbb{E}_{(q,p) \neq (0,0)} \log\big(F^t_{q,p,+1} \cdot F^t_{q,p,-1}\big) \Big) \\
&= \exp\Big( \sum_{t=1}^T \frac{1}{2} \mathbb{E}_{(q,p) \neq (0,0)} \log\big(1 - G^t_{q,p}\big) \Big) \\
&\geq \exp\Big( \sum_{t=1}^T \frac{1}{2} \mathbb{E}_{(q,p) \neq (0,0)} \log\big(1 - \max(0,G^t_{q,p})\big) \Big) \\
&\geq \exp\Big(- \sum_{t=1}^T \mathbb{E}_{(q,p) \neq (0,0)} \max(0,G^t_{q,p}) \Big) \label{eq:log2} \\
&\geq 1 - \sum_{t=1}^T \mathbb{E}_{(q,p) \neq (0,0)} \max(0,G^t_{q,p}) \label{eq:gamma_calculation2}
\end{align}
where
\begin{align}
F^t_{q,p,r} &:= \langle\psi^t_{s_t}| \big(\mathbbm{1} + r \varepsilon E_{q,p}\big)^{\otimes K} |\psi^t_{s_t}\rangle \\
G^t_{q,p} &:= 1 - F^t_{q,p,+1} \cdot F^t_{q,p,-1}
\end{align}
In Equation \ref{eq:Jensen2} we used Jensen's inequality. In Equation \ref{eq:log2} we used $\log{1-x} \geq -2x \ \forall x \in [0,0.79]$, which is valid as long as $G^t_{q,p} \leq 0.79$. This is guaranteed by the assumption $K \leq 1/(2 \sqrt{2} \varepsilon)$, since
\begin{align}
F^t_{q,p,r} &= \bra{\psi}(I + r\varepsilon E_{q,p})^{\otimes K}\ket{\psi} \\
&\geq (1 - \varepsilon \sqrt{2})^{K} \\
&\geq 1 - \sqrt{2} K\varepsilon \qquad (\text{provided $\varepsilon \leq 1 / \sqrt{2}$}) \\
&\geq 1/2 \qquad \text{for the choice of $K$},
\end{align}
so
\begin{equation}
G_{q,p}^t = 1 - F^t_{q,p,+1} \cdot F^t_{q,p,-1} \leq 3/4
\end{equation}

We would like to upper bound $\mathbb{E}_{(q,p) \neq (0,0)} \max(0,G^t_{q,p})$. Let's calculate
\begin{align}
G^t_{q,p} &= 1 - \langle\psi^t_{s_t}| \big(\mathbbm{1} + \varepsilon E_{q,p}\big)^{\otimes K} |\psi^t_{s_t}\rangle \langle\psi^t_{s_t}| \big(\mathbbm{1} - \varepsilon E_{q,p}\big)^{\otimes K} |\psi^t_{s_t}\rangle \\
&= 1 - \sum_{S,S' \subset [K]} (-1)^{|S'|} \varepsilon^{|S| + |S'|} \langle\psi^t_{s_t}| E_{q,p}^{\otimes S} |\psi^t_{s_t}\rangle \langle\psi^t_{s_t}| E_{q,p}^{\otimes S'} |\psi^t_{s_t}\rangle \\
&= 1 - \langle\psi^t_{s_t}| \big(\mathbbm{1} + H^0_{q,p}\big) |\psi^t_{s_t}\rangle^2 + \langle\psi^t_{s_t}| H^1_{q,p} |\psi^t_{s_t}\rangle^2 \\
&= - 2 \langle\psi^t_{s_t}| H^0_{q,p} |\psi^t_{s_t}\rangle - \langle\psi^t_{s_t}| H^0_{q,p} |\psi^t_{s_t}\rangle^2 + \langle\psi^t_{s_t}| H^1_{q,p} |\psi^t_{s_t}\rangle^2 \\
&\leq - 2 \langle\psi^t_{s_t}| H^0_{q,p} |\psi^t_{s_t}\rangle + \langle\psi^t_{s_t}| H^1_{q,p} |\psi^t_{s_t}\rangle^2 \\
\implies \mathbb{E}_{(q,p) \neq (0,0)} \max(0,G^t_{q,p}) &\leq 2 \max_{|\psi\rangle} \mathbb{E}_{(q,p) \neq (0,0)} |\langle\psi|H^0_{q,p}|\psi\rangle| + \max_{|\psi\rangle} \mathbb{E}_{(q,p) \neq (0,0)} \langle\psi|H^1_{q,p}|\psi\rangle^2 \\
&\leq 2 \sqrt{\Gamma^0} + \Gamma^1
\end{align}
where
\begin{align}
H^0_{q,p} &= \sum_{S \subset [K] , |S| \text{even} , S \neq \emptyset} \varepsilon^{|S|} E_{q,p}^{\otimes S} \\
H^1_{q,p} &= \sum_{S \subset [K] , |S| \text{odd}} \varepsilon^{|S|} E_{q,p}^{\otimes S} \\
\Gamma^0 &= \max_{|\psi\rangle} \mathbb{E}_{(q,p) \neq (0,0)} \langle\psi|H^0_{q,p}|\psi\rangle^2 \\
\Gamma^1 &= \max_{|\psi\rangle} \mathbb{E}_{(q,p) \neq (0,0)} \langle\psi|H^1_{q,p}|\psi\rangle^2
\end{align}
and we used Cauchy-Schwarz in the final step. Returning to Equation \ref{eq:gamma_calculation2}, we have
\begin{equation} \label{eq:gamma_calculation3}
\frac{\mathbb{E}_{(q,p) \neq (0,0)} \mathbb{E}_r {\mathbb{P}_{\rho_{q,p,r}}(l)}}{\mathbb{P}_{\mathbbm{1}/d}(l)} \geq 1 - T \big(2 \sqrt{\Gamma^0} + \Gamma^1\big)
\end{equation}

It remains to upper bound $\Gamma^0$ and $\Gamma^1$. Let's first deal with $\Gamma^0$.
\begin{align}
\Gamma^0 &= \max_{|\psi\rangle} \mathbb{E}_{(q,p) \neq (0,0)} \langle\psi|\langle\psi| H^0_{q,p} \otimes H^0_{q,p} |\psi\rangle|\psi\rangle \\
&\leq \left|\left|\mathbb{E}_{(q,p) \neq (0,0)} H^0_{q,p} \otimes H^0_{q,p}\right|\right|_{\op} \\
&\leq \sum_{S,S' \subset [K], |S|,|S'| \text{even} , S,S' \neq \emptyset} \varepsilon^{|S| + |S'|} \left|\left|\mathbb{E}_{(q,p) \neq (0,0)} E_{q,p}^{\otimes S \cup S'}\right|\right|_{\text{op}} \\
&\leq \sum_{4 \leq k \leq 2K , k \ \text{even}} (2K)^k \varepsilon^k \left|\left|\mathbb{E}_{(q,p) \neq (0,0)} E_{q,p}^{\otimes k}\right|\right|_{\text{op}} \label{eq:gamma_0_penultimate}
\end{align}

Lemma \ref{lem:E_norm}, which is stated and proved at the end of the section, tells us that
\begin{align}
\left|\left|\sum_{q,p} E_{q,p}^{\otimes k}\right|\right|_{\text{op}} &\leq 2^{k/2} d \\
\implies \ \left|\left|\mathbb{E}_{(q,p) \neq (0,0)} E_{q,p}^{\otimes k}\right|\right|_{\text{op}} &\leq \frac{2^{k/2} d}{d^2 - 1} + \frac{1}{d^2 - 1} \leq \frac{2^{k/2}}{d-1}
\end{align}

Plugging this into Equation \ref{eq:gamma_0_penultimate} gives
\begin{equation}
\Gamma^0 \leq \frac{1}{d-1} \sum_{4 \leq k \leq 2K , k \ \text{even}} \big(2 \sqrt{2} K \varepsilon\big)^k
\end{equation}
This is a geometric series, which is dominated by its first term since by assumption $K \leq 1 / (2 \sqrt{2} \varepsilon)$. We get
\begin{equation}
\Gamma^0 = \mathcal{O}(K^4 \varepsilon^4 / d)
\end{equation}

The calculation for $\Gamma^1$ is similar.
\begin{align}
\Gamma^1 &\leq \sum_{2 \leq k \leq 2K , k \ \text{even}} (2K)^k \varepsilon^k \left|\left|\mathbb{E}_{(q,p) \neq (0,0)} E_{q,p}^{\otimes k}\right|\right|_{\text{op}} \\
&= \mathcal{O}(K^2 \varepsilon^2 / d)
\end{align}

Returning to Equation \ref{eq:gamma_calculation3}, we have
\begin{equation}
\frac{\mathbb{E}_{(q,p) \neq (0,0)} \mathbb{E}_r {\mathbb{P}_{\rho_{q,p,r}}(l)}}{\mathbb{P}_{\mathbbm{1}/d}(l)} \geq 1 - \mathcal{O}\big(T K^2 \varepsilon^2 / \sqrt{d}\big)
\end{equation}
By Lemma \ref{lem:5.4}, our protocol $\mathcal{A}$ succeeds with probability at most $(1 + \mathcal{O}(T K^2 \varepsilon^2 / \sqrt{d})) / 2$. To succeed with probability $2/3$, $\mathcal{A}$ requires $T = \Omega(\sqrt{d} / (K^2 \varepsilon^2))$. This completes the proof of Proposition \ref{prop:distinguishing_no_conjugate}, modulo Lemma \ref{lem:E_norm}. 
\end{proof}

\begin{lemma} \label{lem:D_norm}
Let the Hilbert space dimension $d$ be prime. For any $1 \leq m \leq k$ with $k$ even,
\begin{equation}
\left|\left|\sum_{q,p} D_{q,p}^{\otimes m} \otimes D_{-q,-p}^{\otimes (k-m)}\right|\right|_{\op} = d
\end{equation}
\end{lemma}
\begin{proof}
Denote
\begin{equation}
\mathcal{D}(m,k) = \sum_{q,p} D_{q,p}^{\otimes m} \otimes D_{-q,-p}^{\otimes (k-m)}
\end{equation}

Consider the action on a basis state $|a_1\rangle \dots |a_k\rangle$.
\begin{align}
&\mathcal{D}(m,k) |a_1\rangle \dots |a_k\rangle \\
&= \sum_{q,p} e^{i \pi k q p / d} \big((X^q)^{\otimes m} \otimes (X^{-q})^{\otimes (k-m)}\big) \big((Z^p)^{\otimes m} \otimes (Z^{-p})^{\otimes (k-m)}\big) |a_1\rangle \dots |a_k\rangle \\
&= \sum_{q,p} e^{i (2 \pi / d)  p \left(\frac{k}{2}q + a_1 + \dots + a_m - a_{m+1} - \dots - a_k\right)} |a_1+q\rangle \dots |a_m+q\rangle |a_{m+1}-q\rangle \dots |a_k+q\rangle \\
&= d |a_1+\hat{q}\rangle \dots |a_m+\hat{q}\rangle |a_{m+1}-\hat{q}\rangle \dots |a_k-\hat{q}\rangle
\end{align}
where $\hat{q}$ is the unique solution to
\begin{equation}
\frac{k}{2} q + a_1 + \dots + a_m - a_{m+1} - \dots - a_k = 0 \mod d
\end{equation}
Here we used that $d$ is prime, so that $\mathbb{Z}_d$ is a field.

Let $g$ be the multiplicative inverse of $k/2$ mod $d$. The operator $\mathcal{D}(m,k) / d$ implements a linear map on basis vectors over $\mathbb{Z}_d^k$ given by the matrix
\begin{equation}
I_k +
\begin{pmatrix}
    g & \dots & g & -g & \dots & -g \\
    \vdots & & \vdots & \vdots & & \vdots \\
    g & \dots & g & -g & \dots & -g
\end{pmatrix}
\end{equation}
This matrix is invertible over $\mathbb{Z}_d^k$ with inverse
\begin{equation}
I_k +
\begin{pmatrix}
    h & \dots & h & -h & \dots & -h \\
    \vdots & & \vdots & \vdots & & \vdots \\
    h & \dots & h & -h & \dots & -h
\end{pmatrix}
\end{equation}
where $h$ solves
\begin{equation}
g + h + (2m-k)gh = 0
\end{equation}

We have shown that $\mathcal{D}(m,k) / d$ is in fact a permutation matrix. In particular,
\begin{equation}
\left|\left|\mathcal{D}(m,k)\right|\right|_{\op} = d
\end{equation}
\end{proof}

\begin{lemma} \label{lem:E_norm}
Let the Hilbert space dimension $d$ be prime. For any even $k$,
\begin{equation}
\left|\left|\sum_{q,p} E_{q,p}^{\otimes k}\right|\right|_{\op} \leq 2^{k/2} d
\end{equation}
\end{lemma}
\begin{proof}
First expand
\begin{align}
E_{q,p}^{\otimes k} &= \big(\chi D_{q,p} + \chi^\ast D_{-q,-p}\big)^{\otimes k} \\
&= \sum_{S \subset [k]} \chi^{|S|} (\chi^\ast)^{k - |S|} D_{q,p}^{\otimes S} \otimes D_{-q,-p}^{\otimes [k] \setminus S}
\end{align}
Now sum over $q,p$ and take the operator norm.
\begin{align}
\left|\left|\sum_{q,p} E_{q,p}^{\otimes k}\right|\right|_{\op} &\leq \frac{1}{2^{k/2}} \sum_{S \subset [k]} \left|\left|\sum_{q,p} D_{q,p}^{\otimes S} \otimes D_{-q,-p}^{\otimes [k] \setminus S}\right|\right|_{\op} \\
&= 2^{k/2} d
\end{align}
using Lemma \ref{lem:D_norm}.
\end{proof}

We suspect Lemma \ref{lem:E_norm} is \emph{not} tight, and the correct bound is $2d$ independent of $k$, but it is sufficient for our purposes.

\section{Generalized Clifford classical shadows} \label{sec:shadows}

In this appendix, we analyse a bosonic version of classical shadows, derived from the uniform distribution over generalized Cliffords in prime dimension.

\subsection{Review of classical shadows framework} \label{sec: shadows review}

We start by reviewing the classical shadows framework of Ref.~\cite{huang2020predicting}, introducing some generalizations that are relevant for generalized Clifford shadows. This review is based in large part on Section 2.2 of Ref. \cite{wan2022matchgate}. See also the review in \cite{mele2023introduction}.

Let $\mathcal{H}$ be a finite-dimensional Hilbert space, and let $\mathcal{B}$ denote an orthonormal basis for $\mathcal{H}$. 
The goal of the classical shadows procedure is to estimate the expectation values $\tr(O_1 \rho), \dots, \tr(O_M\rho)$ of $M$ (possible non-Hermitian) ``observables'' $O_1, \dots, O_M$ $\mathcal{H}$, with respect to an unknown state $\rho$. The input to the procedure consists of copies of $\rho$ and some classical description of $O_1,\dots, O_M$, and (having fixed $\mathcal{B}$) the procedure is parametrised by a distribution $\mathcal{D}$ over unitaries. For each copy of $\rho$, we randomly draw a unitary $U$ from this distribution, and measure $\rho$ in the basis $\{U^\dagger\ket{b}\}_{b\in \mathcal{B}}$.\footnote{This measurement can be implemented by applying $U$ to $\rho$, then measuring in the basis $\mathcal{B}$, and applying $U^\dagger$ to the post-measurement state. Note, however, that for the purpose of implementing the classical shadows procedure, the application of $U^\dagger$ is not necessary, as we only need the measurement statistics and not the post-measurement state.} The quantum channel $\mathcal{M}$ corresponding to this process is given by
\begin{align} \label{measurement channel 1}
    \mathcal{M}(\rho) &= \E_{U \sim \mathcal{D}} \sum_{b \in \mathcal{B}} U^\dagger\ket{b}\bra{b}U \rho U^\dagger\ket{b}\bra{b}U,
\end{align}
which can be rewritten as 
\begin{align}
    \mathcal{M}(\rho) &= \tr_1\left[ \sum_{b\in \mathcal{B}} \E_{U \sim \mathcal{D}} \mathcal{U}^{\otimes 2}(\ket{b}\bra{b}^{\otimes 2})(\rho \otimes I)\right] \label{measurement channel}
\end{align}
in order to isolate the dependence on the distribution $\mathcal{D}$. Here, $\tr_1$ denotes the partial trace over the first tensor component, and $\mathcal{U}$ denotes the unitary channel corresponding to ${U}^\dagger$: $\mathcal{U}(\,\cdot\,) = U^\dagger(\,\cdot\,)U$. 

For certain distributions $\mathcal{D}$, $\mathcal{M}$ is invertible. Then, we can define a random operator $\hat{\rho}$ by 
\begin{equation} \label{hat rho} \hat{\rho} \coloneqq \mathcal{M}^{-1}(\hat{U}^\dagger\ket{\hat{b}}\bra{\hat{b}}\hat{U}), \end{equation}
where $\hat{U}$ is distributed according to $\mathcal{D}$ and $\mathbb{P}[\ket{\hat{b}} = \ket{b} \, | \, \hat{U} = U] = \bra{b}U \rho U^\dagger\ket{b}$. By construction, $\hat{\rho}$ is an unbiased estimator for $\rho$:
\[ \E[\hat{\rho}] = \rho. \]
In the literature, the term ``classical shadow'' is used sometimes to refer to this unbiased estimator $\hat{\rho}$, and sometimes to a sample of it obtained in one realisation of the above procedure (i.e., $\mathcal{M}^{-1}(U^\dagger\ket{b}\bra{b}U)$ for some outcomes $U$ and $\ket{b}$). These samples can be used to estimate the expectation values $\tr(O_i\rho)$, since 
\[ \E[\hat{o}_i] = \tr(O_i\rho),\]
where for $i \in [M]$,
\begin{equation} \label{hat oi} \hat{o}_i \coloneqq \tr(O_i\hat{\rho}) = \tr\left( O_i \mathcal{M}^{-1}(\hat{U}^\dagger\ket{\hat{b}}\bra{\hat{b}}\hat{U})\right) \end{equation} is the unbiased estimator for $\tr(O_i\rho)$ derived from $\hat{\rho}$. 

To bound the number of samples of the classical shadow estimator $\hat{\rho}$ (and hence the number of copies of $\rho$) required to estimate the expectation values to within some desired precision with high probability, we consider the variances of the estimators $\hat{o}_i$:
\begin{align}
    \Var[\hat{o}_i] &= \E[|\hat{o}_i|^2] - |\E[\hat{o}_i]|^2 \\
    &= \E_{U \sim \mathcal{D}} \sum_{b \in \mathcal{B}} \bra{b} U\rho U^\dagger\ket{b} \left|\tr\left(O_i \mathcal{M}^{-1}(U^\dagger \ket{b}\bra{b}U)\right)\right|^2 - |\tr(O_i\rho)|^2 \\
    &= \E_{U \sim \mathcal{D}} \sum_{b \in \mathcal{B}} \tr\left[U^\dagger \ket{b}\bra{b} U \rho \otimes \mathcal{M}^{-1}(U^\dagger \ket{b}\bra{b}U)O_i \otimes  \mathcal{M}^{-1}(U^\dagger\ket{b}\bra{b}U) O_i^\dagger\right] - |\tr(O_i\rho)|^2 \\
    &= \tr\left[\sum_{b \in \mathcal{B}} \E_{U \in \mathcal{D}} \mathcal{U}^{\otimes 3}(\ket{b}\bra{b}^{\otimes 3}) \left(\mathcal{M}^{-1}(O_i) \otimes \mathcal{M}^{-1}(O_i^\dagger) \otimes \rho\right) \right] - |\tr(O_i\rho)|^2, \label{shadows variance}
\end{align}
using the fact that $\mathcal{M}$ is self-adjoint (with respect to the Hilbert-Schmidt inner product) in the last line, so $\tr(\mathcal{M}^{-1}(A)B) =\tr(A\mathcal{M}^{-1}(B))$ for any operators $A$ and $B$. This variance can be upper-bounded by the first term, maximised over all states $\rho$; this gives the square of what is often referred to as the ``shadow norm" in the literature. A useful observation is that since $\tr(\hat{\rho}) = 1$, the variance for any observable $\hat{O}_i$ is the same as that for $\hat{O}_i + cI$ for any $c \in \mathbb{C}$, so we can also also write
\begin{align} \label{shadows variance shifted}
    \Var[\hat{o}_i] = \tr\left[\sum_{b \in \mathcal{B}} \E_{U \in \mathcal{D}} \mathcal{U}^{\otimes 3}(\ket{b}\bra{b}^{\otimes 3}) \left(\mathcal{M}^{-1}(O_i + cI) \otimes \mathcal{M}^{-1}(O_i^\dagger + cI) \otimes \rho\right) \right] - |\tr((O_i + cI)\rho)|^2.
\end{align}
If median-of-means estimators are used, it follows straightforwardly from Chebyshev's and Hoeffding's inequalities that
\begin{equation} \label{Nsample} N_{\mathrm{sample}} = \mathcal{O}\left(\frac{\log (M/\delta)}{\eps^2} \max\limits_{i \in [M]} \Var[\hat{o}_i] \right) \end{equation} 
classical shadows samples (and hence single-copy measurements of $\rho$) are sufficient to ensure that with probability at least $1-\delta$, every $\tr(O_i\rho)$ is estimated to within additive error $\eps$.

\subsection{Preliminaries and notation} 

\subsubsection{Generalized Clifford group}

The generalized (single-qudit) Clifford group in dimension $d$, which we will denote by $\Cld$, is the normaliser of the Pauli group in the $d$-dimensional unitary group.\footnote{Strictly speaking, we should quotient out phases so that it makes sense to write the twirl channel in Eq.~\eqref{Ekd} as an average over a finite group, but the results in this appendix hold even if phases are included.} Hence, a Clifford $U \in \Cld$ satisfies the property that for any $(q,p) \in \Zd^2$, $U^\dagger D_{q,p}U = \alpha D_{q',p'}$
for some $(q',p') \in \Zd^2$ and some phase $\alpha$. Thus, each Clifford can be associated with a map $\Zd^2 \to \Zd^2$, and since its action under conjugation is completely specified by how $X = D_{1,0}$ and $Z = D_{0,1}$ are transformed, this map is linear, and can be written as a $2 \times 2$ matrix $C \in \Zd^{2\times 2}$ so that
\begin{equation} \label{Clifford matrix} U^\dagger D_{q,p} U = \alpha D_{q',p'} \text{ for some phase $\alpha$ } \quad \Leftrightarrow \quad C\begin{pmatrix} q \\ p\end{pmatrix} = \begin{pmatrix} q' \\ p' \end{pmatrix} \end{equation}
for all $q,p \in \Zd$. It follows from the fact that conjugation preserves commutation relations that this matrix also has to be symplectic. It turns out that conversely, for any symplectic $2\times 2$ matrix with entries in $\Zd$, there exists a corresponding Clifford in the sense of Eq.~\eqref{Clifford matrix} \cite{hostens2005stabilizer}.

\begin{fact}[\cite{hostens2005stabilizer}] \label{fact: Clifford symplectic} 
    For any $C \in \Zd^{2\times 2}$, there exists $U \in \Cld$ satisfying Eq.~\eqref{Clifford matrix} if and only if $C$ is symplectic. 
\end{fact}

The following basic fact about symplectic $2\times 2$ matrices will also come in useful.

\begin{fact} \label{fact: 2x2 symplectic}
A $2\times 2$ matrix $C$ is symplectic if and only if $\det(C) = 1$.
\end{fact}

\subsubsection{Liouville representation}

In some parts of this appendix (especially when we analyse the twirl channels of the generalized Clifford group in the following subsection), we will use Liouville representation, which is often useful for specifying quantum channels. In this representation, operators are notated using ``double'' kets, and a ``double'' bracket is used to represent the Hilbert-Schmidt inner product. By convention, all double kets are normalised with respect to the Hilbert-Schmidt norm, unless otherwise noted. Thus, for any nonzero operators $A,B$, 
\begin{equation} \label{Liouville1} \kett{A} \equiv \frac{1}{\sqrt{\tr(A^\dagger A)}}A, \qquad \braakett{A}{B} \equiv \frac{\tr(A^\dagger B)}{\sqrt{\tr(A^\dagger A)\tr(B^\dagger B)}},\end{equation}
and we set $\ket{0} \equiv 0$ for the zero operator. In particular, since the generalized Pauli operators are unitary (Eq.~\eqref{D dagger}) and Hilbert-Schmidt orthogonal (Proposition~\ref{prop:disp_basis}), we have
\begin{equation} \label{Liouville Dqp} \kett{D_{q,p}} \equiv \frac{1}{\sqrt{d}} D_{q,p}, \qquad \braakett{D_{q,p}}{D_{q',p'}} = \delta_{q,q'}\delta_{p,p'}. \end{equation}
As with usual (state) kets, we will freely write e.g., $\kett{A}\kett{B}$ in place of $\kett{A} \otimes \kett{B}$. 

A superoperator $\mathcal{E}$ acting on an operator $A$ is represented by placing the superoperator to the left of the operator's double ket:
\[ \mathcal{E}\kett{A} \equiv \frac{1}{\sqrt{\tr(A^\dagger A)}}\mathcal{E}(A). \]
We will also write $\mathcal{E}\mathcal{E}'$ in place of $\mathcal{E} \circ \mathcal{E}'$. 

Since the generalized Pauli operators form a Hilbert-Schmidt orthogonal basis for the space of operators acting on $d$-dimensional Hilbert space $\mathcal{H}_d$, we have the following resolution of the identity superoperator $\mathcal{I}$:
\begin{equation} \label{resolution of superidentity} \mathcal{I} = \sum_{p,q \in \Zd} \kett{D_{p,q}}\braa{D_{p,q}}. \end{equation}

\subsection{Uniform distribution over generalized Cliffords}

In the rest of this appendix, we consider the classical shadows associated with the uniform distribution over $\Cld$, the generalized (single-qudit) Cliffords in prime dimension $d$. 

As can be seen from Eqs.~\eqref{measurement channel} and~\eqref{shadows variance}, the measurement channel $\mathcal{M}$ in the classical shadows protocol depends on the chosen unitary distribution $\mathcal{D}$ only through its $2$-fold twirl $\E_{U \sim \mathcal{D}}\mathcal{U}^{\otimes 2}$, while the variances of the resulting estimates depend on $\mathcal{D}$ through its $3$-fold twirl $\E_{U \sim \mathcal{D}} \mathcal{U}^{\otimes 3}$. Hence, in this subsection, we evaluate the $2$- and $3$-fold twirl channels
for the uniform distribution over generalized Cliffords, arriving at the theorem below. This then allow us to obtain explicit expressions for the measurement channel $\mathcal{M}$ and the variance for any observable. 

For $k \in \mathbb{Z}_{>0}$, we use $\mathcal{E}_d^{(k)}$ to denote the $k$-fold twirl channel corresponding to the uniform distribution over $\Cld$:
\begin{equation} \label{Ekd}
    \mathcal{E}_d^{(k)} \coloneqq \frac{1}{|\Cld|}\sum_{U \in \Cld}\mathcal{U}^{\otimes k},
\end{equation}
where, as in subsection~\ref{sec: shadows review}, $\mathcal{U}(\,\cdot\,) = U^\dagger (\,\cdot\,) U$. 


\begin{theorem}[First three moments of uniform distribution over generalized Cliffords] \label{thm: twirl channels} Let $\mathcal{E}_d^{(k)}$ be defined as in Eq.~\eqref{Ekd}. Then, for any prime $d$, we have
\begin{enumerate}[(i)]
    \item $\mathcal{E}_d^{(1)} = \kett{I}\braa{I},$
    \item $\mathcal{E}_d^{(2)} = \kett{I}\braa{I}^{\otimes 2} + \kett{\Phi_0}\braa{\Phi_0}$,
    \item $\mathcal{E}_d^{(3)} = \kett{I}\braa{I}^{\otimes 3} + \kett{\Phi_1}\braa{\Phi_1} + \kett{\Phi_2}\braa{\Phi_2}+ \kett{\Phi_3}\braa{\Phi_3} + \sum\limits_{k=1}^{d-2} \kett{\Psi_k}\braa{\Psi_k} + \sum\limits_{l = 1}^{d-1} \kett{\Upsilon_l}\braa{\Upsilon_l}$,
\end{enumerate}
where 
\begin{equation} \label{Phi0}
    \kett{\Phi_0} \coloneqq \frac{1}{\sqrt{d^2 - 1}} \sum_{\substack{q,p\in \Zd \\ (q,p) \neq (0,0)}} \kett{D_{q,p}}\kett{D_{q,p}^\dagger}
\end{equation}
and
\begin{align}
    &\kett{\Phi_1} \coloneqq \frac{1}{\sqrt{d^2-1}} \sum_{\substack{q,p \in\Zd \\ (q,p) \neq (0,0)}} \kett{I} \kett{D_{q,p}}\kett{D_{q,p}^\dagger}, \label{Phi1} \\
    &\kett{\Phi_2} \coloneqq \frac{1}{\sqrt{d^2-1}} \sum_{\substack{q,p \in\Zd \\ (q,p) \neq (0,0)}}  \kett{D_{q,p}}\kett{I}\kett{D_{q,p}^\dagger}, \\
    &\kett{\Phi_3} \coloneqq \frac{1}{\sqrt{d^2-1}} \sum_{\substack{q,p \in\Zd \\ (q,p) \neq (0,0)}}  \kett{D_{q,p}}\kett{D_{q,p}^\dagger}\kett{I}, \\
    &\ket{\Psi_k} \coloneqq \frac{1}{\sqrt{d^2-1}}\sum_{\substack{q,p \in\Zd \\ (q,p) \neq (0,0)}} \kett{D_{q,p}}\kett{D_{kq,kp}} \kett{D_{q,p}^\dagger D_{kq,kp}^\dagger} \\
    &\kett{\Upsilon_l} \coloneqq \frac{1}{\sqrt{d(d^2-1)}} \sum_{\substack{q_1,p_1,q_2,p_2 \in \Zd \\ p_1q_2 - q_1p_2 = l}} \kett{D_{q_1,p_1}}\kett{D_{q_2,p_2}} \kett{D_{q_1,p_1}^\dagger D_{q_2,p_2}^\dagger} \label{Upsilon_l}.
\end{align}
\end{theorem}


The strategy we use to prove Theorem~\ref{thm: twirl channels} is very similar in spirit to that used to determine the $2$- and $3$-fold twirl channels of the matchgate group in Ref.~\cite{wan2022matchgate}. In broad strokes, we start from the simple observation that each twirl channel is an orthogonal projector, and use the invariance  of the twirl channel under composition with suitably chosen group elements to determine its image. The difference in the details of the proof of Theorem~\ref{thm: twirl channels} differ from those in Ref.~\cite{wan2022matchgate} stem from the difference in the structure of the adjoint representation of the Clifford group and that of the matchgate group, with the former being somewhat easier to handle since it is a discrete group. 

The following facts follow straightforwardly from the fact that (for any $k \in \mathbb{Z}_{>0}$) the map $U \mapsto \mathcal{U}^{\otimes k}$ is a representation of the generalized Clifford group; see e.g., \cite{fulton2013representation} Proposition 2.8.

\begin{fact} \label{fact: projector}
For any $k \in \mathbb{Z}_{>0}$, $\mathcal{E}_d^{(k)}$ is the orthogonal projector onto the subspace $\{A \in \mathcal{L}(\mathcal{H}_d)^{\otimes k}: \mathcal{U}^{\otimes k}(A) = A \enspace  \forall \, U \in \Cld\}$.
\end{fact}

\begin{fact} \label{fact: invariance}
For any $k \in \mathbb{Z}_{>0}$ and $U \in \Cld$, 
\[ \mathcal{U}^{\otimes k} \circ \mathcal{E}_d^{(k)} = \mathcal{E}_d^{(k)} = \mathcal{E}_d^{(k)} \circ \mathcal{U}^{\otimes k}. \]
\end{fact}

We will also make use of the following simple fact. 

\begin{fact} \label{fact: numbers}
Let $d$ be a prime number. For any $q,p \in \Zd$ with $(q,p) \neq (0,0)$, there exist $q',p' \in \Zd$ such that
\[ p'q - q' p = 1.\]
\end{fact}
\begin{proof}
    If $q \neq 0$, choose any $q'$ and set $p' = (q'p + 1)q^{-1}$. Otherwise, we must have $p \neq 0$, in which case we choose any $p'$ and set $q' = (p'q - 1)p^{-1}$. 
\end{proof}

\subsubsection{1-fold twirl}
We start by evaluating the $1$-fold twirl $\mathcal{E}_d^{(1)}$, since a similar idea will be used in obtaining the $2$- and $3$-fold twirls.
\begin{proof}[Proof of Theorem~\ref{thm: twirl channels}(i)]
We can rewrite Eq.~\eqref{eq:D_comm} as 
\begin{equation} \label{D conjugation} D_{q',p'} D_{q,p} D_{q',p'}^\dagger = \omega^{p'q - q' p} D_{q,p}, \end{equation}
with $\omega \coloneqq e^{2\pi i/d}$.
By Fact~\ref{fact: numbers}, there exist $q',p'$ such that $p'q-q'p = 1$ provided that $(q,p) \neq (0,0)$. Thus, for any $(q,p) \neq (0,0)$, there exists a Pauli $U \in \Cld$ such that
\[ \mathcal{U}\kett{D_{q,p}} = \omega\kett{D_{q,p}}. \] Hence, 
\begin{align*}
    \mathcal{E}_d^{(1)} \kett{D_{q,p}} &= \frac{1}{d} \sum_{j = 0}^{d-1} \mathcal{E}_d^{(1)} \mathcal{U}^{j} \kett{D_{q,p}} \\
    &= \frac{1}{d}\mathcal{E}_d^{(1)} \left(\sum_{j=0}^{d-1} \omega^{j} \right)\kett{D_{q,p}} \\
    &= 0
\end{align*}
for any $(q,p) \neq (0,0)$, where we use Fact~\ref{fact: invariance} in the first line, while for $(q,p) = (0,0)$, we have $\mathcal{E}_d^{(1)} \kett{D_{0,0}} = \mathcal{E}_d^{(1)} \kett{I} = \kett{I}$. Since $\{D_{q,p}\}_{q,p \in \Zd}$ forms a basis for the space of operators, it follows that
\[ \mathcal{E}_1 = \kett{I}\braa{I}. \]
\end{proof}

\subsubsection{2-fold twirl}

Now, we evaluate the $2$-fold twirl $\mathcal{E}_d^{(2)}$, which will allow us to determine the measurement channel (Eq.~\eqref{measurement channel}) for the generalized Clifford shadows. We start with a lemma that precludes certain basis states (in the Pauli basis) from being in the image of $\mathcal{E}_d^{(2)}$. 

\begin{lemma} \label{lem: 2-fold kernel} Let $d$ be a prime number. 
    For $q_1,p_1, q_2,p_2 \in \Zd$, $\mathcal{E}_d^{(2)}\kett{D_{q_1,p_1}}\kett{D_{q_2,p_2}} \neq 0$ only if $(q_1 + q_2, p_1 + p_2) = (0,0)$.
\end{lemma}
\begin{proof} The proof is very similar to that of Theorem~\ref{thm: twirl channels}(i) above. 
    From Eq.~\eqref{D conjugation}, 
    \[ D_{q',p'}^{\otimes 2} (D_{q_1,p_1}\otimes D_{q_2,p_2})(D_{q',p'}^{\otimes 2})^\dagger = \omega^{p'(q_1 + q_2) - q'(p_1 + p_2)}D_{q_1,p_1}\otimes D_{q_2,p_2}, \] 
    so it follows from Fact~\ref{fact: numbers} that for any $q_1,p_1,q_2,p_2 \in \Zd$ with $(q_1 + q_2, p_1 + p_2) \neq (0,0)$, there exists a Pauli $U \in \Cld$ such that
    \[ \mathcal{U}^{\otimes 2}\kett{D_{q_1,p_1}}\kett{D_{q_2,p_2}} = \omega \kett{D_{q_1,p_1}}\kett{D_{q_2,p_2}}. \]
    Hence,
    \begin{align*}
        \mathcal{E}_d^{(2)}\kett{D_{q_1,p_1}}\kett{D_{q_2,p_2}} &= \frac{1}{d} \sum_{j=0}^{d-1} \mathcal{E}_d^{(2)} (\mathcal{U}^j)^{\otimes 2} \kett{D_{q_1,p_1}}\kett{D_{q_2,p_2}} \\
        &= \frac{1}{d} \mathcal{E}_d^{(2)} \left( \sum_{j=0}^{d-1} \omega^j\right)\kett{D_{q_1,p_1}}\kett{D_{q_2,p_2}} \\
        &= 0.
    \end{align*}
\end{proof}

Next, we show that all $D_{q,p}$ with $(q,p) \neq (0,0)$ are in the same orbit under the action of $\Cld$ by conjugation (if we ignore global phases). 

\begin{lemma} \label{lem: 2-fold orbit} Let $d$ be a prime number. 
    For any $q,p,q',p' \in \Zd$ with $(q,p), (q',p') \neq (0,0)$, there exists $U \in \Cld$ such that 
    \[ \mathcal{U} \kett{D_{q,p}} = \alpha\kett{D_{q',p'}}\]
    for some phase $\alpha$.
\end{lemma}
\begin{proof}
    By Fact~\ref{fact: Clifford symplectic} and Eq.~\eqref{Clifford matrix}, such a Clifford exists provided that there exists a symplectic matrix $\begin{pmatrix} a &b \\ c &d\end{pmatrix}$ such that
    \[ \begin{pmatrix} w &x \\ y &z\end{pmatrix}\begin{pmatrix} q\\p \end{pmatrix} =\begin{pmatrix} q'\\p' \end{pmatrix}. \]
    If $q \neq 0$, we set $w = (q' - xp)q^{-1}$ and $y = (p' - zp)q^{-1}$. Then, $\det \begin{pmatrix} w &x \\ y &z\end{pmatrix} = (q'z-xp')q^{-1}$. By Fact~\ref{fact: numbers}, since $(q',p') \neq (0,0)$, we can choose $z, x$ such that $(q'z-xp') = q$, so that the determinant is $1$---i.e., the matrix is symplectic by Fact~\ref{fact: 2x2 symplectic}. Otherwise, we must have $p \neq 0$, in which case we set $x = (q' - wq)p^{-1}$ and $z = (p' - yq)p^{-1}$, and choose $w, y$ such that $(wp' - q'y) = p$ via Fact~\ref{fact: numbers}.
\end{proof}

\begin{proof}[Proof of Theorem~\ref{thm: twirl channels}(ii)]
Inserting resolutions of the identity (Eq.~\eqref{resolution of superidentity}) and using Lemma~\ref{lem: 2-fold kernel}, we have
\begin{align}
    \mathcal{E}_d^{(2)} &= \sum_{q,p,q',p' \in \Zd} \kett{D_{q,p}} \kett{D_{-q,-p}} \braa{D_{q,p}}\braa{D_{-q,-p}} \mathcal{E}_d^{(2)} \kett{D_{q',p'}} \kett{D_{-q',-p'}} \braa{D_{q',p'}}\braa{D_{-q',-p'}} \\
    &= \sum_{q,p,q',p' \in \Zd} \kett{D_{q,p}} \kett{D_{q,p}^\dagger} \braa{D_{q,p}}\braa{D_{q,p}^\dagger} \mathcal{E}_d^{(2)} \kett{D_{q',p'}} \kett{D_{q',p'}^\dagger} \braa{D_{q',p'}}\braa{D_{q',p'}^\dagger} \\
    &= \kett{I}\kett{I}\braa{I}\braa{I}+ \sum_{\substack{q,p,q',p'\in\Zd \\ (q,p), (q',p') \neq (0,0)}} \kett{D_{q,p}} \kett{D_{q,p}^\dagger} \braa{D_{q,p}}\braa{D_{q,p}^\dagger} \mathcal{E}_d^{(2)} \kett{D_{q',p'}} \kett{D_{q',p'}^\dagger} \braa{D_{q',p'}}\braa{D_{q',p'}^\dagger} \label{Ed2 calculation}
\end{align}
(note that since $\mathcal{E}_d^{(2)\dagger} = \mathcal{E}_d^{(2)}$ by Fact~\ref{fact: projector}, Lemma~\ref{lem: 2-fold kernel} also implies that $\braa{D_{q_1,p_1}}\braa{D_{q_2,p_2}}\mathcal{E}_d^{(2)}  = 0$ if $(q_2,p_2) \neq (-q_1,-p_1)$.) Here, the second line uses Eq.~\eqref{D dagger} and the third line follows from the obvious fact that $\mathcal{E}_d^{(2)} \kett{D_{0,0}} = \kett{D_{0,0}} = \kett{I}$.  

Now, we use Lemma~\ref{lem: 2-fold orbit} to show that $\braa{D_{q,p}}\braa{D_{q,p}^\dagger} \mathcal{E}_d^{(2)} \kett{D_{q',p'}} \kett{D_{q',p'}^\dagger}$ is the same for all $(q,p), (q',p') \neq (0,0)$, i.e., 
\begin{equation} \label{2-fold inner product}
\braa{D_{q,p}}\braa{D_{q,p}^\dagger} \mathcal{E}_d^{(2)} \kett{D_{q',p'}} \kett{D_{q',p'}^\dagger} = c
\end{equation}
for some constant $c$. To see this,  note from Lemma~\ref{lem: 2-fold orbit} that for any $(q'',p'') \neq (0,0)$, there exists $U \in \Cld$ such that $\mathcal{U}\kett{D_{q',p'}} = \alpha \kett{D_{q'',p''}}$ for some phase $\alpha$, so
\begin{equation} \label{phases cancel}
\mathcal{U}^{\otimes 2}\kett{D_{q',p'}}\kett{D_{q',p'}^\dagger} = \alpha\kett{D_{q'',p''}} \alpha^*\kett{D_{q'',p''}^\dagger} = \kett{D_{q'',p''}} \kett{D_{q'',p''}^\dagger}
\end{equation}
Then, use Fact~\ref{fact: invariance} to write
\begin{align*}
    \mathcal{E}_d^{(2)} \kett{D_{q',p'}} \kett{D_{q',p'}^\dagger} &= \mathcal{E}_d^{(2)}\mathcal{U}^{\otimes 2} \kett{D_{q',p'}} \kett{D_{q',p'}^\dagger} 
    = \mathcal{E}_d^{(2)} \kett{D_{q'',p''}} \kett{D_{q'',p''}^\dagger}.
\end{align*}

Inserting Eq.~\eqref{2-fold inner product} into Eq.~\eqref{Ed2 calculation}, we arrive at
\begin{align*}
    \mathcal{E}_d^{(2)} &= \kett{I}\braa{I}^{\otimes 2} + c\sum_{\substack{q,p,q',p'\in\Zd \\ (q,p), (q',p') \neq (0,0)}} \kett{D_{q,p}}\kett{D_{q,p}^\dagger}\braa{D_{q',p'}}\braa{D_{q',p'}^\dagger} \\
    &= \kett{I}\braa{I}^{\otimes 2} + c \kett{\Phi_0}\braa{\Phi_0},
\end{align*}
where $\kett{\Phi_0}$ is defined as in Eq.~\eqref{Phi0}. Since $\mathcal{E}_d^{(2)}$ is a projector (Fact~\ref{fact: projector}), $c$ must be equal to $0$ or $1$. 
We see from Eq.~\eqref{phases cancel} that for any $U \in \Cld$, $\mathcal{U}^{\otimes 2}$ permutes the basis states of the form $\kett{D_{q,p}}\kett{D_{q,p}^\dagger}$ for $(q,p) \neq (0,0)$, without incurring any phases. Since $\kett{\Phi_0}$ is the equal superposition over these basis states, it follows that $\mathcal{U}^{\otimes 2}\kett{\Phi_0} = \kett{\Phi_0}$, so $\mathcal{E}_d^{(2)}\kett{\Phi_0} = \kett{\Phi_0}$ and $c$ must be $1$.  
\end{proof}

\subsubsection{3-fold twirl}

Finally, we evaluate the $3$-fold twirl $\mathcal{E}_d^{(3)}$, which will allow us to analyse the variance (Eq.~\eqref{shadows variance}) of estimates obtained using generalized Clifford shadows. The high-level strategy for evaluating $\mathcal{E}_d^{(3)}$ is is essentially the same as that for evaluating the $2$-fold twirl channel $\mathcal{E}_d^{(2)}$. We start with the following lemma, which precludes certain basis states from being in the image of $\mathcal{E}_d^{(3)}$. This lemma is the straightforward extension of Lemma~\ref{lem: 2-fold kernel}, which was used in evaluating $\mathcal{E}_d^{(2)}$. (In fact, there is a natural generalization of these lemmas to $\mathcal{E}_d^{(k)}$ for any $k \in \mathbb{Z}_{>0}$.) 

\begin{lemma} \label{lem: 3-fold kernel}
Let $d$ be a prime number. For $q_1,p_1, q_2, p_2, q_3, p_3 \in \Zd$, $\mathcal{E}_d^{(3)}\kett{D_{q_1,p_1}}\kett{D_{q_2,p_2}}\kett{D_{q_3,p_3}} \neq 0$ only if $(q_1 + q_2 + q_3, p_1 + p_2 + p_3) = (0,0)$.
\end{lemma}
\begin{proof}
    The proof is essentially the same as that of Lemma~\ref{lem: 2-fold kernel}, except we consider three tensor components instead of two.
\end{proof}

Lemma~\ref{lem: 3-fold kernel} shows that we need only consider the action of $\mathcal{E}_d^{(3)}$ on basis states of the form $\kett{D_{q_1,p_1}}\kett{D_{q_2,p_2}}\kett{D_{-q_1-q_2,-p_1-p_2}}$. For these states, the third tensor component is fully determined by the first two. For the purpose of the next lemma, we will say that two pairs of Paulis $(D_{q_1,p_1},D_{q_2,p_2})$ and $(D_{q_1',p_1'},D_{q_2',p_2'})$ are in the ``orbit''\footnote{Strictly speaking, these are not orbits in the group theoretic sense, as we allow phases in Eq.~\eqref{pair orbit def}. However, this definition will suffice for our purposes, because as can be seen from Eq.~\eqref{phases cancel 3}, we will eventually choose the particular form of the basis states for $\mathcal{L}(\mathcal{H}_d)^{\otimes 3}$ in such a way that any phases from Eq.~\eqref{pair orbit def} are cancelled out.} if there exists $U \in \Cld$ such that $\mathcal{U}^{\otimes 2}\kett{D_{q_1,p_1}}\kett{D_{q_2,p_2}} = \alpha\kett{D_{q_1',p_1'}}\kett{D_{q_2',p_2'}}$ for some phase $\alpha$,  or equivalently, \begin{equation} \label{pair orbit def} \mathcal{U}\kett{D_{q_1,p_1}} = \alpha_1\kett{D_{q_1',p_1'}} \quad \text{ and } \quad \mathcal{U}\kett{D_{q_2,p_2}} = \alpha_2\kett{D_{q_2',p_2'}} \end{equation} for some phases $\alpha_1$ and $\alpha_2$. The next lemma is essentially the analogue of Lemma~\ref{lem: 2-fold orbit}, but for pairs of Paulis. 

\begin{lemma} \label{3-fold orbit}
The set of pairs of Paulis $\{(D_{q_1,p_1},D_{q_2,p_2}): q_1,p_1,q_2,p_2 \in \Zd\}$ is partitioned into the following orbits (in the sense of Eq.~\eqref{pair orbit def}):
\begin{enumerate}[a)]
    \item $\{(D_{0,0},D_{0,0})\}$ ($1$ element)
    \item $\{(D_{0,0}, D_{q,p}): q,p \in \Zd, (q,p) \neq (0,0)\}$ ($d^2 - 1$ elements)
    \item $\{(D_{q,p},D_{0,0}):q,p \in \Zd, (q,p) \neq (0,0)\}$ ($d^2 - 1$ elements),
    \item $\{(D_{q,p}, D_{kq,kp}): q,p \in \Zd, (q,p) \neq (0,0) \}$ for $k \in \{1,\dots, d-1\}$ ($d^2- 1$ elements for each $k$)
    \item $\{(D_{q_1,p_1},D_{q_2,p_2}): q_1,p_1,q_2,p_2 \in \Zd, p_1q_2 - q_1 p_2 = l\}$ for $l \in \{1, \dots, d-1\}$ ($d(d^2 - 1)$ elements for each $l$).
\end{enumerate}
\end{lemma}
\begin{proof}
    First, we show that these subsets indeed partition the set of all pairs of Paulis. The subsets in a), b) and c) are clearly mutually disjoint. The remaining pairs are of the form $(D_{q_1,p_1}, D_{q_2,p_2})$ with $(q_1,p_1), (q_2,p_2) \neq 0$. The pairs with $p_1 q_2 - q_1 p_2 \neq 0$ can be partitioned into the subsets in e). This leaves pairs $(D_{q_1,p_1}, D_{q_2,p_2})$ with $(q_1,p_1), (q_2,p_2) \neq 0$ and $p_1 q_2 - q_1 p_2 = 0$. Since $(q_1,p_1), (q_2,p_2) \neq 0$, we must have that $q_1, q_2 \neq 0$ and/or $p_1,p_2 \neq 0$. If $q_1,q_2 \neq 0$, then we can write $q_2 = (q_1^{-1}q_2) q_1$, and from $p_1q_2 - q_1p_2 = 0$, we have $p_2 = (q_1^{-1}q_2)p_1$. If $p_1,p_2 \neq 0$, then $q_2 = (p_1^{-1}p_2)q_1$ and $p_2 = (p_1^{-1}p_2)p_1$. In both cases, $q_2 = kq_1$ and $p_2 = kp_1$ for some $k \in \{1,\dots, d-1\}$, which coincides with the subsets in d). It remains to check that these subsets are indeed orbits.

    \begin{enumerate}[a)]
    \item Since $\mathcal{U}\kett{D_{0,0}} = \kett{D_{0,0}}$ (recall $D_{0,0} = I$) for any $U \in \Cld$, it is clear that $\{(D_{0,0}, D_{0,0})\}$ is an orbit. 
    \item Since $\mathcal{U}\kett{D_{0,0}} = \kett{D_{0,0}}$ for any $U\in \Cld$, $(D_{q_1',p_1'}, D_{q_2',p_2'})$ is in the same orbit as $(D_{0,0},D_{q,p})$ if and only if $(q_1',p_1') = (0,0)$ and there exists $U \in \Cld$ such that $\mathcal{U}\kett{D_{q,p}} \propto \kett{D_{q_2',p_2'}}$. By Lemma~\ref{lem: 2-fold orbit}, for $(q,p) \neq (0,0)$, the latter is true whenever $(q_2', p_2') \neq 0$. Hence, $\{(D_{0,0}, D_{q,p}): q,p \in \Zd, (q,p) \neq (0,0)\}$ is an orbit.
    \item Same proof as for b). 
    \item For $(q,p) \neq (0,0)$ and $k \neq 0$, $(D_{q_1',p_1'}, D_{q_2',p_2'})$ is in the same orbit as $(D_{q,p},D_{kq,kp})$ if and only if there exists $U \in \Cld$ such that $\mathcal{U}\kett{D_{q,p}} = \kett{D_{q_1',p_1'}}$ and $\mathcal{U}\kett{D_{kq,kp}} = \kett{D_{q_2',p_2'}}$. By Lemma~\ref{lem: 2-fold orbit}, there exists $U \in \Cld$ such that $\mathcal{U}\kett{D_{q,p}} = \kett{D_{q_1',p_1'}}$ whenever $(q_1',p_1') \neq (0,0)$. Then, since $D_{kq,kp} = D_{q,p}^k$, we have $\mathcal{U}(D_{kq,kp}) \propto U^\dagger (D_{q,p}^k)U = D_{q_1',p_1'}^k = D_{kq_1',kp_1'}$. Thus, all elements in the same orbit as $(D_{q,p},D_{kq,kp})$ are of the form $(D_{q_1',p_1'}, D_{kq_1',kp_2'})$.
    \item Finally, we show that for $p_1q_2 - q_1p_2 \neq 0$, $(D_{q_1',p_1'}, D_{q_2',p_2'})$ is in the same orbit as $(D_{q_1,p_1},D_{q_2,p_2})$---i.e., there exists $U \in \Cld$ satisfying Eq.~\eqref{pair orbit def}---if and only if $p_1' q_2' - q_1' p_2' = p_1q_2 - q_1p_2$. It follows from Facts~\ref{fact: Clifford symplectic} and~\ref{fact: 2x2 symplectic} that there exists such a Clifford $U$ if and only if there exists a $2\times 2$ matrix $C$ such that 
    \begin{equation} \label{C equation} C\begin{pmatrix} q_1 &q_2 \\ p_1 &p_2 \end{pmatrix} = \begin{pmatrix} q_1' &q_2' \\ p_1' &p_2' \end{pmatrix} \end{equation}
    and $\det(C) = 1$. Since $\det\begin{pmatrix} q_1 &q_2 \\ p_1 &p_2 \end{pmatrix} = q_1p_2 - p_1 q_2\neq 0$ by assumption, $\begin{pmatrix} q_1 &q_2 \\ p_1 &p_2 \end{pmatrix}$ is invertible and Eq.~\eqref{C equation} is equivalent to $C = \begin{pmatrix} q_1' &q_2' \\ p_1' &p_2' \end{pmatrix}\begin{pmatrix} q_1 &q_2 \\ p_1 &p_2 \end{pmatrix}^{-1}$. Then, 
    \[ \det(C) = \det\begin{pmatrix} q_1' &q_2' \\ p_1' &p_2' \end{pmatrix}\det \begin{pmatrix} q_1 &q_2 \\ p_1 &p_2 \end{pmatrix}^{-1} = (q_1' p_2' - p_1' q_2')(q_1p_2 - p_1 q_2)^{-1}, \]
    which is equal to $1$ if and only if $p_1' q_2' - q_1' p_2' = p_1q_2 - q_1p_2$.
    \end{enumerate}
\end{proof}

We are now ready to evaluate $\mathcal{E}_d^{(3)}$, proving Theorem~\ref{thm: twirl channels}(iii). Observe that the operators (Eq.~\eqref{Phi1} - \eqref{Upsilon_l}) appearing in our final expression for $\mathcal{E}_d^{(3)}$ corresponds to the non-trivial orbits we found in Lemma~\ref{3-fold orbit}. For each orbit of the form $\{(D_{q_1,p_1}D_{q_2,p_2})\}$, we have an operator that is an equal superposition over states that are proportional to $\kett{D_{q_1,p_1}}\kett{D_{q_2,p_2}}\kett{D_{-q_1-q_2,-p_1-p_2}}$ (and recall that states of this form are the states not ruled out by Lemma~\ref{lem: 3-fold kernel}). In this sense, $\kett{\Phi_1}$ corresponds to the orbit in Lemma~\ref{lem: 3-fold kernel}b), $\kett{\Phi_2}$ corresponds to the orbit in c), $\kett{\Phi_3}$ corresponds to the orbit in d) for $k = d-1$, the $\kett{\Psi_k}$'s
correspond to the orbits in d) for $k \neq d-1$, and the $\kett{\Upsilon_l}$'s correspond to the orbits in e).

\begin{proof}[Proof of Theorem~\ref{thm: twirl channels}(iii)]
Inserting resolutions of the identity (Eq.~\eqref{resolution of superidentity}) and using Lemma~\ref{lem: 3-fold kernel}, we have
\begin{align*}
    \mathcal{E}_d^{(3)} &= \sum_{q_1,p_1, q_2, p_2, q_1', p_1', q_2', p_2' \in \Zd} \braa{D_{q_1,p_1}}\braa{D_{q_2,p_2}}\braa{D_{-q_1-q_2,-p_1-p_2}} \mathcal{E}_d^{(3)} \kett{D_{q_1',p_1'}}\kett{D_{q_2',p_2'}}\kett{D_{-q_1'-q_2',-p_1'-p_2'}} 
    \\ 
    &\qquad\qquad \qquad\qquad \qquad\times \kett{D_{q_1,p_1}}\kett{D_{q_2,p_2}}\kett{D_{-q_1-q_2,-p_1-p_2}}\braa{D_{q_1',p_1'}}\braa{D_{q_2',p_2'}}\braa{D_{-q_1'-q_2',-p_1'-p_2'}}\\
    &= \sum_{q_1,p_1, q_2, p_2, q_1', p_1', q_2', p_2' \in \Zd} \braa{D_{q_1,p_1}}\braa{D_{q_2,p_2}}\braa{D_{q_1,p_1}^\dagger D_{q_2,p_2}^\dagger} \mathcal{E}_d^{(3)} \kett{D_{q_1',p_1'}}\kett{D_{q_2',p_2'}}\kett{D_{q_1',p_1'}^\dagger D_{q_2',p_2'}^\dagger} 
    \\ 
    &\qquad\qquad \qquad\qquad \qquad\times \kett{D_{q_1,p_1}}\kett{D_{q_2,p_2}}\kett{D_{q_1,p_1}^\dagger D_{q_2,p_2}^\dagger}\braa{D_{q_1',p_1'}}\braa{D_{q_2',p_2'}}\braa{D_{q_1',p_1'}^\dagger D_{q_2',p_2'}^\dagger},
\end{align*}
using the fact that $D_{-q_1-q_2, -p_1-p_2} \propto D_{-q_1, -p_1}D_{-q_2,-p_2} = D_{q_1,p_1}^\dagger D_{q_2,p_2}^\dagger$ in the second equality. 

Next, note from the definition of the orbits (Eq.~\eqref{pair orbit def}) that there exists $U \in \Cld$ such that $\braa{D_{q_1,p_1}}\braa{D_{q_2,p_2}}\braa{D_{q_1,p_1}^\dagger D_{q_2,p_2}^\dagger} \mathcal{U}^{\otimes 3} \kett{D_{q_1',p_1'}}\kett{D_{q_2',p_2'}}\kett{D_{q_1',p_1'}^\dagger D_{q_2',p_2'}^\dagger} \neq 0$ if and only if $(D_{q_1,p_1},D_{q_2,p_2})$ and $(D_{q_1',p_1'},D_{q_2',p_2'})$ are in the same orbit. Hence, since since $\mathcal{E}_d^{(3)}$ is the average of $\mathcal{U}^{\otimes 3}$ for $U \in \Cld$, $\braa{D_{q_1,p_1}}\braa{D_{q_2,p_2}}\braa{D_{q_1,p_1}^\dagger D_{q_2,p_2}^\dagger} \mathcal{E}_d^{(3)} \kett{D_{q_1',p_1'}}\kett{D_{q_2',p_2'}}\kett{D_{q_1',p_1'}^\dagger D_{q_2',p_2'}^\dagger} \neq 0$ only if $(D_{q_1,p_1},D_{q_2,p_2})$ and $(D_{q_1',p_1'},D_{q_2',p_2'})$ are in the same orbit, so we can write 
\begin{align*}
    \mathcal{E}_d^{(3)} &= \sum_{\text{orbits}} \sum_{\substack{(D_{q_1,p_1},D_{q_2,p_2}),\\ (D_{q_1',p_1'},D_{q_2',p_2'}) \in \text{ orbit}}} \braa{D_{q_1,p_1}}\braa{D_{q_2,p_2}}\braa{D_{q_1,p_1}^\dagger D_{q_2,p_2}^\dagger} \mathcal{E}_d^{(3)} \kett{D_{q_1',p_1'}}\kett{D_{q_2',p_2'}}\kett{D_{q_1',p_1'}^\dagger D_{q_2',p_2'}^\dagger} 
    \\ 
    &\qquad\qquad \qquad\qquad \qquad\times \kett{D_{q_1,p_1}}\kett{D_{q_2,p_2}}\kett{D_{q_1,p_1}^\dagger D_{q_2,p_2}^\dagger}\braa{D_{q_1',p_1'}}\braa{D_{q_2',p_2'}}\braa{D_{q_1',p_1'}^\dagger D_{q_2',p_2'}^\dagger}. 
\end{align*}

Now, we show that the overlap $\braa{D_{q_1,p_1}}\braa{D_{q_2,p_2}}\braa{D_{q_1,p_1}^\dagger D_{q_2,p_2}^\dagger} \mathcal{E}_d^{(3)} \kett{D_{q_1',p_1'}}\kett{D_{q_2',p_2'}}\kett{D_{q_1',p_1'}^\dagger D_{q_2',p_2'}^\dagger}$ is the same within each orbit, i.e., for any orbit, there is a constant $c$ such that
\[ \braa{D_{q_1,p_1}}\braa{D_{q_2,p_2}}\braa{D_{q_1,p_1}^\dagger D_{q_2,p_2}^\dagger} \mathcal{E}_d^{(3)} \kett{D_{q_1',p_1'}}\kett{D_{q_2',p_2'}}\kett{D_{q_1',p_1'}^\dagger D_{q_2',p_2'}^\dagger} = c\]
for all $(D_{q_1,p_1},D_{q_2,p_2})$ and $(D_{q_1',p_1'},D_{q_2',p_2'})$ in that orbit. To see this, note from the definition of orbits that for any other $(D_{q_1'',p_1''},D_{q_2'',p_2''})$ in the orbit,
there exist $U \in \Cld$ such that $\mathcal{U}^{\otimes 2} \kett{D_{q_1',p_1'}}\kett{D_{q_2',p_2'}} = \alpha\kett{D_{q_1'',p_1''}}\kett{D_{q_2'',p_2''}}$ for some phase $\alpha$. Then,
\begin{align} \mathcal{U}^{\otimes 3}\kett{D_{q_1',p_1'}}\kett{D_{q_2',p_2'}}\kett{D_{q_1',p_1'}^\dagger D_{q_2',p_2'}^\dagger} &= \alpha \kett{D_{q_1'',p_1''}}\kett{D_{q_2'',p_2''}} \alpha^*\kett{D_{q_1'',p_1''}^\dagger D_{q_2'',p_2''}^\dagger} \nonumber \\
&= \kett{D_{q_1'',p_1''}}\kett{D_{q_2'',p_2''}} \kett{D_{q_1'',p_1''}^\dagger D_{q_2'',p_2''}^\dagger}, \label{phases cancel 3} 
\end{align}
so we can use Fact~\ref{fact: invariance} to write
\begin{align*}
    \mathcal{E}_d^{(3)} \kett{D_{q_1',p_1'}}\kett{D_{q_2',p_2'}}\kett{D_{q_1',p_1'}^\dagger D_{q_2',p_2'}^\dagger} &= \mathcal{E}_d^{(3)} \mathcal{U}^{\otimes 3}\kett{D_{q_1',p_1'}}\kett{D_{q_2',p_2'}}\kett{D_{q_1',p_1'}^\dagger D_{q_2',p_2'}^\dagger} \\
    &= \mathcal{E}_d^{(3)} \kett{D_{q_1'',p_1''}}\kett{D_{q_2'',p_2''}}\kett{D_{q_1'',p_1''}^\dagger D_{q_2'',p_2''}^\dagger}.
\end{align*}
Thus, noting that the states defined in Eqs.~\eqref{Phi1} - \eqref{Upsilon_l} are precisely the equal superpositions over the $\kett{D_{q_1,p_1}}\kett{D_{q_2,p_2}}\kett{D_{q_1,p_1}^\dagger D_{q_2,p_2}^\dagger}$ for each orbit, we have
\begin{align} \label{Ed3 with constants}
\mathcal{E}_d^{(3)} = a \kett{I}\braa{I}^{\otimes 3} + b_1\kett{\Phi_1}\braa{\Phi_1} + b_2\kett{\Phi_2}\braa{\Phi_2}+ b_3\kett{\Phi_3}\braa{\Phi_3} + \sum\limits_{k=1}^{d-2} c_k\kett{\Psi_k}\braa{\Psi_k} + \sum\limits_{l = 1}^{d-1} d_l\kett{\Upsilon_l}\braa{\Upsilon_l}
\end{align}
for constants $a, b_1, b_2, b_3, c_k, d_l$. Since $\mathcal{E}_d^{(3)}$ is a projector (Fact~\ref{fact: projector}), each of these constants must be either $0$ or $1$. Finally, we see from Eq.~\eqref{phases cancel 3} and the definition of orbits (Eq.~\eqref{pair orbit def}) that for any $U \in \Cld$, $\mathcal{U}^{\otimes 3}$ permutes the basis states of the form $\kett{D_{q_1,p_1}}\kett{D_{q_2,p_2}}\kett{D_{q_1,p_1}^\dagger D_{q_2,p_2}^\dagger}$ within each orbit, without incurring any phases. Since each of the states appearing in Eq.~\eqref{Ed3 with constants} is an equal superposition over these $\kett{D_{q_1,p_1}}\kett{D_{q_2,p_2}}\kett{D_{q_1,p_1}^\dagger D_{q_2,p_2}^\dagger}$ states  for each orbit, it follows that $\mathcal{U}^{\otimes 3}$ stabilises each of these states, and hence so does $\mathcal{E}_d^{(3)}$. Therefore, each of the constants in Eq.~\eqref{Ed3 with constants} must be equal to $1$.
\end{proof}

\subsection{Formulae for generalized Clifford shadows} 

Having evaluating the $2$- and $3$-fold twirl channels for $\Cld$ (Theorem~\ref{thm: twirl channels}), we can now characterise the classical shadows corresponding to the uniform distribution over generalized Cliffords, in any prime dimension. 

\subsubsection{Measurement channel}

First, we determine the classical shadows measurement channel $\mathcal{M}$, by substituting our expression for the $2$-fold twirl channel from Theorem~\ref{thm: twirl channels}(ii) into Eq.~\eqref{measurement channel}. 

\begin{theorem}[Generalized Clifford shadows measurement channel] \label{thm: M}
    Let $\mathcal{M}$ be the quantum channel defined in Eq.~\eqref{measurement channel 1}. When $\mathcal{D}$ is the uniform distribution over $\Cld$ for some prime $d$ and $\mathcal{B}$ is the computational basis $\{\ket{j}:j \in \Zd\}$, $\mathcal{M}$ is given by
    \begin{equation} \label{Clifford M} \mathcal{M}(A) = \frac{1}{d+1}(\tr(A) I + A), \end{equation}
    for an arbitrary operator $A$ (equivalently, $\mathcal{M} = \frac{d}{d+1}\kett{I}\braa{I} + \frac{1}{d+1} \mathcal{I}$). Therefore, $\mathcal{M}$ is invertible, with inverse given by
    \begin{equation} \label{M inverse} \mathcal{M}^{-1}(A) = (d+1)A -\tr(A) I. \end{equation}
\end{theorem}
\begin{proof}
    Since $\E_{U \sim \mathcal{D}} \mathcal{U}^{\otimes 2} = \mathcal{E}_d^{(2)}$ when $\mathcal{D}$ is the uniform distribution over $\Cld$, Eq.~\eqref{measurement channel} becomes
    \[ \mathcal{M}(A) = \tr_1 \left[\sum_{j \in \Zd} \mathcal{E}_d^{(2)}(\ket{j}\bra{j}^{\otimes 2}) (A \otimes I) \right]. \]
    We can simplify this by observing that $\mathcal{E}_d^{(2)}(\ket{j}\bra{j}^{\otimes 2}) = \mathcal{E}_d^{(2)}(\ket{0}\bra{0}^{\otimes 2})$ for any $j \in \Zd$. This is because for any $j \in \Zd$, there exists $U \in \Cld$ (specifically, $U = (X^j)^\dagger$ such that $\mathcal{U}(\ket{j}\bra{j}) = \ket{0}\bra{0}$. Hence, by Fact~\ref{fact: invariance}, $\mathcal{E}_d^{(2)}(\ket{j}\bra{j}^{\otimes 2}) = \mathcal{E}_d^{(2)} \circ \mathcal{U}^{\otimes 2}(\ket{j}\bra{j}^{\otimes 2}) = \mathcal{E}_d^{(2)}(\ket{0}\bra{0}^{\otimes 2})$. Thus,
    \begin{equation} \label{M E0}  \mathcal{M}(A) = d \, \tr_1\left[\mathcal{E}_d^{(2)} (\ket{0}\bra{0}^{\otimes 2}) (A \otimes I) \right].\end{equation}

    We now calculate $\mathcal{E}_d^{(2)}(\ket{0}\bra{0}^{\otimes 2})$ using Theorem~\ref{thm: twirl channels}(ii). From Eq.~\eqref{Z}, it is easy to see that
    \[ \ket{0}\bra{0} = \frac{1}{d}\sum_{z \in \Zd} Z^z = \frac{1}{d}\sum_{z \in \Zd} D_{0,z}. \] Defining $\kett{\Pi_0} \equiv \ket{0}\bra{0}$ for convenience, this becomes
    \begin{equation} \label{Pi0 to D} \kett{\Pi_0} = \frac{1}{\sqrt{d}}\sum_{z \in \Zd}\kett{D_{0,z}},\end{equation}
    and we want to find 
    \begin{align*}
        \mathcal{E}_d^{(2)}\kett{\Pi_0}^{\otimes 2} &= (\kett{I}\braa{I}^{\otimes 2} + \kett{\Phi_0}\braa{\Phi_0}) \kett{\Pi_0}^{\otimes 2}.
    \end{align*}
    From Eqs.~\eqref{Liouville Dqp} and~\eqref{Pi0 to D}, we obtain
    \begin{align} \label{overlap D Pi0}
        \braakett{D_{q,p}}{\Pi_0} = \braakett{D_{q,p}^\dagger}{\Pi_0} &= \frac{1}{\sqrt{d}} \delta_{q,0},
    \end{align}
    so we have
    \[ \braa{I}^{\otimes 2}\kett{\Pi_0}^{\otimes 2} = \braakett{D_{0,0}}{\Pi_0}^2 = \frac{1}{d} \]
    and 
    \begin{align*}
        \braa{\Phi_0}(\kett{\Pi_0}^{\otimes 2}) &= \frac{1}{\sqrt{d^2-1}}\sum_{\substack{q,p \in \Zd \\ (q,p) \neq (0,0)}} \braa{D_{q,p}}\braa{D_{q,p}^\dagger} (\kett{\Pi_0}^{\otimes 2}) = \frac{1}{\sqrt{d^2-1}}\frac{d-1}{d}.
    \end{align*}
    Thus,
    \begin{align*}
        \mathcal{E}_d^{(2)}\kett{\Pi_0}^{\otimes 2} &= \frac{1}{d}\kett{I}\kett{I} + \frac{1}{d^2 - 1}\frac{d-1}{d} \sum_{\substack{q,p \in \Zd \\ (q,p) \neq (0,0)}} \kett{D_{q,p}}\kett{D_{q,p}^\dagger} \\
        &= \frac{1}{d^2} I \otimes I + \frac{1}{d^2(d+1)}\sum_{\substack{q,p \in \Zd \\ (q,p) \neq (0,0)}} D_{q,p} \otimes D_{q,p}^\dagger \\
        &= \frac{1}{d(d+1)} I \otimes I + \frac{1}{d^2(d+1)}\sum_{q,p \in \Zd}D_{q,p}\otimes D_{q,p}^\dagger.
    \end{align*}
    Substituting this into Eq.~\eqref{M E0} gives
    \begin{align*}
        \mathcal{M}(A) 
        &= \frac{1}{d+1}\tr(A)I + \frac{1}{d(d+1)}\sum_{q,p \in \Zd}\tr(D_{q,p}A) D_{q,p}^\dagger \\
        &= \frac{1}{d+1}\tr(A) I + \frac{1}{d+1} A,
    \end{align*}
    where we use the fact that $\{\frac{1}{\sqrt{d}} D_{q,p}^\dagger\}_{q,p \in \Zd}$ forms a Hilbert-Schmidt orthonormal basis to obtain the last equality.  
\end{proof}

Thus, for the uniform distribution over generalized Cliffords (in prime dimension), $\mathcal{M}$ and $\mathcal{M}^{-1}$ are depolarising channels. 
This is reminiscent of the fact that the measurement channel for the classical shadows associated with $n$-qubit Clifford circuits, analysed in Ref.~\cite{huang2020predicting}, is also a depolarising channel. In fact, the depolarising parameter in that case has the same dependence on the dimension, i.e., the measurement channel has the same form as that in Eq.~\eqref{Clifford M}, but with $d$ replaced by $2^n$. Note, however, that $n$-qubit Clifford circuits are different from generalized single-qudit Cliffords with $d = 2^n$ (except for $n = 1$). Actually, it can be shown that the measurement channel for generalized Cliffords with $d = 2^n$ is not a depolarising channel, though this calculation is outside the scope of this paper. 

Substituting Eq.~\eqref{M inverse} into Eq.~\eqref{hat rho}, the classical shadow estimator when $D$ is the uniform distribution over $\Cld$ has the form
\[ \hat{\rho} = (d+1) \hat{U}^\dagger \ket{\hat{b}}\bra{\hat{b}}\hat{U} - I, \]
and the estimator for $\tr(O\rho)$ for any observable $O$ is given by
\begin{align}
    \hat{o} = \tr(O\mathcal{M}^{-1}(\hat{U}^\dagger \ket{\hat{b}}\bra{\hat{b}}\hat{U})) &= (d+1) \tr(O \hat{U}^\dagger \ket{\hat{b}}\bra{\hat{b}}\hat{U}) - \tr(O) \\
    &= (d+1) \bra{\hat{b}}\hat{U} O \hat{U}^\dagger\ket{b} - \tr(O). \label{shadows postprocessing}
\end{align}

\subsubsection{Variance}

With both $\mathcal{E}_d^{(3)}$ and $\mathcal{M}^{-1}$ in hand, we can now evaluate the variance, Eq.~\eqref{shadows variance}, for any operator $O$. 

\begin{theorem}[Generalized Clifford shadows variance] \label{thm: shadows variance}
    Let $\hat{o}$ denote the classical shadows estimator for the expectation value $\tr(O\rho)$, defined as in Eq.~\eqref{hat oi}. When $\mathcal{D}$ is the uniform distribution over $\Cld$ for some prime $d$ and $\mathcal{B}$ is the computational basis $\{\ket{j}: j \in \Zd\}$, the variance $\Var[\hat{o}]$ of $\hat{o}$ for any observable $O$ is given by
    \begin{equation} \label{Varhato} \Var[\hat{o}] = \frac{d+1}{d}\|O_0\|_{\text{H-S}}^2 + \frac{d+1}{d^2}\sum_{k=1}^{d-2}\sum_{\substack{q,p\in\Zd\\ (q,p) \neq (0,0)}}\tr(\widetilde{D}_{q,p}^\dagger O_0)\tr(\widetilde{D}_{kq,kp}^\dagger O_0^\dagger) \tr(\widetilde{D}_{kq,kp}\widetilde{D}_{q,p}\rho) - |\tr(O_0\rho)|^2,\end{equation}
    where $O_0 \coloneqq O - \tr(O)I/d$ denotes the traceless part of $O$, and $\widetilde{D}_{q,p} \coloneqq U D_{q,p} U^\dagger$ for any choice of $U \in \Cld$.
\end{theorem}
\begin{proof}
    From Eq.~\eqref{shadows variance shifted}, we have
    \begin{align*}
        \Var[\hat{o}] &= \E[|\hat{o}|^2] - |\E[\hat{o}]|^2 \\
        &= \tr\left[\sum_{j \in \Zd} \mathcal{E}_d^{(3)}(\ket{j}\bra{j}^{\otimes 3}) \left(\mathcal{M}^{-1}(O_0) \otimes \mathcal{M}^{-1}(O_0^\dagger) \otimes \rho\right) \right] - |\tr(O_0\rho)|^2,
    \end{align*}
    where $\mathcal{M}^{-1}$ in this case is given by Eq.~\eqref{M inverse} in Theorem~\ref{thm: M}. For essentially the same reason that $\mathcal{E}_d^{(2)}(\ket{j}\bra{j}^{\otimes 2}) = \mathcal{E}_d^{(2)}(\ket{0}\bra{0}^{\otimes 2})$ for any $j \in \Zd$ (see the proof of Theorem~\ref{thm: M}), we also have $\mathcal{E}_d^{(3)}(\ket{j}\bra{j}^{\otimes 3}) = \mathcal{E}_d^{(3)}(\ket{0}\bra{0}^{\otimes 3})$
    for any $j \in \Zd$. Thus, the first term becomes
    \begin{align*}
    \E[|\hat{o}|^2] &= d\, \tr\left[ \mathcal{E}_d^{(3)}(\ket{0}\bra{0}^{\otimes 3}) \left(\mathcal{M}^{-1}(O_0) \otimes \mathcal{M}^{-1}(O_0^\dagger) \otimes \rho\right) \right].
    \end{align*}
    Now, we notate $\kett{\Pi_0} \equiv \ket{0}\bra{0}$, $\kett{\rho} \equiv \rho$, $\kett{O} \equiv O$, $\kett{O}^\dagger \equiv O^\dagger$ in Liouville notation (note that here, we deviate from the convention in Eq.~\eqref{Liouville1}, in that $\kett{\rho}$, $\kett{O}$, and $\kett{O}^\dagger$ are not necessarily normalised). Then, since $\mathcal{E}_d^{(3)}(\ket{0}\bra{0}^{\otimes 3})$ is Hermitian, we can rewrite
    \begin{align*}
        \E[|\hat{o}|^2] &= d\, \braa{\Pi_0}^{\otimes 3} \mathcal{E}_d^{(3)\dagger}(\mathcal{M}^{-1}\kett{O_0} \otimes \mathcal{M}^{-1}\kett{O_0^\dagger} \otimes \kett{\rho}).
    \end{align*}
    We begin by calculating $\mathcal{E}_d^{(3)}$. By Theorem~\ref{thm: twirl channels}(iii), 
    \begin{align*}
        \mathcal{E}_d^{(3)}\kett{\Pi_0}^{\otimes 3} = \left(\kett{I}\braa{I}^{\otimes 3} + \kett{\Phi_1}\braa{\Phi_1} + \kett{\Phi_2}\braa{\Phi_2}+ \kett{\Phi_3}\braa{\Phi_3} + \sum\limits_{k=1}^{d-2} \kett{\Psi_k}\braa{\Psi_k} + \sum\limits_{l = 1}^{d-1} \kett{\Upsilon_l}\braa{\Upsilon_l}\right)\ket{\Pi_0}^{\otimes 3}.
    \end{align*}
    Using Eqs.~\eqref{Pi0 to D} and~\eqref{overlap D Pi0}, we find
    \begin{align*}
        &\braa{I}^{\otimes 3}\kett{\Pi_0}^{\otimes 3} = \frac{1}{d^{3/2}}, \\
        &\braa{\Phi_i}(\kett{\Pi_0}^{\otimes 3}) = \frac{1}{d^{3/2}}\frac{d-1}{\sqrt{d^2-1}} \qquad \text{for $i \in \{1,2,3\}$}, \\
        &\braa{\Psi_k}(\kett{\Pi_0}^{\otimes 3}) = \frac{1}{d^{3/2}}\frac{d-1}{\sqrt{d^2-1}} \qquad \text{for $k \in \{1,\dots,d-2\}$}, \\
        &\braa{\Upsilon_l}(\kett{\Pi_0}^{\otimes 3}) = 0 \qquad \text{for $l \in \{1,\dots, d-1\}$},
    \end{align*}
    so
    \begin{align*}
        \mathcal{E}_d^{(3)}\kett{\Pi_0}^{\otimes 3} = \frac{1}{d^{3/2}}\kett{I}^{\otimes 3} + \frac{1}{d^{3/2}}\frac{d-1}{\sqrt{d^2-1}} \left(\kett{\Phi_1} + \kett{\Phi_2} + \kett{\Phi_3} + \sum_{k = 1}^{d-2}\kett{\Psi_k} \right).
    \end{align*}
    Hence, 
    \begin{align*}
    \E[|\hat{o}|^2] = \left[\frac{1}{\sqrt{d}}\braa{I}^{\otimes 3} + \frac{1}{\sqrt{d}}\frac{d-1}{\sqrt{d^2-1}} \left(\braa{\Phi_1} + \braa{\Phi_2} + \braa{\Phi_3} + \sum_{k = 1}^{d-2}\braa{\Psi_k} \right) \right](\mathcal{M}^{-1}\kett{O_0} \otimes \mathcal{M}^{-1}\kett{O_0^\dagger} \otimes \kett{\rho}).
    \end{align*}
    From Eq.~\eqref{M inverse}, we have
    \[ \mathcal{M}^{-1}\kett{D_{q,p}} = (d+1)\kett{D_{q,p}} - \delta_{q,0}\delta_{p,0}\kett{I}.\]
    Using this along with the fact that $\mathcal{M}^{-1}$ is self-adjoint, i.e., $\braa{A}\mathcal{M}^{-1}\kett{B} = \braakett{\mathcal{M}^{-1}(A)}{B}$, we obtain
    \begin{align*}
    &\frac{1}{\sqrt{d}}\braa{I}^{\otimes 3}(\mathcal{M}^{-1}\kett{O_0} \otimes \mathcal{M}^{-1}\kett{O_0^\dagger} \otimes \kett{\rho}) = \frac{1}{d^2}|\tr(O_0)|^2 = 0, \\
    &\frac{1}{\sqrt{d}}\frac{d-1}{\sqrt{d^2-1}}\braa{\Phi_1}(\mathcal{M}^{-1}\kett{O_0} \otimes \mathcal{M}^{-1}\kett{O_0^\dagger} \otimes \kett{\rho}) = \frac{1}{d}\tr(O_0)\tr(\rho O_0^\dagger) - \frac{1}{d^2}|\tr(O_0)|^2 = 0 \\
    &\frac{1}{\sqrt{d}}\frac{d-1}{\sqrt{d^2-1}}\braa{\Phi_2}(\mathcal{M}^{-1}\kett{O_0} \otimes \mathcal{M}^{-1}\kett{O_0^\dagger} \otimes \kett{\rho}) =  \frac{1}{d}\tr(O^\dagger) \tr(\rho O_0) - \frac{1}{d^2}|\tr(O_0)|^2 = 0\\
    &\frac{1}{\sqrt{d}}\frac{d-1}{\sqrt{d^2-1}}\braa{\Phi_3}(\mathcal{M}^{-1}\kett{O_0} \otimes \mathcal{M}^{-1}\kett{O_0^\dagger} \otimes \kett{\rho}) = \frac{d+1}{d}\tr(O_0^\dagger O_0) - \frac{d+1}{d^2}|\tr(O_0)|^2 = \frac{d+1}{d}\|O_0\|_{\text{H-S}}^2 \\
    &\frac{1}{\sqrt{d}}\frac{d-1}{\sqrt{d^2-1}} \braa{\Psi_k}(\mathcal{M}^{-1}\kett{O} \otimes \mathcal{M}^{-1}\kett{O_0^\dagger} \otimes \kett{\rho}) = \frac{d+1}{d^2}\sum_{\substack{p,q \in\Zd\\ (p,q) \neq (0,0)}} \tr(D_{q,p}^\dagger O_0)\tr(D_{kq,kp}^\dagger O_0^\dagger) \tr(D_{kq,kp} D_{q,p}\rho).
    \end{align*}
    Thus, we arrive at
    \begin{align} \label{E[o2]}
        \E[|\hat{o}|^2] 
        &= \frac{d+1}{d}\|O_0\|_{\text{H-S}}^2 + \frac{d+1}{d^2}\sum_{k=1}^{d-2}\sum_{\substack{q,p\in\Zd\\ (q,p) \neq (0,0)}}\tr(D_{q,p}^\dagger O_0)\tr(D_{kq,kp}^\dagger O_0^\dagger) \tr(D_{kq,kp}D_{q,p}\rho),
    \end{align}
    giving
    \[ \Var[\hat{o}] = \frac{d+1}{d}\|O_0\|_{\text{H-S}}^2 + \frac{d+1}{d^2}\sum_{k=1}^{d-2}\sum_{\substack{q,p\in\Zd\\ (q,p) \neq (0,0)}}\tr({D}_{q,p}^\dagger O_0)\tr({D}_{kq,kp}^\dagger O_0^\dagger) \tr({D}_{kq,kp}{D}_{q,p}\rho) - |\tr(O_0\rho)|^2. \]
To see that this actually holds with $D_{q,p}$ replaced by $\widetilde{D}_{q,p} \coloneqq U D_{q,p} U^\dagger$ for any $U \in \Cld$, note that since $\mathcal{U} \circ \mathcal{E}_d^{(3)} \circ \mathcal{U}^\dagger$ for any $U \in \Cld$ from Fact~\ref{fact: invariance}, we can replace $\kett{\Phi_i}$, $\kett{\Psi_k}$, $\kett{\Upsilon_l}$ with $\mathcal{U}\kett{\Phi_i}$, $\mathcal{U}\kett{\Psi_k}$, $\mathcal{U}\kett{\Upsilon_l}$, respectively, in all of our calculations. These have the same form as $\kett{\Phi_i}$, $\kett{\Psi_k}$, $\kett{\Upsilon_l}$ (Eqs.~\eqref{Phi1} - \eqref{Upsilon_l}), respectively, except with $D_{q,p}$ replaced by $\widetilde{D}_{q,p}$, thus leading to Eq.~\eqref{Varhato}.
\end{proof}

We can sanity-check this result by comparing to that for the uniform distribution over $n$-qubit Clifford circuits; $2$ is prime, so our results should coincide for $d = 2$ (i.e., $n =1$). 
For $d = 2$, Eq.~\eqref{E[o2]} becomes
\[ \E[|\hat{o}|^2] = \frac{3}{2}\|O_0\|^2_{\text{H-S}}, \]
since the outside sum in the second term is empty. On the other hand, For $n$-qubit Clifford circuits, Eq.~(S43) of Ref.~\cite{huang2020predicting} gives
$\E[|\hat{o}|^2] = \frac{2^n + 1}{2^n + 2} [\tr(O_0^2) + 2\tr(\rho O_0^2)]$,
 which is 
 \begin{equation} \label{Eo n2} \E[|\hat{o}|^2] =\frac{3}{4} [\tr(O_0^2) + 2\tr(\rho O_0^2)] \end{equation} for $n = 1$. Ref.~\cite{huang2020predicting} considers only Hermitian $O_0$, so the RHS is always real, and $\tr(O_0^2) = \|O_0\|^2_{\text{H-S}}$. Now note that for $n = 1$, $O_0$ has dimension $2$, so since it is traceless, we must have $O_0^2 = aI$ for some $a \in \mathbb{R}$. Thus, $\|O_0\|^2_{\text{H-S}} = \tr(O_0^2) = a\tr(I) = 2a$, and $\tr(\rho O_0^2) = a\tr(\rho) = a = \|O_0\|^2_{\text{H-S}}/2$, so we can write Eq.~\eqref{Eo n2}
 \[ \E[|\hat{o}|^2] = \frac{3}{4}\left[\|O_0\|_{\text{H-s}}^2 + 2\frac{\|O_0\|_{\text{H-s}}^2}{2}\right] = \frac{3}{2}\|O_0\|_{\text{H-s}}^2, \]
which matches our result.

\subsection{Some observables of interest}

In the section, we apply our general results, Theorems~\ref{thm: M} and~\ref{thm: shadows variance} to some natural observables to consider in the qudit setting.

\subsubsection{Transition elements in stabiliser bases}

First, we consider observables of the form $U\ket{j}\bra{i}U^\dagger$ for $U \in \Cld$ and computational basis states $\ket{i}$, $\ket{j}$. The expectation values of these with respect to an unknown state $\rho$ are $\tr(U\ket{j}\bra{i}U^\dagger) = \bra{i}U^\dagger \rho U\ket{j}$, which are the quantities considered in Theorem~\ref{thm:clifford_shadow}. 

\begin{proof}[Proof of Theorem~\ref{thm:clifford_shadow}]
Consider any generalized Clifford $U \in \Cld$, and any computational basis states $\ket{i}$ and $\ket{j}$ with $i \neq j$. We have
\[ U\ket{j}\bra{i}U^\dagger = U X^i \ket{j-i} \bra{0}(X^i)^\dagger U^\dagger = UX^i \ket{j'}\bra{0}(UX^i)^\dagger, \]
defining $j' \equiv j - i \enspace (\text{mod } d)$. Hence, since $UX^i \in \Cld$, we can choose $\widetilde{D}_{q,p} = (UX^i) D_{q,p} (UX^i)^\dagger$ in Theorem~\ref{thm: shadows variance},
to get that for $O = U\ket{j}\bra{i}U^\dagger$ (which is already traceless, for $j \neq i$),
\begin{align*}
    \Var[\hat{o}]\Big|_{O = U\ket{j}\bra{i}U^\dagger} &= \frac{d+1}{d}\|U\ket{j}\bra{i}U^\dagger\|_{\text{H-S}}^2 + \frac{d+1}{d^2}\sum_{k=1}^{d-2}\sum_{\substack{q,p\in\Zd\\ (q,p) \neq (0,0)}}\tr(\widetilde{D}_{q,p}^\dagger U\ket{j}\bra{i}U^\dagger)\tr(\widetilde{D}_{kq,kp}^\dagger (U\ket{j}\bra{i}U^\dagger)^\dagger) \tr(\widetilde{D}_{kq,kp}\widetilde{D}_{q,p}\rho) \\
    &\quad - |\tr(U\ket{j}\bra{i}U^\dagger\rho)|^2 \\
    &= \frac{d+1}{d} + \frac{d+1}{d^2}\sum_{k=1}^{d-2}\sum_{\substack{q,p\in\Zd\\ (q,p) \neq (0,0)}} \tr(D_{q,p}^\dagger \ket{j'}\bra{0}) \tr(D_{kq,kp}^\dagger \ket{0}\bra{j'})\tr(\widetilde{D}_{kq,kp}\widetilde{D}_{q,p}\rho) -|\tr(U\ket{j}\bra{i}U^\dagger\rho)|^2 \\
    &\leq \frac{d+1}{d} + \frac{d+1}{d^2}\sum_{k=1}^{d-2}\sum_{\substack{q,p\in\Zd\\ (q,p) \neq (0,0)}} |\bra{0}D_{q,p}^\dagger\ket{j'} \bra{j'}D_{kq,kp}^\dagger\ket{0}|,
\end{align*}
where in the second line, we use the fact that 
\[ \tr(\widetilde{D}_{q,p}^\dagger U\ket{j}\bra{i}U^\dagger) = \tr((UX^i)D_{q,p}(UX^i)^\dagger UX^i \ket{j'}\bra{0}(UX_i)^\dagger) = \tr(\ket{j'}\bra{0}) \]
for our choice of $\widetilde{D}_{q,p}$, and in the third line we use $|\tr(\widetilde{D}_{kq,kp}\widetilde{D}_{q,p}\rho)| \leq 1$.
Now, from Eq.~\eqref{Dqp act}, we have
\begin{equation} \label{Dqp ketbra} D_{q,p} = \omega^{qp/2} \sum_{j\in\Zd}\omega^{jp}\ket{j + q} \bra{j}, \end{equation}
from which we obtain
$|\bra{0} D_{q,p}^\dagger\ket{j'} = \delta_{q,j'}$
and 
$|\bra{j'}D_{kq,kp}^\dagger\ket{0} = \delta_{kq,-j'}$, so
\begin{align*}
    \Var[\hat{o}]\Big|_{O = U\ket{j}\bra{i}U^\dagger} \leq \frac{d + 1}{d} + \frac{d + 1}{d^2} \sum_{k=1}^{d-2} \sum_{\substack{q,p\in\Zd\\ (q,p) \neq (0,0)}} \delta_{q,j'}\delta_{kq,-j'} = \frac{d + 1}{d} + \frac{d+1}{d^2}\sum_{k=1}^{d-2} \sum_{p = 1}^{d-1} \delta_{kj', -j'}.
\end{align*}
Since $i \neq j$, we have $j' = j - i \neq 0$, so $kj' = -j'$ only for $k = d-1$, which is not included in the sum over $k$. Therefore, the second term vanishes, and we have 
\[ \mathrm{Var}[\hat{o}]\Big|_{O = U\ket{j}\bra{i}U^\dagger} \leq \frac{d +1}{d} 
< 2, \]
for any choice of $U \in \Cld$, and $\ket{i}, \ket{j}$ with $i \neq j$. 

Substituting this into Eq.~\eqref{Nsample}, it follows that 
\[ N_{\text{sample}} = \mathcal{O}\left(\frac{\log(M/\delta)}{\eps^2} \right) \]
single-copy measurements of $\rho$ suffice to estimate any $M$ observables of the form $\tr(U\ket{j}\bra{i}U^\dagger \rho) = \bra{i}U^\dagger \rho U\ket{j}$ with $U \in \Cld$ and $i \neq j$, within additive error $\eps$ with high probability. Since $|\Cld| = \mathcal{O}(d^5)$~\cite{gross2006hudson}, there are $\mathcal{O}(d^7)$ possible such observables, $N_{\text{sample}} = \mathcal{O}(\log(d)/\varepsilon^2)$ suffice to estimate all of them. 

By Eq.~\eqref{shadows postprocessing}, if $U' \in \Cld$ was the Clifford sampled in one iteration of the classical shadows protocol, and $\ket{b}$ was the computational basis state that was measured, the corresponding sample of the estimator $\hat{o}$ for $\bra{i} U^\dagger \rho U\ket{j}$ is
\begin{align*} \tr(U\ket{j}\bra{i}U^\dagger \mathcal{M}^{-1}(U'^\dagger \ket{b}\bra{b}U')) &= (d+1) \bra{b} U' U\ket{j}\bra{i}U^\dagger U'^\dagger \ket{b}.
\end{align*}
Since $\ket{b}, \ket{i},\ket{j}$ are computational basis states and $U'U \in \Cld$, each of $\bra{b}U'U\ket{j}$ and $\bra{i}U^\dagger U'^\dagger\ket{b}$ can be efficiently computed (classically) using the $d$-dimensional generalization of the Gottesman-Knill theorem~\cite{kocia2017discrete}, so the classical shadows samples can be efficiently computed. 
\end{proof}

On the other hand, for an observable of the form $U\ket{0}\bra{0}U^\dagger$ for $U \in \Cld$ (which subsumes all observables $U\ket{i}\bra{i}U^\dagger$ for $U \in \Cld$ and computational states $\ket{i}$), the traceless part is $U\ket{i}\bra{i}U^\dagger - I/d$, and choosing $\widetilde{D}_{q,p} = UD_{q,p}U^\dagger$ in Theorem~\ref{thm: shadows variance} gives
\begin{align*}
    \Var[\hat{o}]\Big|_{O = U\ket{0}\bra{0}U^\dagger} &= \frac{d+1}{d}\left(1 - \frac{1}{d}\right) \\
    &\quad + \frac{d+1}{d^2}\sum_{k=1}^{d-2} \sum_{\substack{q,p \in \Zd \\ (q,p) \neq (0,0)}} \tr\left(\widetilde{D}_{q,p}^\dagger \Big(U\ket{0}\bra{0} U^\dagger- \frac{I}{d}\Big)\right) \tr\left(\widetilde{D}_{kq,kp}^\dagger \Big(U\ket{0}\bra{0} U^\dagger- \frac{I}{d}\Big)\right) \tr(\widetilde{D}_{kq,kp}\widetilde{D}_{q,p}\rho) \\
    &\quad - \left|\tr\left(\Big(U\ket{0}\bra{0} U^\dagger- \frac{I}{d}\Big) \rho\right)\right|^2 \\
    &= 1 - \frac{1}{d^2} + \frac{d+1}{d^2} \sum_{k=1}^{d-2} \sum_{\substack{q,p \in \Zd\\ (q,p) \neq (0,0)}} \bra{0}D_{q,p}^\dagger \ket{0} \bra{0}D_{kq,kp}^\dagger \ket{0} \tr({D}_{kq,kp}{D}_{q,p}U^\dagger\rho U) - (\bra{0}U^\dagger \rho U \ket{0} - 1/d)^2.
\end{align*}
From Eq.~\eqref{Dqp ketbra}, 
\[ \bra{0}D_{q,p}^\dagger \ket{0} = \delta_{q,0}, \]
so the double sum becomes
\begin{align*}
    \frac{d+1}{d^2} \sum_{k = 1}^{d-2} \sum_{\substack{q,p \in \Zd \\ (q,p) \neq (0,0)}} \delta_{q,0}\delta_{kq,0} \tr({D}_{kq,kp}{D}_{q,p}U^\dagger\rho U) &= \frac{d + 1}{d^2} \sum_{k =1}^{d-2} \sum_{p = 1}^{d - 1} \tr(D_{0,kp}D_{0,p} U^\dagger \rho U) \\
    &= \frac{d + 1}{d^2} \sum_{k = 1}^{d-2} \sum_{p =1}^{d-1} \tr(Z^{(k+1)p} U^\dagger \rho U) \\
    &= \frac{(d+1)(d-2)}{d^2} \tr\left(\sum_{p = 1}^{d-1} Z^p U^\dagger \rho U\right) \\
    &= \frac{(d+1)(d-2)}{d^2} \tr\left((d\ket{0} \bra{0} - I) U^\dagger \rho U \right).
\end{align*}
Hence,
\begin{align*}
    \Var[\hat{o}]\Big|_{O = U\ket{0}\bra{0}U^\dagger} &= 1 - \frac{1}{d^2} + \frac{(d+1)(d-2)}{d^2}(d\bra{0}U^\dagger \rho U \ket{0} - 1) - (\bra{0}U^\dagger \rho U \ket{0} - 1/d)^2 \\
    &= \frac{1}{d} + (d-1)\bra{0}U^\dagger \rho U \ket{0} - \bra{0}U^\dagger \rho U\ket{0}^2,
\end{align*}
which be $\Omega(d)$ for certain states $\rho$. 

\subsubsection{Displacement operators}

Now, we calculate the variance for estimating the expectation values of displacement operators $D_{a,b}$. For $(a,b) \neq (0,0)$, i.e., $D_{a,b} \neq I$, $D_{a,b}$ is traceless, so choosing $\widetilde{D}_{q,p} = D_{q,p}$ in Theorem~\ref{thm: shadows variance} gives
\begin{align*}
    \Var[\hat{o}]\Big|_{O = D_{a,b}} &= \frac{d + 1}{d}\|D_{a,b}\|_{\text{H-S}}^2 + \frac{d + 1}{d^2} \sum_{k=1}^{d-2} \sum_{\substack{q,p \in \Zd\\ (q,p) \neq (0,0)}} \tr(D_{q,p}^\dagger D_{a,b}) \tr(D_{kq,kp}^\dagger D_{a,b}^\dagger) \tr(D_{kq,kp}D_{q,p}\rho) - |\tr(D_{a,b}\rho)|^2 \\
    &= \frac{d+ 1}{d}(d)  + \frac{d+1}{d^2} \sum_{k=1}^{d-2} \sum_{\substack{q,p \in \Zd\\ (q,p) \neq (0,0)}}  \delta_{q,a}\delta_{p,b} \delta_{kq,-a}\delta_{kp,-b} \tr(D_{kq,kp}D_{q,p}\rho) - |\tr(D_{a,b}\rho)|^2\\ 
    &= d + 1 + \frac{d + 1}{d^2} \sum_{k = 1}^{d-2} \delta_{ka,-a}\delta_{kb,-b} \tr(D_{ka,ka}D_{a,b}\rho) - |\tr(D_{a,b}\rho)|^2.
\end{align*}
If $a \neq 0$, then $ka = -a$ only for $k = d-1$, which does not appear in the sum over $k$. Otherwise, we must have $b \neq 0$, and similarly no value of $k$ in the sum satisfies $kb = -b$. Hence, for any $D_{a,b} \neq I$,  
\[ \Var[\hat{o}]\Big|_{O = D_{a,b}} = d + 1 - |\tr(D_{a,b}\rho)|^2 \geq d. \]
This implies that if we were to use the generalized Clifford classical shadows to estimate $\tr(D_{a,b}\rho)$, and we chose the number of samples according to the variance, via Eq.~\eqref{Nsample}, the number of copies of $\rho$ would scale with $d$ and $\eps$ as $d/\eps^2$. This is consistent with Theorem~\ref{thm:displacement_lower_bd_1_informal}. Note, however, that Theorem~\ref{thm:displacement_lower_bd_1_informal} concerns measuring $|\tr(D_{a,b}\rho)|$ for all $a,b$, whereas Eq.~\eqref{Nsample} calls for $\mathcal{O}(d/\varepsilon^2)$ measurements even for estimating the expectation value of just one displacement operator $D_{a,b}$.

\end{document}